\documentclass[11pt,a4paper,oneside]{article}
\usepackage[top=3cm, bottom=3cm, left=2cm, right=2cm]{geometry} 
\usepackage[english]{babel}
\usepackage[utf8]{inputenc}
\usepackage[T1]{fontenc}

\usepackage{amsthm,amsmath,amssymb,mathrsfs,dsfont,mathtools}
\usepackage{booktabs, subcaption}
\usepackage{bbm}
\usepackage{bm}
\usepackage{graphicx}
\usepackage[colorlinks=true, allcolors=blue]{hyperref}
\usepackage{comment}
\usepackage{microtype} 
\usepackage[all]{nowidow} 

\usepackage{caption}
\usepackage[authoryear, round]{natbib}
\newcommand{\iid}{\stackrel{\mbox{\scriptsize iid}}{\sim}}
\newcommand{\ind}{\stackrel{\mbox{\scriptsize ind}}{\sim}}
\newcommand{\indicator}{\ensuremath{\mathbbm{1}}}

\newcommand{\calB}{\mathcal{B}}

\newcommand{\calL}{\mathcal{L}}
\newcommand{\calM}{\mathcal{M}}

\newcommand{\Ecr}{\mathscr{E}}

\newcommand{\dd}{\mathrm d}

\newcommand{\Pp}{\mathds{P}}
\newcommand{\X}{\mathds{X}}

\newcommand{\E}{\mathds{E}}

\newcommand{\N}{\mathds{N}}

\newcommand{\CP}{\textsc{cp}}

\newcommand{\CRV}{\textsc{crv}}

\renewcommand{\mid}{\ensuremath{\,|\,}}

\usepackage[normalem]{ulem}

\newcommand{\TRsubs}[2]{{\color{purple}}{\sim}}

\usepackage{authblk}
\date{}
\title{Bayesian nonparametric modeling of multivariate count data with an unknown number of traits}
\author[1]{Lorenzo Ghilotti}
\author[1]{Federico Camerlenghi}
\author[1]{Tommaso Rigon}
\author[2]{Michele Guindani}

\affil[1]{Department of Economics, Management, and Statistics, University of Milano--Bicocca, 20126 Milano, Italy}
\affil[2]{Department of Biostatistics, UCLA Fielding School of Public Health, Los Angeles, CA 90095, US}

\providecommand{\keywords}[1]{
  \small 
  \textbf{\textit{Keywords:}} #1
  \normalsize
}

\usepackage[sf]{titlesec}
\usepackage{sectsty}
\allsectionsfont{\centering \normalfont\scshape} 

\newtheorem{theorem}{Theorem}

\newtheorem{lemma}{Lemma}
\newtheorem{proposition}{Proposition}
\theoremstyle{definition}

\newtheorem{definition}{Definition}
\newtheorem{remark}{Remark}
\theoremstyle{remark}
\newtheorem{example}{Example}

\begin{document}

\maketitle

\begin{abstract}
Feature and trait allocation models are fundamental objects in Bayesian nonparametrics and play a prominent role in several applications. Existing approaches, however, typically assume full exchangeability of the data, which may be restrictive in settings characterized by heterogeneous but related groups. In this paper, we introduce a general and tractable class of Bayesian nonparametric priors for partially exchangeable trait allocation models, relying on completely random vectors. We provide a comprehensive theoretical analysis, including closed-form expressions for marginal and posterior distributions, and illustrate the tractability of our framework in the cases of binary and Poisson-distributed traits. A distinctive aspect of our approach is that the number of traits is a random quantity, thereby allowing us to model and estimate unobserved traits. Building on these results, we also develop a novel mixture model that infers the group partition structure from the data, effectively clustering trait allocations. This extension generalizes Bayesian nonparametric latent class models and avoids the systematic overclustering that arises when the number of traits is assumed to be fixed. We demonstrate the practical usefulness of our methodology through an application to the `Ndrangheta criminal network from the \emph{Operazione Infinito} investigation, where our model provides insights into the organization of illicit activities. 
\end{abstract}


\section{Introduction}\label{sec:introduction}


Feature allocation models gained increasing popularity within the Bayesian nonparametric community since the definition of the Indian buffet process (\textsc{ibp}) by \cite{Gri06, Gri11}. In these models, each subject (or observation) is characterized by the presence of a set of latent characteristics, referred to as \emph{features}. A key development was due to \cite{Thi07}, who determined the de Finetti measure of the \textsc{ibp}. 
These seminal contributions paved the way for a rigorous theoretical foundation of exchangeable feature allocations \citep{Bro(13)}. Relevant extensions are the three-parameter \textsc{ibp} \citep{Teh(09)}, scaled processes \citep{Mas22, Camerlenghi2024}, and product-form feature models defined by \citet{Bat18} and studied by \citet{Ghilotti2025}. See also \citet{beraha25} for a perspective based on point processes theory and predictive characterizations for feature models. Trait allocation models \citep{Jam(17), Camp(18)} are an important generalization of feature allocation models, where each subject's characteristics, referred to as \textit{traits}, are associated with non-negative measurements rather than binary features.  Notable examples of trait allocation models are described in \citet{Bro15, Heaukulani2016, Beraha2023}.  Applications of feature and trait allocation models encompass several fields, including genomics \citep{Beraha2023, Mas22}, Bayesian factor analysis \citep{Aye21}, ecology \citep{Stolf2025, Ghilotti2025}, gene expression modeling \citep{Kno11}, document classification and topic modeling \citep{zhou2016priors,Wil10}, image segmentation and object recognition~\citep{Gri11, Bro15}.

Despite their popularity, most existing contributions on both feature and trait allocation models assume exchangeability among subjects, which implies a form of homogeneity across observations. In many applied areas, this assumption is restrictive, as data come from multiple related studies and are intrinsically heterogeneous. In such cases, the notion of partial exchangeability is more appropriate.  Under this assumption, observations are divided into groups that are exchangeable within, but not across, the different groups. Partial exchangeability has a long history to handle data from dependent populations of species: refer e.g. to \cite{Nip(14), Gri(17), Cam(19)AoS, Colombi2025, franzolini2025} for specific structures, and to  \cite{Quintana2022} for a complete review. Nonetheless, only a few and recent contributions are available for trait and feature allocation models. For instance, \citet{Mas(18), Beraha2023, james2024HIBP} introduced hierarchical constructions to induce dependence among groups, while \cite{shen2024} proposed a bivariate beta process. 

This paper introduces a novel general and tractable class of Bayesian nonparametric priors suitable for modeling partially exchangeable trait allocations. Our proposal relies on completely random vectors (\textsc{crv}s), see, \cite{Cat(21)AoS}. Also, see \citet{Kallenberg_2017}. The proposed theoretical framework is very broad and therefore we focus on a notable subclass, i.e., finite completely random vectors (\textsc{fcrv}s), which prescribe that the number of traits in the population is finite but random. We provide a  comprehensive theoretical analysis of partially exchangeable trait allocations under \textsc{crv}s and \textsc{fcrv}s, including: (i) the marginal distribution of a sample and (ii) posterior representations. To illustrate the applicability of our framework, we discuss two examples involving binary traits (i.e. features) and Poisson counts. In these key special cases, the analytical derivations are extremely tractable.

To illustrate the practical relevance of our methodological framework, we analyze the criminal network dataset of the `Ndrangheta, previously examined in \citet{esbm_rigon, Lu2025}. The data were collected during \emph{Operazione Infinito} \citep{Calderoni17}, a large-scale law enforcement initiative aimed at dismantling the core branch of the `Ndrangheta Mafia in the Milan area. The dataset records multivariate binary outcomes describing the attendance of known affiliates (subjects) at a series of meetings (binary traits, i.e., features). In addition, affiliates can be grouped according to their \emph{locali} affiliation, as documented in juridical records. This partition naturally suggests a partially exchangeable (or \emph{known-groups}) framework, where inference can be carried out in closed form by leveraging our theoretical results. A distinctive feature of the proposed trait allocation model, in contrast with classical approaches to multivariate count data, is that the number of traits, that is, the columns of the data matrix, is itself a random variable. In our motivating application, this means that some meetings may remain unobserved because they were not detected by law enforcement. Our methodology explicitly accounts for this possibility while also enabling the estimation of the number of unseen traits.

Unfortunately, the \emph{locali} partition, while highly informative, is insufficient to fully capture the complexity of the relationships among subjects \citep{esbm_rigon, Lu2025}. To address this, we extend the partially exchangeable model by allowing the partition structure itself to be inferred from the data, thereby leading to a clustering problem. In this setting, the groups are not predetermined but are learned from the data; for this reason, we refer to it as the \emph{unknown-groups} framework. The second main contribution of this paper is thus a novel mixture model designed to cluster trait allocations. Inference is enabled by the distributional results established for the simpler known-groups setting, leading to an efficient Gibbs sampling strategy. Importantly, by treating the partition as a random quantity, the proposed model generalizes Bayesian nonparametric latent class models, such as \citet{Dunson2009}. A key distinction, however, is that the total number of traits is unknown and must be estimated from the data. We formally demonstrate that ignoring this aspect inevitably results in overclustering.

The paper is organized as follows. Section~\ref{sec:methodology} reviews the classical framework of exchangeable trait allocation models, which serves as the foundation for our approach to modeling multivariate count data. Section~\ref{sec:partially_ex_setting} introduces the partially exchangeable setting for trait allocation models and provides a complete Bayesian analysis, including closed-form expressions for the marginal and posterior distributions. Section~\ref{sec:learning_clustering} develops our more elaborate methodology for clustering count data with potentially unobserved traits. We also examine a na\"ive specification that disregards the unseen traits, and we demonstrate, both theoretically and empirically, that such an approach can bias the resulting analyses. Section~\ref{sec:overview_simulation_study} discusses simulation studies assessing the performance of our methodology, while Section~\ref{sec:application} applies it to the criminal network dataset from the \emph{Operazione Infinito} investigation. The paper ends with a discussion; proofs, additional theorems and simulation studies are collected in the Supplementary Material.

\section{Background on exchangeable trait allocation models}
\label{sec:methodology}

We begin by reviewing exchangeable trait allocation models \citep{Camp(18)}, in which data for each of the $n$ subjects are conditionally independent and identically distributed (iid) draws from a common and simple distribution. These highly tractable models serve as building blocks for the more flexible approaches introduced later in Sections~\ref{sec:partially_ex_setting}-\ref{sec:learning_clustering}.

\begin{figure}[tbp]
\centering
\includegraphics[width=0.8\textwidth]{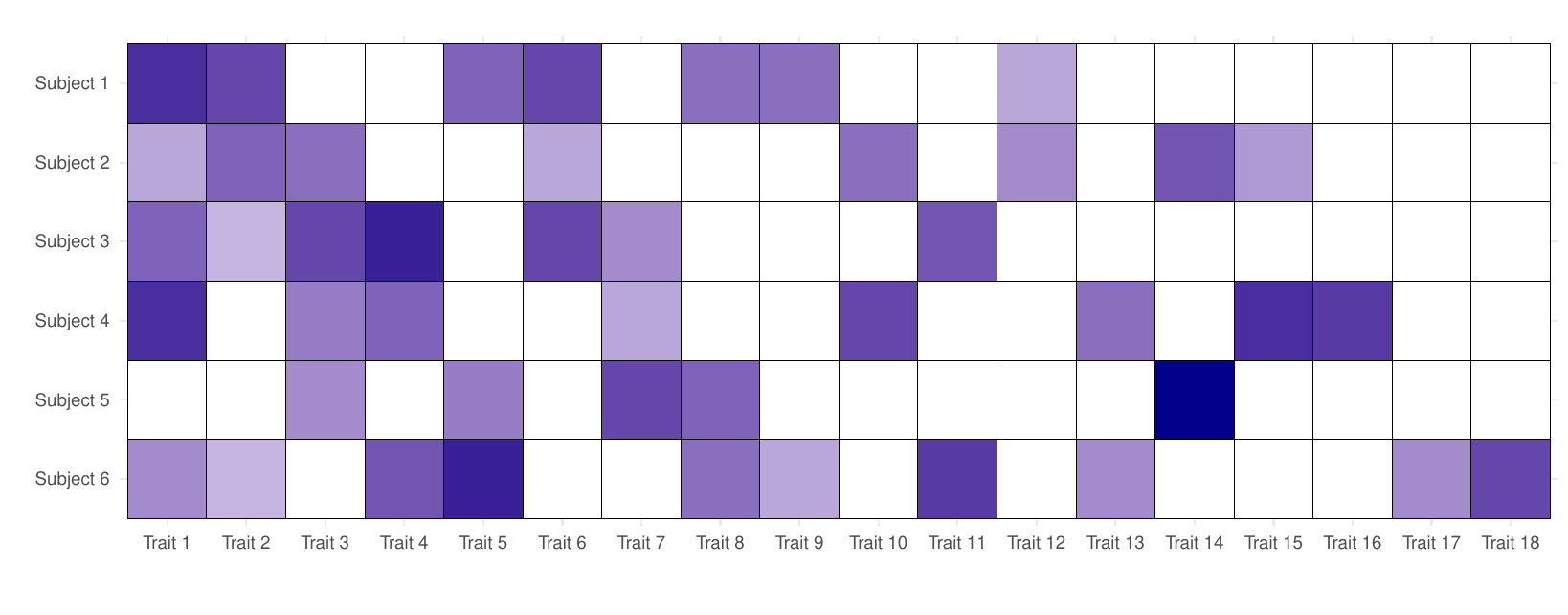}
\caption{Observed data from an exchangeable trait model: matrix of counts $\bm{A}$, with $n = 6$ subjects and $K_n = 18$ observed traits. Rows and columns are arranged in no particular order. White cells, such as $A_{13} = 0$, indicate the absence of a trait for a given subject, while darker shades of blue represent higher values of the corresponding counts $A_{i\ell} \in \{1, 2,\dots\}$. \label{fig:traits}}
\end{figure}

Exchangeable trait allocation models  describe how a set of traits is distributed across a sample of $n$ subjects, where the presence of each trait in a subject is associated with a quantitative measurement, typically an integer, that reflects the expression or abundance of that trait in the subject. More formally, suppose we observe $n$ subjects and $K_n = k$ distinct traits. The data can be represented by an $n \times k$ matrix $\bm{A}$, where each entry $A_{i\ell} \in \{0, 1, 2, \ldots\}$ denotes the count of the $\ell$th trait (column) for the $i$th subject (row), as illustrated in Figure~\ref{fig:traits}. We say that trait $\ell$ is \emph{absent} (or not observed) in subject $i$ if and only if $A_{i\ell} = 0$. Note that each column of $\bm{A}$ contains at least one non-zero entry. Each trait is associated with a distinct \emph{label}, denoted $X_\ell \in \X$, which serves as a placeholder and is not explicitly modeled; it simply identifies the column. Here $\X$ denotes the space of the labels, say $\X = (0, 1)$ for simplicity. Importantly, a given trait $X_\ell$ may be shared by multiple subjects.  A defining feature of trait allocation models, unlike classical multivariate count data models, is that the number of traits (columns) is itself random. In other words, some traits may remain \emph{unseen}.  To model this explicitly, let $(\tilde{X}_j)_{j \geq 1}$ denote the sequence of all possible trait labels and let $\tilde{A}_{ij} \in \{0, 1, 2, \dots\}$ represent the abundance of trait $\tilde{X}_j$ in subject $i$.  In a sample of size $n$, a trait $\tilde{X}_j$ is observed only if $\tilde{A}_{ij} > 0$ for at least one subject. In other words, the observed traits $X_1, \dots, X_{K_n}$ form a subsample of the latent traits $(\tilde{X}_j)_{j \geq 1}$, and the corresponding observed counts satisfy $\sum_{i=1}^n A_{i\ell} > 0$; otherwise, the trait would not be observed. For mathematical convenience, we can organize the pairs $((\tilde{A}_{ij}, \tilde{X}_j))_{j \ge 1}$ by means of subject-specific counting measures $(Z_i)_{ i \ge 1}$ on $\X$, namely
\begin{equation}
    \label{eq:Zi}
    Z_i(\cdot) = \sum_{j\geq 1} \tilde{A}_{ij} \delta_{\tilde{X}_j}(\cdot),
\end{equation}
where $\delta_x$ denotes the Dirac delta mass at the point $x \in \X$. Common assumption in trait allocation models requires that each subject may exhibit only a finite number of traits, ensuring that the total number of distinct traits $K_n$ is almost surely finite in any given sample.

In the exchangeable case, the random variables $\tilde{A}_{ij}$, given a sequence of parameters $(\theta_j)_{j \geq 1}$, are conditionally independent and identically distributed (iid) across subjects (rows) for any fixed $j$, that is
\begin{equation}
\label{eq:cond_prob_counts}
\tilde{A}_{ij} \mid \theta_j \iid P(\cdot\,; \theta_j), \qquad i \geq 1,
\end{equation}
and they are also conditionally independent across traits (columns) for $j \geq 1$. Here, $P(\cdot\,; \theta)$ denotes any parametric distribution supported on the non-negative integers, such as a Poisson distribution, depending on a positive parameter $\theta > 0$. The parameters $(\theta_j)_{j \ge 1}$ can be organized in a discrete measure $\tilde{\mu}$ on $\X$, defined as
\begin{equation}
\label{eq:mui_def}
    \tilde{\mu}(\cdot) = \sum_{j \geq 1} \theta_j \delta_{\tilde{X}_j}(\cdot).
\end{equation}
Note that the atoms of the discrete measure $\tilde{\mu}$, i.e., the trait labels $\tilde{X}_j$, are common across all subjects, so that the same traits are allowed to be observed in multiple subjects. Summarizing, the full Bayesian specification is
\begin{equation*}
\begin{aligned}
Z_i\mid \tilde{\mu} &\iid \CP (\tilde{\mu}), \qquad i \ge 1,\\
\tilde{\mu} &\sim \mathcal{Q},
\end{aligned}
\end{equation*}
which means that $Z_i$ in \eqref{eq:Zi} are iid from a \emph{process of counts} (\textsc{cp}) with parameter $\tilde{\mu}$ defined by \eqref{eq:cond_prob_counts}-\eqref{eq:mui_def}. Here $\mathcal{Q}$ denotes the de Finetti measure, i.e., the prior distribution of the random measure $\tilde{\mu}$. 

The exchangeable setting, especially when traits are binary, has been thoroughly investigated, e.g., in \citet{Jam(17), Camp(18)}, and recently by \citet{Ghilotti2025, beraha25} in the binary case. 
In the following, we present two relevant examples for the  distribution $P(\cdot; \theta)$. 

\begin{example}[Exchangeable binary traits]\label{example:ex_bin_traits}
In our motivating application involving meetings of criminals, we track the attendance of `Ndrangheta affiliates (subjects) at various meetings (traits).  Thus, the $K_n = k$ observed traits correspond to distinct meetings where at least one of the $n$ affiliates has been identified by investigators. Clearly, it is likely that a few meetings have not been spotted therefore it is reasonable to model the total number of meetings as a random variable. In this case study, the count measurement $\tilde{A}_{ij}$ are binary, reducing the trait allocation framework to the special case of feature allocation models \citep{Bro(13)}. Thus, we let $\tilde{A}_{ij} \mid \theta_j \iid \text{Bernoulli}(\theta_j)$ for $i \ge 1$ and any fixed $j$ with success probabilities $\theta_j \in (0, 1)$, so that  $Z_i\mid \tilde{\mu} \iid \textsc{cp}(\tilde{\mu})$ are iid draws from a Bernoulli process \citep{Gri11}. A broad class of tractable prior distributions for $\tilde{\mu}$ are described in \citet{Ghilotti2025} to which we refer for further details and applications to ecology.   
\end{example}

\begin{example}[Exchangeable Poisson counts]
When the data $\tilde{A}_{ij}$ take values in $\{0, 1, 2, \dots\}$, as those depicted in Figure~\ref{fig:traits}, a natural choice is $P(a; \theta) = (a!)^{-1}{\theta}^a e^{-\theta}$ for $a \in \{0, 1,\dots\}$, that is a Poisson distribution with mean $\theta > 0$. In other words, we assume $\tilde{A}_{ij} \mid \theta_j \iid \text{Poisson}(\theta_j)$ for $i \ge 1$ and any fixed~$j$. This specification has been less explored compared to the binary case. 
\end{example}

\begin{remark}\label{rmk:single_parameter_count}
Throughout the paper, we assume that the parametric distribution $P(\cdot\,; \theta)$ is governed by a single positive parameter $\theta > 0$. In the binary case this assumption is not restrictive. However, it may be limiting when modeling count data or categorical data. For instance, more flexible alternatives such as zero-inflated Poisson or multinomial distributions may be preferable. In principle, a theoretical analysis allowing for a general parameter space is possible, but it would require moving beyond completely random vectors---the main technical tool discussed in Section~\ref{sec:partially_ex_setting}---and instead rely on general Poisson processes. To streamline the presentation, we therefore focus on this subclass of models, particularly in light of our motivating application involving binary data. Nonetheless, in Section \ref{app:ext_genereral_space} of the Supplementary Material, we discuss that all theoretical results remain valid with only minor modifications, and we provide an example.
\end{remark}

\section{Partially exchangeable finite trait allocation models} \label{sec:partially_ex_setting}

\subsection{Model specification and completely random vectors}
\label{sec:model_form}

In this section, we introduce a novel class of trait allocation models that relaxes the assumption of exchangeability. We consider a framework in which subjects are partitioned into subpopulations: subjects from different groups are conditionally independent but not identically distributed, while exchangeability holds within each group. This structure, known as partial exchangeability, is well-suited to the `Ndrangheta network data, where affiliates can be naturally grouped according to their membership in specific \emph{locali}. Notably, this extension preserves full analytical tractability, as the posterior distribution of remains available in closed form.

\begin{figure}[tb]
\centering
\includegraphics[width=0.8\textwidth]{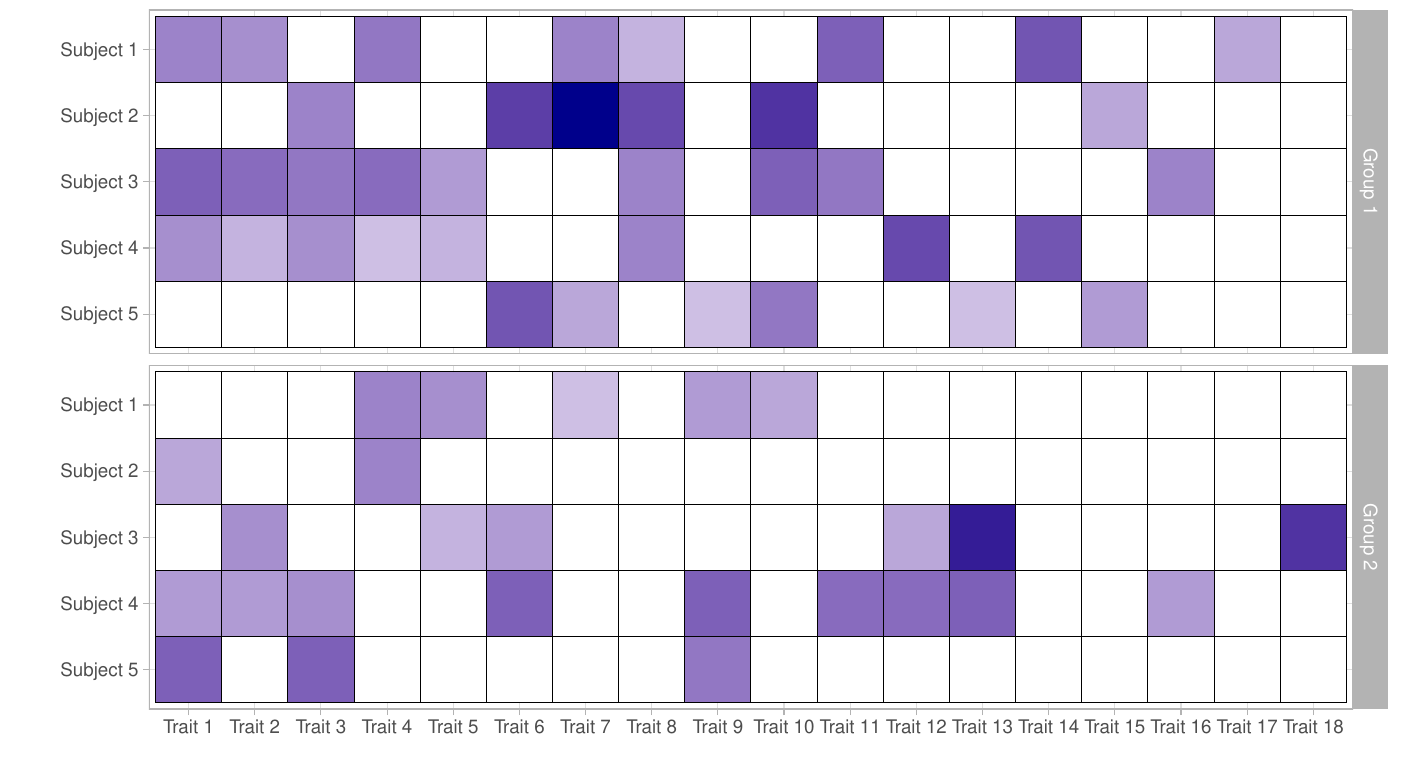}
\caption{Observed data from a partially exchangeable trait model ($d = 2$): two matrices of counts $\bm{A}_1$ and $\bm{A}_2$, each with $n_1 = n_2= 5$ subjects
and $K_n = 18$ observed traits. As in Figure~\ref{fig:traits}, white cells, such
as $A_{221} = 0$, indicate the absence of a trait for a given subject, while darker shades of blue represent
higher values of the corresponding counts $A_{i\ell q} \in \{ 1, 2, \dots\}$.\label{fig:traits_partial}}
\end{figure}

Let $d$ be the number of subpopulations, and suppose we observe a sample of size $n$, with $n_q$ subjects from group $q$, for $q = 1, \ldots, d$, so that $\sum_{q=1}^d n_q = n$. Let $K_n = k$ denote the total number of traits observed across all subjects and groups. The data can be represented by a collection of matrices $\bm{A}_q$, each of dimension $n_q \times k$, where the entry $A_{i\ell q} \in \{0, 1, 2, \ldots\}$ denotes the count of the $\ell$th trait for the $i$th subject in group $q$, as illustrated in Figure~\ref{fig:traits_partial}. Note that a trait may be unobserved within the sample of a specific group. As before, let $(\tilde{X}_j)_{j \geq 1}$ denote the sequence of all possible trait labels and let $\tilde{A}_{ijq} \in \{0, 1, 2, \dots\}$ be the abundance of trait $\tilde{X}_j$ for subject $i$ in group $q$. In a sample of size $n$, a trait $\tilde{X}_j$ is observed only if $\tilde{A}_{ijq} > 0$ for at least one subject belonging to any group. 

We organize these quantities into counting measures  $Z_{iq}(\cdot) = \sum_{j \ge 1} \tilde{A}_{ijq} \, \delta_{\tilde{X}_j}(\cdot)$ for each subject $i$ in subpopulation $q$, with $i \geq 1$ and $q = 1, \ldots, d$. In the partially exchangeable case, the random variables $\tilde{A}_{ijq}$, given the sequences of parameters $(\theta_{j1})_{j \geq 1}, \dots, (\theta_{jd})_{j \geq 1}$, are conditionally iid across subjects belonging to the same group and for a given trait $j$ and group $q$, that is
\begin{equation*}
\tilde{A}_{ijq} \mid \theta_{jq} \iid P(\cdot\,; \theta_{jq}), \qquad i \geq 1,
\end{equation*}
and they are also conditionally independent across traits for $j \geq 1$ and subpopulations $q = 1,\dots,d$. Thus, the main difference compared to the exchangeable case is that the random variables $\tilde{A}_{ijq}$ have different parameters when they refer to subjects belonging to different subpopulations. Moreover,  the parameters $(\theta_{jq})_{j \ge 1}$ can be organized in a group-specific discrete measure $\tilde{\mu}_q(\cdot) = \sum_{j \geq 1} \theta_{jq} \delta_{\tilde{X}_j}(\cdot)$ for $q=1,\dots,d$. Summarising, the full Bayesian specification for partially exchangeable data is
\begin{equation}    \label{eq:partially_ex_traits}
\begin{aligned}
    Z_{iq} \mid \tilde{\mu}_q &\ind \CP (\tilde{\mu}_q), \qquad i\geq 1, \quad q=1,\ldots,d,\\
    (\tilde{\mu}_1,\ldots,\tilde{\mu}_d) &\sim \mathcal{Q}_d,
\end{aligned}
\end{equation}
where $\mathcal{Q}_d$ denotes the de Finetti measure of the vector of random measures $\tilde{\bm{\mu}} = (\tilde{\mu}_1,\ldots,\tilde{\mu}_d)$, that is, the \emph{a priori} distribution. As a clarification, note that for any fixed group $q$ the measures $Z_{iq}$ are conditionally iid draws for $i\geq 1$. Moreover, the exchangeable case discussed in Section~\ref{sec:methodology} is recovered when $d = 1$. 

Introducing some form of dependence among $\tilde{\mu}_1, \dots, \tilde{\mu}_d$ is desirable, as it enables borrowing of information across subpopulations. The challenging task is to elicit a prior $\mathcal{Q}_d$ that induces dependence among the measures while remaining analytically tractable. In pursuing this, we note that inferential goals might relate to: (i) the estimation of trait- and group-specific parameters $\theta_{jq}$; (ii) the estimation of the number of unseen traits in the observed sample. To enable the latter goal, we assume that the total number of traits in the population, denoted with $N$, is finite and random. This implies there will be a finite collection of random variables $\tilde{A}_{i1q}, \dots, \tilde{A}_{iNq}$ for each subject and group, with associated parameters $\theta_{1q},\dots,\theta_{Nq}$ for $q = 1,\dots,d$. Moreover, we assume $N$ is a Poisson random variable with parameter $\lambda>0$. In other terms, the group-specific measures $\tilde{\mu}_q$ in \eqref{eq:partially_ex_traits} take the form
\begin{equation}
\label{eq:finite_CRV}
    \tilde{\mu}_q(\cdot) = \sum_{j=1}^{ N} \theta_{jq} \delta_{\tilde{X}_j}(\cdot), \qquad N \sim \mathrm{Poisson}(\lambda),
\end{equation}
as $q=1,\ldots,d$. Clearly, if $K_n = k$ traits are observed in the finite sample, the number of unseen traits equals $N-k$. Moreover, we assume the parameter vectors $(\theta_{j1},\dots,\theta_{jd})$ are iid draws from a probability law $H^{(d)}$ defined on $(0, \infty)^d$, namely
\begin{equation}
\label{eq:param_H}
(\theta_{j1},\dots,\theta_{jd}) \iid H^{(d)}, \qquad j = 1,\dots,N.
\end{equation}
In this model specification, information is borrowed across groups in two ways: (i) by assuming a common total number $N$ of traits, and (ii) by inducing dependence between the parameters $\theta_{jq_1}$ and $\theta_{jq_2}$ for $q_1 \neq q_2$ through a joint probability distribution $H^{(d)}$. In practice, subpopulations are often exchangeable, making it natural to consider a factorized measure of the form $H^{(d)}(\cdot\,;\psi) = H(\cdot\,;\psi) \times \cdots \times H(\cdot\,;\psi)$, where $\psi$ is a common parameter endowed with a hyperprior. Under this formulation, the parameters $\theta_{jq}$ are conditionally iid for $j = 1,\dots,N$ and $q = 1,\dots,d$, given $\psi$, according to a probability measure $H(\cdot\,;\psi)$. Throughout the paper, we focus on this special case, as it is the most relevant in practice. General results are provided in Section \ref{app:distribution_theory_CRV} of the Supplementary Material. Finally, we need to specify a prior on the atoms $\tilde{X}_j$. As previously mentioned, since the common atoms $\tilde{X}_j$ serve only to label different traits in this formulation, it is sufficient for them to be almost surely distinct. For example, one may assume $\tilde{X}_j \iid P_0$ for $j = 1,\dots, N$, where $P_0$ is any non-atomic distribution.

In equations~\eqref{eq:finite_CRV}--\eqref{eq:param_H}, we specified a prior distribution $\mathcal{Q}_d$ for the vector of random measures $\tilde{\bm{\mu}}$ in~\eqref{eq:partially_ex_traits} as a finite collection of random variables. This specification is both transparent and constructive, but it does not make explicit the connection with infinite-dimensional trait models \citep[e.g.][]{Jam(17), shen2024}. In Section~\ref{app:CRV} of the Supplementary Material, we prove that $\tilde{\bm{\mu}}$ is a finite completely random vector, which is a special case of the broader class of completely random vectors \citep{Cat(21)AoS}; see also \citet{Kallenberg_2017} for a comprehensive treatment and Section~\ref{app:CRV} of the Supplementary Material for a concise overview. This result crucially relies on the Poisson specification for $N$. More precisely, $\tilde{\bm{\mu}}$ can be interpreted as an \textsc{fcrv} with parameters $H^{(d)}$, $\lambda$, and $P_0$, that is, a \textsc{crv} whose Lévy intensity takes the form $H^{(d)}(\dd \theta_1\, \ldots\, \dd \theta_d) \cdot \lambda P_0(\dd x)$, and we write $\tilde{\bm{\mu}} \sim \textsc{fcrv}(H^{(d)}, \lambda, P_0)$. This connection to the theory of completely random vectors is helpful, since the results presented in Section~\ref{sec:bayesian_analysis} follow as special cases of this general theory. In Section \ref{app:distribution_theory_CRV} of the Supplementary Material we present results for the whole class of \textsc{crv}s characterized by Lévy intensities of the form $\rho_d(\dd \theta_1\, \ldots\, \dd \theta_d) \cdot \lambda P_0(\dd x)$, where $\rho_d$ is a possibly infinite measure. This general framework encompasses both finite- and infinite-dimensional trait allocation models. The main practical challenge is identifying suitable, non-degenerate choices for $\rho_d$. Our proposal focuses on the finite case $\rho_d = H^{(d)}$ with $H^{(d)}$ being a probability distribution, which allows for the estimation of the total number of traits, while \citet{shen2024} provide an alternative construction in the infinite-dimensional setting.


\subsection{Distribution theory and posterior inference}\label{sec:bayesian_analysis}

Here we provide a full Bayesian analysis of the proposed model in the known-groups case. More specifically, we obtain tractable closed-form expressions of the marginal distribution of a sample, the posterior distribution of $\tilde{\bm{\mu}}$, and the predictive distribution of a future observation, under model~\eqref{eq:partially_ex_traits} with the prior $\tilde{\bm{\mu}} \sim \textsc{fcrv}(H^{(d)}, \lambda, P_0)$ as specified in equations~\eqref{eq:finite_CRV}--\eqref{eq:param_H}. As mentioned earlier, we focus on the case where $H^{(d)}(\cdot\,;\psi) = H(\cdot\,;\psi) \times \cdots \times H(\cdot\,; \psi)$ factorizes, for the sake of simplicity. Readers interested in the more general setting are referred to the Supplementary Material, in particular Section~\ref{app:distribution_theory_CRV}.

We start by describing the marginal distribution of a sample from model \eqref{eq:partially_ex_traits}, and here we offer a simple and constructive proof. With \emph{marginal distribution of the sample}  $\bm{Z}=(Z_{iq}:  i= 1,\ldots,n_q; q=1,\ldots,d)$, we specifically mean determining the probabilities of the event $(\bm{A} = \bm{a}, K_n = k)$, having denoted by
$\bm A = (A_{i \ell q} : i=1,\dots, n_q; \ell = 1,\dots,k; q = 1,\dots,d)$ the observed counts, where the $K_n=k$ observed traits in the sample are randomly ordered. To this end, suppose for now that the total number of traits $N$ is fixed and let $\tilde{\bm{A}} = (\tilde{A}_{ijq} : i = 1, \dots, n_q; j = 1, \dots, N; q = 1, \dots , d)$ denote the latent counts. Focus on any event $(\tilde{\bm{A}} = \tilde{\bm{a}}, K_n = k)$ that contains the observed event $(\bm{A} = \bm{a}, K_n = k)$.
Let $\mathcal{A} = \{j = 1,\dots,N : \sum_{q=1}^d\sum_{i=1}^{n_q} \tilde{a}_{ijq} > 0\}$ be the indexes of observed traits, therefore $j \in \mathcal{A}$ if the $j$th latent trait is observed and $|\mathcal{A}| = k$. Conditionally on the parameters $\bm{\theta} = (\theta_{jq} : j=1,\dots,N; q = 1,\dots,d)$ and $N$, the likelihood function $\mathscr{L}(\bm{\theta}, N; \tilde{\bm{a}})$ for the event $(\tilde{\bm{A}} = \tilde{\bm{a}}, K_n = k)$ is
\begin{equation}
\label{eq:conditional_lik}
\mathscr{L}(\bm{\theta}, N; \tilde{\bm{a}})= \prod_{q=1}^d  \prod_{j=1}^N\prod_{i=1}^{n_q}P(\tilde{a}_{ijq};\theta_{j q})= \left[\prod_{q=1}^d\prod_{j \not\in \mathcal{A}} P(0;\theta_{j q})^{n_q}\right]\left[ \prod_{q=1}^d \prod_{j\in \mathcal{A}}\prod_{i=1}^{n_q}P(\tilde{a}_{ijq};\theta_{j q})\right].
\end{equation}
In the last term, the first product accounts for the unobserved traits, whereas the second relates to the observed ones.  Note that the quantity $\prod_{q=1}^d P(0; \theta_{jq})^{n_q}$ may be interpreted as the probability of not observing trait $\tilde{X}_j$ among all subjects and groups. Then, we can combine the likelihood~\eqref{eq:conditional_lik} with all the (iid) prior distributions $H(\mathrm{d}\theta;\psi)$, take the integral over $\bm{\theta}$ and sum over all possible sets $\mathcal{A}$, leading to the following marginal probability for the event $(\bm{A} = \bm{a}, K_n = k)$: 
\begin{equation}
\label{eq:conditional_ETFP}
\pi_n (\bm{a}; N, \psi) = \binom{N}{k}\left[\prod_{q=1}^d  \int P (0 ; \theta)^{n_q}
        H (\mathrm{d}\theta; \psi) \right]^{N-k} \prod_{\ell=1}^k \prod_{q=1}^d  \int  \prod_{i=1}^{n_q} P (a_{i\ell q} ; \theta)  H(\mathrm{d}\theta; \psi),
\end{equation}
where the binomial coefficient is introduced to account for all the possible ways we can arrange the $k$ observed traits among all the $N$ traits. However, since the total number of traits $N$ is random and follows a Poisson distribution with mean $\lambda$, the final formula requires marginalizing the conditional law in \eqref{eq:conditional_ETFP} with respect to $N$. By exploiting well-known properties of the Poisson distribution, we obtain a simple closed-form expression, summarized in the theorem below.

\begin{theorem}[Marginal distribution, p\textsc{etpf}]
    \label{thm:marginal_fCRV}
    Let $\bm{Z}$ be a sample from the statistical model~\eqref{eq:partially_ex_traits}, assuming the prior $\tilde{\bm{\mu}} \sim \textsc{fcrv}(H^{(d)}, \lambda, P_0)$ with $H^{(d)}(\cdot\,;\psi) = H(\cdot\,;\psi) \times \cdots \times H(\cdot\,; \psi)$. The probability that  $\bm{Z}$ displays $K_n = k$ distinct traits with counts $\bm A = \bm{a}$ is given by
    \begin{equation}
        \label{eq:pEFPF_fCRV}
        \begin{split}
        \pi_n (\bm{a}; \lambda, \psi) = \frac{\lambda^k}{k!} \exp \left\{ -\lambda \left( 1-\prod_{q=1}^d  \int P (0 ; \theta)^{n_q}
        H (\dd\theta; \psi) \right) \right\} \prod_{\ell=1}^k \prod_{q=1}^d  \int  \prod_{i=1}^{n_q} P (a_{i\ell q} ; \theta)  H(\dd \theta; \psi),
        \end{split}
    \end{equation}
    where $n=\sum_{q=1}^d n_q$ and $\bm{n}=(n_1, \ldots , n_d)$ are the sample sizes.
\end{theorem}

A rigorous and general proof of Theorem~\ref{thm:marginal_fCRV}, not based on the above argument, is provided in the Supplementary Material, Section \ref{app:distribution_theory_CRV}, under general \textsc{crv} priors, covering the case with infinitely many traits. It is interesting to note that the probabilistic quantity in \eqref{eq:pEFPF_fCRV} appears related to the partially Exchangeable Partition Probability Function (p\textsc{eppf}) discussed in \citet{franzolini2025} for partially exchangeable species sampling models; however, here we consider the distribution of a random trait allocation rather than a random partition. By analogy, we refer to equations \eqref{eq:conditional_ETFP} and \eqref{eq:pEFPF_fCRV} as a \emph{partially Exchangeable Trait Probability Function}s (p\textsc{etpf}), with the former being conditional on $N$ and the latter being a Poisson mixture  of \eqref{eq:conditional_ETFP} over $N$.


\begin{example}[Binary traits]\label{ex:binary_marginal}
If the traits are binary and $P(a;\theta) = \theta^a(1-\theta)^{1-a}$, for $a \in \{0, 1\}$, the marginal law of $\bm{Z}$ depends on a sufficient statistic of $\bm a$, corresponding to the feature frequencies in different groups. Define the random frequency of feature $X_\ell$ in group $q$ as $M_{\ell q} := \sum_{i=1}^{n_q} A_{i \ell q}$ with observed values $m_{\ell q}$, stored in matrices $\bm M$ and $\bm{m}$, respectively. Suppose that the prior law $H(\cdot \,;\psi)$ corresponds to a beta distribution with parameters $\psi = (-\alpha, \alpha + \beta)$, with $\alpha < 0$ and $\beta > -\alpha$. Then, the probability that $\bm{Z}$ 
displays $K_n = k$ distinct traits with feature frequencies $\bm M = \bm{m}$ is
   \begin{equation}
        \label{eq:pEFPF_BINOMIAL}
        \begin{split}
        \pi_n (\bm{m}; \lambda, \psi) = \frac{\lambda^k}{k!} \exp \left\{ -\lambda \left[1- \prod_{q=1}^d \frac{(\alpha + \beta)_{n_q}}{(\beta)_{n_q}}  \right]\right\} \prod_{\ell=1}^k \prod_{q=1}^d  \frac{1}{(\beta)_{n_q}}(1-\alpha)_{m_{\ell q}-1}(\alpha + \beta)_{n_q - m_{\ell q}},
        \end{split}
    \end{equation}
where $(a)_n = a(a+1)\cdots(a+n-1)$ is the ascending factorial with $a > 0$. When $d = 1$, this corresponds to the \textsc{efpf} of the Poisson mixture of Beta-Bernoulli  \citep{Ghilotti2025}.
\end{example}

\begin{example}[Poisson counts]\label{ex:count_marginal}
Another remarkable simplification is obtained when the trait counts $\tilde{A}_{ijq}$ take values in $\{0, 1, 2, \dots\}$, and $P(a; \theta) = (a!)^{-1}{\theta}^a e^{-\theta}$ for $a \in \{0, 1,\dots\}$, i.e., a Poisson distribution with mean $\theta > 0$, and $H(\dd\theta;\psi) = \beta^{\alpha}/\Gamma(\alpha)\theta^{\alpha-1}e^{-\beta\theta} \dd \theta$, i.e., a gamma prior with parameters $\psi = (\alpha, \beta)$. In this case, the probability that $\bm{Z}$ 
displays $K_n = k$ distinct traits with counts $\bm A = \bm{a}$ equals
\begin{equation*}
\pi_n (\bm{a}; \lambda, \psi) = \frac{\lambda^k}{k!} \exp \left\{ -\lambda \left[ 1-\prod_{q=1}^d  \left(\frac{\beta}{\beta + n_q}\right)^\alpha \right] \right\} \prod_{\ell=1}^k \prod_{q=1}^d  \frac{ \beta^\alpha(\alpha)_{m_{\ell q}}}{(\beta + n_q)^{\alpha + m_{\ell q}}}\prod_{i=1}^{n_q}\frac{1}{a_{i\ell q}!},
\end{equation*}
where $m_{\ell q} = \sum_{i=1}^{n_q}a_{i\ell q}$.
\end{example}

We now characterize the posterior distribution of $\tilde{\bm{\mu}}$, conditionally to a sample $\bm{Z}$, which is essential to provide Bayesian estimators of the parameters $\theta_{jq}$ and the total number of traits $N$.

\begin{theorem}[Posterior distribution]
    \label{thm:posterior_fCRV}
       Let $\bm{Z}$ be a sample from the statistical model~\eqref{eq:partially_ex_traits}, assuming the prior $\tilde{\bm{\mu}} \sim \textsc{fcrv}(H^{(d)}, \lambda, P_0)$ with $H^{(d)}(\cdot\,;\psi) = H(\cdot\,;\psi) \times \cdots \times H(\cdot\,; \psi)$. If $\bm{Z}$ displays $K_n = k$ distinct traits labeled $X_1, \ldots , X_k$, with associated counts $\bm{a}$, then the posterior distribution of 
       $\tilde{\bm\mu}$ satisfies the distributional equality       \begin{equation}
           \label{eq:posterior_fCRV}
            (\tilde{\mu}_1, \ldots, \tilde{\mu}_d) \mid \bm{Z} \stackrel{d}{=}  (\mu_1^*, \ldots, \mu_d^*) +  (\mu_1', \ldots, \mu_d'),
       \end{equation}
       where $\bm{\mu}^*:= (\mu_1^*, \ldots, \mu_d^*)$ and $\bm{ \mu}':=(\mu_1', \ldots, \mu_d')$ are independent random vectors such that
       \begin{itemize}
           \item[(i)] the components of the vector $\bm{\mu}^*$ are defined as
           $ \mu_q^*(\cdot) = \sum_{\ell=1}^k \theta_{\ell q}^* \delta_{X_\ell}(\cdot)$, for $ q=1, \ldots , d
           $, and the random variables $\theta^*_{\ell q}$ are independent with distribution $H_{\ell q} (\dd\theta;\psi) \propto  \prod_{i = 1}^{n_q} P(a_{i\ell q}; \theta)  H(\dd\theta; \psi)$;
           \item[(ii)] the vector $(\mu_1', \ldots, \mu_d')$ is a $\textsc{fcrv} (H^{\prime(d)}, \lambda',  P_0)$, where $H^{\prime(d)}(\cdot\,;\psi) = H_1^\prime(\cdot\,;\psi) \times \cdots \times H_d^\prime(\cdot\,;\psi)$ and
           \[
           \lambda' = \lambda \prod_{q=1}^d \int P(0; \theta)^{n_q} H(\dd\theta; \psi),
           \quad 
           H'_q(\dd\theta; \psi)  \propto  P(0;  \theta)^{n_q} H(\dd\theta;\psi). 
           \]
           This is equivalent to say that for each $q=1,\ldots,d$
           \begin{equation} \label{eq:mu_pq}
           \mu^\prime_q(\cdot) = \sum_{j=1}^{N^\prime} \theta_{j q}^\prime \delta_{\tilde{X}_j^\prime}(\cdot), \qquad \theta^\prime_{jq} \iid H^\prime_{q}, \qquad \tilde{X}^\prime_j \iid P_0, \qquad j=1,\ldots,N^\prime,          \end{equation}
           where $N^\prime \sim \mathrm{Poisson}(\lambda^\prime)$.
       \end{itemize}
\end{theorem}
See Section \ref{app:distribution_theory_CRV} of the Supplementary Material for a proof.  Despite the complex notation, Theorem~\ref{thm:posterior_fCRV} has actually a very intuitive and clear interpretation. For each subpopulation $q$, the posterior distribution of $\tilde{\mu}_q$ describes both the observed traits out of the initial sample and the unseen traits, and indeed it consists of the sum of two terms in \eqref{eq:posterior_fCRV}. The first term $\mu^*_q$, referring to the observed traits $X_1,\ldots,X_k$, simply describes the posterior distribution of the associated parameters, indicated with variables $\theta^*_{\ell q}$. Indeed, from point (i) of Theorem \ref{thm:posterior_fCRV}, we see that the law of each $\theta^*_{\ell q}$ is $H_{\ell q} (\dd\theta;\psi) \propto  \prod_{i = 1}^{n_q} P(a_{i\ell q}; \theta)  H(\dd\theta; \psi)$, which corresponds to the plain application of Bayes theorem. The second term represents the novel and interesting component, since $\mu^\prime_q$ takes into account potentially unobserved traits in the sample. From point (ii) of the theorem, $N^\prime$ represents the number of unseen traits \emph{a posteriori}, that is $N \overset{d}{=} N' + k$, which is distributed as a Poisson random variable with updated parameter $\lambda^\prime$. Inference on the number of unseen traits can be carried out by inspecting the value of $\lambda^\prime$.
The expressions in Theorem~\ref{thm:posterior_fCRV} substantially simplify for the cases considered in Examples \ref{ex:binary_marginal} and \ref{ex:count_marginal}.

\begin{example}[Binary traits, cont'd]
If the traits are binary and assuming a beta prior law $H(\cdot;\psi)$ with parameters $(-\alpha, \alpha + \beta)$, as in Example~\ref{ex:binary_marginal}, then the posterior distribution of $\tilde{\bm{\mu}}$ given $\bm{Z}$ decomposes as the sum of $\bm{\mu}^*$ and $\bm{\mu}'$, as in the general case. However, the probabilities $\theta^*_{\ell q}$ of re-observing an old feature $X_\ell$ are distributed as $\theta^*_{\ell q} \ind \text{Beta}(m_{\ell q} - \alpha, \alpha + \beta + n_q - m_{\ell q})$, for $\ell=1,\dots,k$, $q=1,\dots,d$.
Furthermore, the random parameters of each $\mu_q^\prime$, which governs unobserved traits in subpopulation $q$, also have a tractable form. In particular, the distribution of the number of unseen features $N^\prime$ is
\begin{equation*}
N^\prime \sim \text{Poisson}(\lambda^\prime), \qquad \lambda^\prime = \lambda \left(1- \prod_{q=1}^d \frac{(\alpha + \beta)_{n_q}}{(\beta)_{n_q}}\right),
\end{equation*} 
whereas the probabilities $\theta_{jq}^\prime$ of observing these traits are distributed as $\theta_{jq}^\prime \iid \text{Beta}(-\alpha, \alpha + \beta + n_q)$.
\end{example}

\begin{example}[Poisson counts, cont'd]
Suppose the traits are Poisson distributed and assume a gamma prior law $H(\cdot;\psi)$ with parameters $(\alpha, \beta)$, as in Example~\ref{ex:count_marginal}. The parameters $\theta^*_{\ell q}$ of the Poisson random variables associated to an observed $X_\ell$ are distributed as $\theta^*_{\ell q} \ind \text{Gamma}(\alpha + m_{\ell q} , \beta + n_q)$, for $\ell=1,\dots,k$, $q=1,\dots,d$, where $m_{\ell q} = \sum_{i=1}^{n_q}a_{i\ell q}$. Furthermore, the distribution of the number of unseen traits $N^\prime$ is
\begin{equation*}
N^\prime \sim \text{Poisson}(\lambda^\prime), \qquad \lambda^\prime = \lambda \left[ 1-\prod_{q=1}^d  \left(\frac{\beta}{\beta + n_q}\right)^\alpha \right],
\end{equation*} 
whereas the parameters $\theta_{jq}^\prime$ of the unseen traits are distributed as $\theta_{jq}^\prime \iid \text{Gamma}(\alpha, \beta + n_q)$.
\end{example}

\subsection{Hyperprior elicitation}\label{sec:hyper_fitting}

When prior information for eliciting the specific values of the parameters $\psi$ and $\lambda$ is not available, a common Bayesian solution is to consider a hyperprior. In the two examples discussed so far, i.e. the binary traits and the Poisson-distributed traits, there is no closed-form characterization for the posterior of $\psi = (\alpha, \beta)$; hence one needs to recut to
Markov chain Monte Carlo (\textsc{mcmc}) sampling. The availability of a closed-form and tractable expression for the marginal law in~\eqref{thm:marginal_fCRV} is crucial in this regard, as it enables the practical implementation of Metropolis–Hastings steps.

As for $\lambda$, let us assume that $\lambda \sim \text{Gamma}(\alpha_\lambda, \beta_\lambda)$. This choice implies a negative binomial distribution for $N$. The chosen hyperprior yields the following posterior distribution $p(\lambda\mid \bm Z) \propto \pi_n(\bm a; \lambda, \psi)p(\lambda; \alpha_\lambda, \beta_\lambda)$, where $\pi_n(\bm a; \lambda, \psi)$ is defined in Theorem \ref{thm:marginal_fCRV}, and $p(\lambda; \alpha_\lambda, \beta_\lambda)$ denotes the density of the gamma distribution with parameters $(\alpha_\lambda, \beta_\lambda)$. It turns out that this hyperprior is conjugate to our model, regardless of the choice of the likelihood $P(\cdot; \theta)$ and the prior $H(\cdot; \psi)$. Indeed, we obtain
\begin{equation}\label{eq:posterior_lambda}
\lambda \mid \bm{Z} \sim \text{Gamma}\left(\alpha_\lambda + k, \beta_\lambda + 1 - \prod_{q=1}^d \int P(0; \theta)^{n_q} H(\dd\theta; \psi)\right).
\end{equation}
The above integral simplifies in specific models, such as Examples~\ref{ex:binary_marginal} and \ref{ex:count_marginal}. As a result, inference on the number of unseen traits, addressed via $N^\prime$ in Theorem~\ref{thm:posterior_fCRV}, is affected.  In particular, under the assumptions of $N^\prime \mid \lambda^\prime \sim \mathrm{Poisson}(\lambda^\prime)$, where $\lambda^\prime$—defined in point (ii) of Theorem~\ref{thm:posterior_fCRV}—follows a gamma distribution, we get a marginal negative binomial distribution for $N^\prime$.

\section{Latent class models with an unknown number of traits}\label{sec:learning_clustering}

\subsection{A mixture model for trait allocations}
\label{sec:mixture_traits}

The grouped modeling framework introduced in Section~\ref{sec:partially_ex_setting} naturally accommodates group-structured data under the assumption of homogeneity within predefined subpopulations. In practice, however, two issues may arise: (i) no prior information on a grouping structure is available, or (ii) the known partition does not adequately capture the underlying organization in the data. In such cases, one may instead seek to learn the group structure directly from the data.
In the `Ndrangheta application, affiliates can be naturally grouped according to their membership in \emph{locali}. Nonetheless, it is important to assess whether this external information truly reflects the structure underlying attendance patterns at meetings. If the clustering inferred from the data diverges from the \emph{locali}, this discrepancy could highlight previously unrecognized collaborations or reveal novel dynamics among members belonging to different \emph{locali}.

To this end, in this section we consider an unknown-groups setting by embedding the model~\eqref{eq:partially_ex_traits}–\eqref{eq:finite_CRV}–\eqref{eq:param_H} within a Bayesian clustering framework, allowing for cluster detection while accounting for unseen traits. In essence, our goal is to estimate a partition of the subjects $Z_1, \ldots, Z_n$ rather than relying on predefined groups. Let $\bm{C} = \{C_1,\dots,C_d\}$ be a partition of the statistical units $\{1,\dots,n\}$, so that $i,i' \in C_q$ if and only if subjects $i$ and $i'$ belong to the same group. We denote by $n_q = |C_q|$ the size of cluster $q$, with $\sum_{q=1}^d n_q = n$, where $d$ is the number of clusters. If the partition structure $\bm{C}$ is known, we recover the known-groups setting of Section~\ref{sec:partially_ex_setting}. When instead $\bm{C}$ is unknown, we can assign a prior distribution. In this way, both the cluster memberships and the number of groups $d$ are random and can be learned from the data. 

When the partition is unknown, the observed data take the form of an $n \times k$ matrix $\bm{A}$ with entries $A_{i\ell}$ and observed labels $X_\ell$, as in the  exchangeable case. Moreover, let $\tilde{A}_{ij} \in \{0, 1,2,\dots\}$ be the abundance of $j$th latent trait $\tilde{X}_j$ in the $i$th subject, as before. Conditionally on a partition $\bm{C} = \{C_1,\dots,C_d\}$, for a fixed trait $j$ and group $q$ we assume, as in the known-groups case, that
\begin{equation*}
\tilde{A}_{ij} \mid \theta_{jq} \iid P(\cdot; \theta_{jq}),\qquad i \in C_q.
\end{equation*}
The random variables are also conditionally independent across traits for $j \ge 1$ and clusters $q=1,\dots,d$. As in equations \eqref{eq:finite_CRV}–\eqref{eq:param_H}, we let the total number of traits $N$ to be random and the distribution of group-specific parameters $\theta_{jq}$ being equal to
\begin{equation*}
N \sim \text{Poisson}(\lambda), \qquad (\theta_{j1}, \dots,\theta_{jd}) \iid H^{(d)}, \qquad j=1,\dots,N.
\end{equation*}
Note that the number of groups $d$ is itself a random variable. The main difference, compared to Section~\ref{sec:partially_ex_setting}, is that we now specify a prior for the partition $\bm{C}$. Among the various options, we consider the prior induced by a Pitman–Yor process \citep{Per(92), Pit(97)}, which stands for its analytical tractability and whose density is
\begin{equation}\label{eq:partition}
\mathds{P}(\bm{C} = \{C_1,\dots,C_d\}) =\frac{\prod_{q=1}^{d-1}(\gamma + q\sigma)}{(\gamma + 1)_{n-1}}\prod_{q=1}^d(1 - \sigma)_{n_q-1},
\end{equation}
where $\gamma > -\sigma$, $\sigma \in [0, 1)$. Expression \eqref{eq:partition}  is referred to as \emph{exchangeable partition probability function} \citep{Pitman1996}. When $\sigma = 0$, the Pitman–Yor process reduces to the partition law of a Dirichlet process \citep{Ferguson1973}. For a general characterization of exchangeable Gibbs-type random partitions, including the Pitman–Yor as a special case, see \citet{DeBlasi2015}. Recent developments on random partitions with finitely many clusters are discussed in \citet{Lijoi2020, Arg(22), Lijoi2024, Colombi2025}. 

We can equivalently express this model in the following hierarchical form
\begin{equation}\label{eq:mixture_model2}
\begin{aligned}
        Z_i \mid \mu_i & \ind \CP (\mu_i), \qquad \mu_i \mid \tilde{p} \iid \tilde{p}, \qquad i \ge 1,\\
         \tilde{p} &\sim \mathcal{P},
\end{aligned}
\end{equation}
where $\tilde{p}$ is a discrete random probability measure, and  $\mathcal{P}$ denotes its prior distribution. More precisely, we assume that
\begin{equation*}
    \tilde{p}(\cdot) = \sum_{h \geq 1} \xi_h \delta_{\eta_h}(\cdot),
\end{equation*}
where $(\xi_h)_{h \geq 1}$ is a sequence of random probability weights summing to $1$, and $(\eta_h)_{h \geq 1}$ is a sequence of discrete random measures of the form $\eta_h = \sum_{j=1}^N\phi_{jh}\delta_{\tilde{X}_j}$; these two sequences are assumed to be independent. The discreteness of $\tilde{p}$ implies that there is a positive probability that $\mu_i$ and $\mu_{i'}$ are be identical, inducing clustering among the observations $Z_1,\dots,Z_n$. We denote with $\tilde{\mu}_1,\dots,\tilde{\mu}_d$ the $d$ distinct random measures among $\mu_1,\dots,\mu_n$. The random weights $\xi_h$ represent the clustering probabilities, that is $\xi_h$ is the probability that $\mu_i$ is drawn from $h$th component of $\tilde{p}$, that is $\mu_i = \eta_h$. The weights $\xi_h$ follow the stick-breaking weights of the Pitman-Yor process so that $\xi_h = V_h \prod_{r=1}^{h-1}(1- V_r)$, with $V_h \ind {\rm Beta}(1- \sigma, \gamma + h\sigma)$, where we agree that $\xi_1 = V_1$.
Moreover, we characterize the law of the sequence of $(\eta_h)_{h\geq1}$ through its finite-dimensional distributions: we assume that any $d$-dimensional subset of $(\eta_h)_{h \geq 1}$, denoted with $\tilde{\bm{\eta}} = (\tilde{\eta}_1,\ldots,\tilde{\eta}_d)$, is distributed as $\tilde{\bm{\eta}} \sim \textsc{fcrv}(H^{(d)}, \lambda, P_0)$. This means, in particular, that the distinct measures are distributed as $\tilde{\bm{\mu}} = (\tilde{\mu}_1,\dots,\tilde{\mu}_d) \sim \textsc{fcrv}(H^{(d)}, \lambda, P_0)$. The representation in \eqref{eq:mixture_model2} helps clarify a crucial aspect of the model, as we now discuss. By employing an exchangeable prior for the latent partition, we are implicitly returning to an exchangeable framework
\begin{equation}\label{eq:mixture_model}
\begin{aligned}
        Z_i \mid \tilde{p} &\iid  \int \textsc{cp} (\mu)\tilde{p} (\dd \mu) \\ 
         \tilde{p} &\sim \mathcal{P},
\end{aligned}
\end{equation}
where the subjects $Z_i$ are conditionally iid draws from a \emph{process of counts mixture}. However, this model is far more flexible than the one considered in Section~\ref{sec:methodology}, where $Z_i \mid \tilde{\mu} \iid \CP(\tilde{\mu})$. Indeed, under model \eqref{eq:mixture_model} we have that $Z_i(\cdot) = \sum_{j=1}^N \tilde{A}_{ij}\delta_{\tilde{X}_j}(\cdot)$ and
\begin{equation*}
\mathds{P}(\tilde{A}_{i1} = \tilde{a}_1,\dots,\tilde{A}_{iN} = \tilde{a}_N\mid \tilde{p}) = \sum_{h \ge 1} \xi_h \prod_{j=1}^NP(\tilde{a}_j; \phi_{jh}), \qquad i \ge 1,
\end{equation*}
where the random vectors $(\tilde{A}_{i1},\dots, \tilde{A}_{iN})$ for $i \ge 1$ are conditionally iid from the above law. The prior for $\tilde{p}$ specifies that $N \sim \text{Poisson}(\lambda)$, $\tilde{X}_j \iid P_0$, the weights $(\xi_h)_{h \ge 1}$ follow the Pitman--Yor stick-breaking distribution, and that the parameters $\phi_{jh}$, under the factorized structure $H^{(d)} = H(\cdot;\psi) \times \cdots \times H(\cdot;\psi)$, are iid draws $\phi_{jh} \iid H(\cdot;\psi)$ for $j = 1,\dots,N$ and $h \ge 1$. 

\begin{remark}
As mentioned in the introduction, the Bayesian model defined in equation~\eqref{eq:mixture_model} is related to a classical statistical approach for modeling multivariate discrete data, known as latent class analysis \citep{Lazarsfeld1968, Goodman1974, Hagenaars2002}. The more recent work of \citet{Dunson2009} provides a Bayesian nonparametric extension of this framework, making it even more closely related to ours. However, all these approaches assume that $N$ is a known constant, meaning that all traits are known in advance, including those equal to zero for all subjects. While this may be a reasonable assumption in many contexts, such as when the number of variables is fixed, it may be untenable in others. Ignoring unseen traits can skew the analysis and directly affect clustering, as we now show.
\end{remark}

\subsection{Effect of accounting for potentially unseen traits}\label{sec:comparison_no_unseen}

To assess the impact of accounting for potentially unseen traits in the sample, we compare the predictive allocation probabilities for a generic subject $i$ under the model defined in~\eqref{eq:mixture_model2}, with the Pitman-Yor prior for $\xi_h$, and a na\"ive model almost identical to \eqref{eq:mixture_model2}, in which it is assumed that there are no unseen traits, i.e. we suppose $N = k$. Specifically, let 
\begin{equation*}
p_{iq} = \mathds{P}(\mu_i = \tilde{\mu}_q \mid \bm{Z}, \bm{\mu}_{-i}), \qquad p_{i,\text{new}} = \mathds{P}(\mu_i = \text{``new''} \mid \bm{Z}, \bm{\mu}_{-i}), \qquad i=1,\dots,n,
\end{equation*}
where $\bm{\mu}_{-i} = (\mu_1,\dots\mu_{i-1},\mu_{i+1},\dots,\mu_n)$ with $d_{-i}$ distinct values $\tilde{\mu}_1,\dots,\tilde{\mu}_{d_{-i}}$, for $q = 1,\dots,d_{-i}$. Thus each $p_{iq}$ denotes the predictive probability that subject $i$ belongs to the $q$th cluster formed by the remaining subjects, and $p_{i,\text{new}}$ is the predictive probability for subject $i$ to form their own cluster. We similarly define $p_{iq}^*$ and $p_{i,\text{new}}^*$ for the na\"ive model in which it is wrongly assumed that all traits are observed with $N = k$. 

\begin{proposition}\label{prop:cluster_comparison} The predictive allocation probabilities for a generic subject $i$ under model~\eqref{eq:mixture_model} and a na\"ive model in which it is assumed that there are no unseen traits, with $N = k$, satisfy the inequality:
\begin{equation*}
 \frac{p_{i,\textup{new}}}{p_{iq}} <  \frac{p_{i,\textup{new}}^*}{p_{iq}^*}, \qquad q=1,\dots,d_{-i}, \quad i=1,\dots,n.
\end{equation*}
\end{proposition}
Refer to Section \ref{proof:cluster_comparison} of the Supplementary Material for the proof.
The inequality in Proposition \ref{prop:cluster_comparison} shows that, once potentially unseen traits are taken into account, the probability of allocating new clusters decreases compared to the naïve model. The practical implications are substantial. In the naïve model, which ignores unseen traits, the likelihood of assigning new subjects to their own clusters is systematically overestimated, leading to a fragmented view of the data and less interpretable clusters. By contrast, our proposed model incorporates the role of unseen traits, which tends to stabilize cluster allocation probabilities. This adjustment reflects a more cautious and principled approach to inference, acknowledging that the possibility of unseen traits carries essential information. 
We emphasize that Proposition~\ref{prop:cluster_comparison} is not limited to the Pitman–Yor specification; rather, it applies to any prior on the set of weights $(\xi_h)_{h \ge 1}$ of the discrete random probability measure $\tilde{p}$ in model \eqref{eq:mixture_model}. In Section \ref{sec:sim_k_fixed} of the Supplementary Material, we numerically illustrate a simple scenario where the discrepancy highlighted in Proposition \ref{prop:cluster_comparison} has a substantial impact on the inference process, demonstrating the importance of accounting for unseen traits in practical applications.

\subsection{Gibbs sampling and update of the clustering structure}\label{sec:clustering_estimation}

Posterior inference for mixture model \eqref{eq:mixture_model} can be efficiently tackled via a simple marginal Markov Chain Monte Carlo (\textsc{mcmc}) algorithm. The procedure is greatly facilitated by the availability of a closed-form expression for the p\textsc{etpf}, which is given in Theorem~\ref{thm:marginal_fCRV}.

The primary goal of the the Gibbs sampling is the approximation of posterior distribution of the latent clustering structure $\bm{C}$. The probability distribution on $(\xi_h)_{h \geq 1}$ determines the \emph{a priori} predictive allocation probabilities of the sampling model, which describe the stochastic generation of $\mu_1,\dots,\mu_n$ as well as their clustering structure $\bm{C}$. Specifically, given the previously observed measures $\mu_1,\dots,\mu_n$, with $d$ distinct values $\tilde{\mu}_1,\dots,\tilde{\mu}_d$ and frequencies $n_1,\dots,n_d$, a Pitman--Yor process prescribes that, a priori, we have
\begin{equation*}
\mathds{P}(\mu_{n+1} = \tilde{\mu}_q \mid \mu_1,\dots,\mu_n) = \frac{n_q - \sigma}{n + \gamma}, \qquad \mathds{P}(\mu_{n+1} = \text{``new''} \mid \mu_1,\dots,\mu_n) = \frac{\gamma + d\sigma}{n + \gamma}.
\end{equation*}
This well-known sequential mechanism induces a partition of the statistical units. In practice, however, we do not explicitly compute the random measures $\mu_1,\dots,\mu_n$. If the focus is on the clustering structure, we only need to keep track of the labels associated with $\mu_1,\dots,\mu_n$ and the implied partition. The posterior distribution of the clustering $\bm{C}$ under the mixture model \eqref{eq:mixture_model} is obtained through a Gibbs sampling procedure, by iteratively updating the labels of $\mu_i \mid \bm Z, \bm{\mu}_{-i}$ for $i=1,\dots,n$ according to the previously defined full conditional allocation probabilities $p_{iq} = \mathds{P}(\mu_i = \tilde{\mu}_q \mid \bm{Z}, \bm{\mu}_{-i})$ and $p_{i,\text{new}} = \mathds{P}(\mu_i = \text{``new''} \mid \bm{Z}, \bm{\mu}_{-i})$, where $\bm{\mu}_{-i} = (\mu_1,\dots\mu_{i-1},\mu_{i+1},\dots,\mu_n)$ comprises $d_{-i}$ distinct values $\tilde{\mu}_1,\dots,\tilde{\mu}_{d_{-i}}$ with frequencies $n_{q,-i}$ for $q=1,\dots,d_{-i}$ and forms a partition $\bm{C}_{-i}$. In addition, the full conditional clustering probabilities $p_{iq}$ and $p_{i, \textup{new}}$ have a simple expression, obtained by combining the a priori Pitman--Yor sequential scheme with the marginal likelihood:
\begin{equation}\label{eq:update_clustering}
p_{iq} \propto \frac{n_{q,-i} - \sigma}{n + \gamma - 1}\,\pi_n(\bm{a}_{iq}; \lambda, \psi), \qquad p_{i,\textup{new}} \propto \frac{\gamma + d_{-i}\sigma}{n+\gamma-1}\,\pi_n(\bm{a}_{i,\textup{new}}; \lambda, \psi),
\end{equation}
where $\bm{a}_{iq}$ denotes the observed trait values organized according to the partition $\bm{C}_{-i}$, with subject $i$ allocated to cluster $q$. Similarly, $\bm{a}_{i,\textup{new}}$ corresponds to the trait values organized according to the partition $\bm{C}_{-i}$ when subject $i$ is assigned to a new cluster. This updating scheme is particularly straightforward thanks to the closed-form expression of $\pi_n(\bm a; \lambda, \psi)$ provided in Theorem~\ref{thm:marginal_fCRV}. In a model with binary traits and beta priors described in Example~\ref{ex:binary_marginal} and a model with Poisson counts and gamma priors described in Example~\ref{ex:count_marginal}, the evaluation of $\pi_n(\bm a; \lambda, \psi)$ is highly efficient, resulting in a fast marginal sampler. 

When hyperpriors are placed on the parameters $\lambda$ and $\psi$, the \textsc{mcmc} scheme described above is coupled with updates of $\lambda$ and $\psi$ from their respective full conditional distributions, i.e. equation~\eqref{eq:posterior_lambda} in the case of $\lambda$. Exact sampling from the full conditional of $\psi$ can also be replaced by a Metropolis-within-Gibbs step for $\psi$. 

\section{Overview of simulation studies} \label{sec:overview_simulation_study}

The simulation studies serve two main purposes. 
The first, detailed in Section~\ref{sec:sim_k_fixed} of the Supplementary Material, provides empirical evidence for the discrepancy highlighted in Proposition~\ref{prop:cluster_comparison}, which can lead to substantial inferential differences. 
To illustrate, we consider a simulated binary-outcome example  comparing our proposed \emph{unknown-groups} model \eqref{eq:mixture_model} with a latent class model fixing $N = k$, thereby excluding the possibility of unseen traits as in \citet{Dunson2009}. In line with Proposition~\ref{prop:cluster_comparison}, our model yields fewer clusters than this na\"ive specification, with a substantial difference in our comparison.
The second goal, discussed in Section~\ref{sec:simulations_network} of the Supplementary Material, is to assess the performance of the \emph{unknown-groups} model introduced in Section~\ref{sec:learning_clustering}. Motivated by criminal network analysis, we simulate data reflecting realistic network structures. In such contexts, the focus is on the weighted adjacency matrix, which records the number of meetings co-attended by pairs of affiliates \citep{esbm_rigon, Lu2025}. This matrix summarizes the raw data, originally collected as multivariate binary observations, where each affiliate is associated with the list of meetings they attended. Our framework induces a probabilistic structure on it.
To evaluate model flexibility, we consider two simulated scenarios, each producing distinct adjacency matrix properties resembling those in \citet{Lu2025}. Beyond inference on the adjacency matrix, a central goal is estimating the number of meetings that went undetected by investigators. We benchmark our model-based estimates against the \emph{negative binomial mixture of \textsc{bb}s} from \citet{Ghilotti2025}, a natural and simple competitor that does not take the group structure into account. Across both scenarios, the proposed model accurately recovers the underlying structure: the posterior mean of the adjacency matrix closely matches the simulated one, and the estimates of undetected meetings are reliable. In contrast, the competitor fails to capture the number of unseen features in the first scenario and the structure of the connectivity patterns in the second scenario. Its poor fit stems from its enforced homogeneity across groups, which yields a posterior expectation of the adjacency matrix assigning identical values to all entries.

\section{Analysis of the Infinito network}
\label{sec:application}

In this section, we analyze the \emph{Infinito} `Ndrangheta network \citep{esbm_rigon, Lu2025}. These data stem from \emph{Operazione Infinito} \citep{Calderoni17}, a six-year-long, large-scale law enforcement operation aimed at monitoring and dismantling the core structure of \emph{La Lombardia}, the highly pervasive branch of the `Ndrangheta Mafia operating in the Milan area.  The raw data collected during the investigation are publicly available at \href{https://sites.google.com/site/ucinetsoftware/datasets/covert-networks}{https://sites.google.com/site/ucinetsoftware/datasets/covert-networks} in the form of multivariate binary outcome observations, making them well-suited for analysis using the binary traits specification of our proposed model.

The attendance of the $n$ affiliates at the $k$ distinct meetings can be represented by the $n \times k$ binary matrix $\bm{A}$, where each entry $A_{i\ell} \in \{0, 1\}$ indicates whether individual $i$ attended meeting $\ell$. A graphical representation of the observed  weighted adjacency matrix $\bm{W} = \bm{A}\bm{A}^T$, whose generic entry $W_{ii'}$ represents the number of meetings jointly attended by affiliates $i$ and $i'$, is shown in Figure~\ref{fig:ndrangheta_obs_mixture}. \cite{esbm_rigon, Lu2025} directly modelled the matrix $\bm{W}$ without considering the information of the binary traits $A_{i\ell}$. Instead, we model $A_{i\ell}$, which will induce a model for $\bm{W}$ as a by-product, therefore making use of all the available information. We focus on the $d = 5$ most populated \emph{locali}, with sizes $n_1 = 16,n_2 = 14, n_3 = 23, n_4 = 15, n_5 = 16$, of the \emph{Infinito} dataset, for a total of $n = 84$ affiliates and $k = 44$ recorded meetings among them. The affiliates' locale membership provides a natural partition into five groups, making it reasonable to begin the analysis by leveraging this a priori known clustering. This motivates the known-groups framework described in Section \ref{sec:partially_ex_setting}. In particular, we consider model \eqref{eq:partially_ex_traits} for binary traits, where $H(\cdot;\psi)$ is the beta distribution with parameters $(-\alpha, \alpha + \beta)$, with $\psi = (\alpha, \beta)$. Prior distributions for the model parameters $\lambda$, $-\alpha$, and $\alpha + \beta$ are specified as follows. The parameter $\lambda$, governing the Poisson-distributed total number of meetings, is assigned a gamma prior with parameters $(\alpha_\lambda, \beta_\lambda)$, as detailed in Section \ref{sec:hyper_fitting}. The hyperparameters $(\alpha_\lambda, \beta_\lambda)$ are chosen so that the prior expected value of $N$ equals $\hat{N} = 1.5k$, where $k = 44$ denotes the observed number of meetings in the dataset, and the prior variance of $N$ is set to $10\hat{N}$. For the parameters $(a, b) = (-\alpha, \alpha + \beta)$ of the beta distribution $H$, independent gamma priors are assumed with hyperparameters $(\alpha_a, \beta_a)$ and $(\alpha_b, \beta_b)$, respectively. Specifically, $(\alpha_a, \beta_a)$ are set to induce a prior mean of $0.2$ for $a$ with a large variance, while $(\alpha_b, \beta_b)$ are chosen so that the prior mean of $b$ is $10$, also with high variance. Posterior inference relies on 10,000 iterations of a vanilla \textsc{mcmc} algorithm, with the first 1,000 samples discarded as burn-in and a thinning interval of 2.

\begin{figure}[tbp]
    \centering

    \begin{subfigure}[t]{0.48\linewidth}
        \centering
        \includegraphics[width=\linewidth]{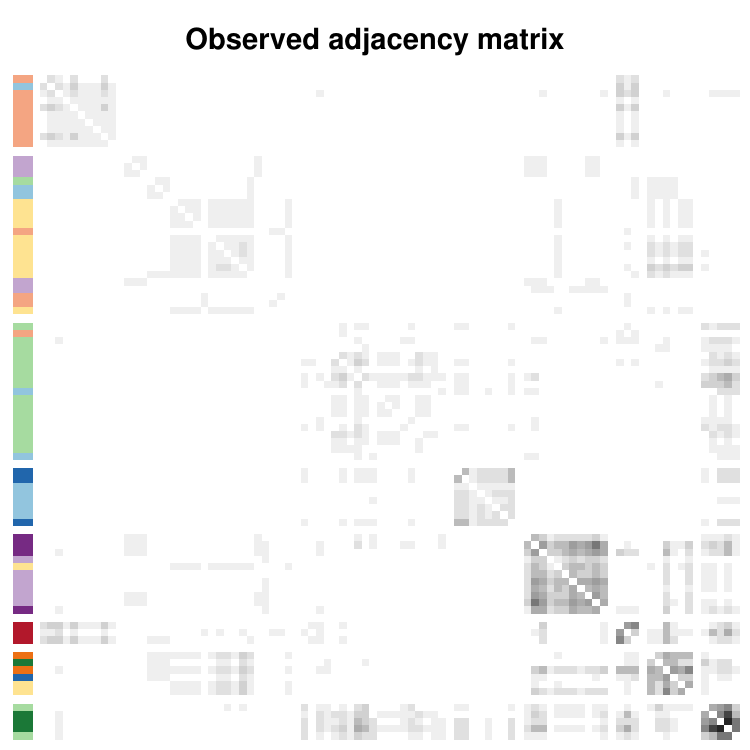}
        \caption{Observed adjacency matrix, ordered and partitioned according to the \textsc{vi} clustering.  
        Side colors correspond to the \emph{locali}, with darker/lighter shades denoting bosses/affiliates.}
        \label{fig:ndrangheta_obs_mixture}
    \end{subfigure}
    \hfill
    \begin{subfigure}[t]{0.48\linewidth}
        \centering
        \includegraphics[width=\linewidth]{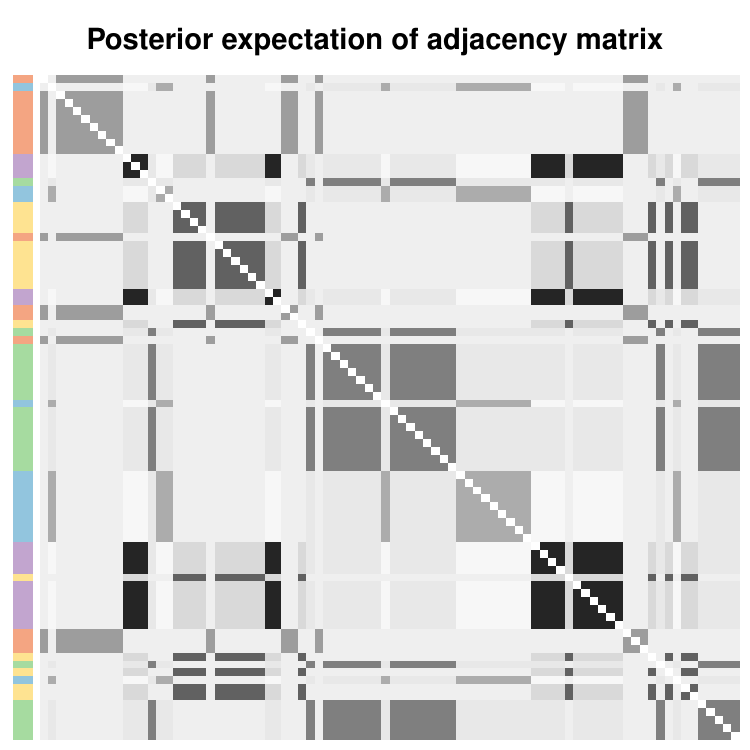}
        \caption{Posterior expectation under a known-groups specification with partition induced by \emph{locali} membership. Ordering matches panels (a).}
        \label{fig:ndrangheta_partial}
    \end{subfigure}

    \vspace{1em}

    \begin{subfigure}[t]{0.48\linewidth}
        \centering
        \includegraphics[width=\linewidth]{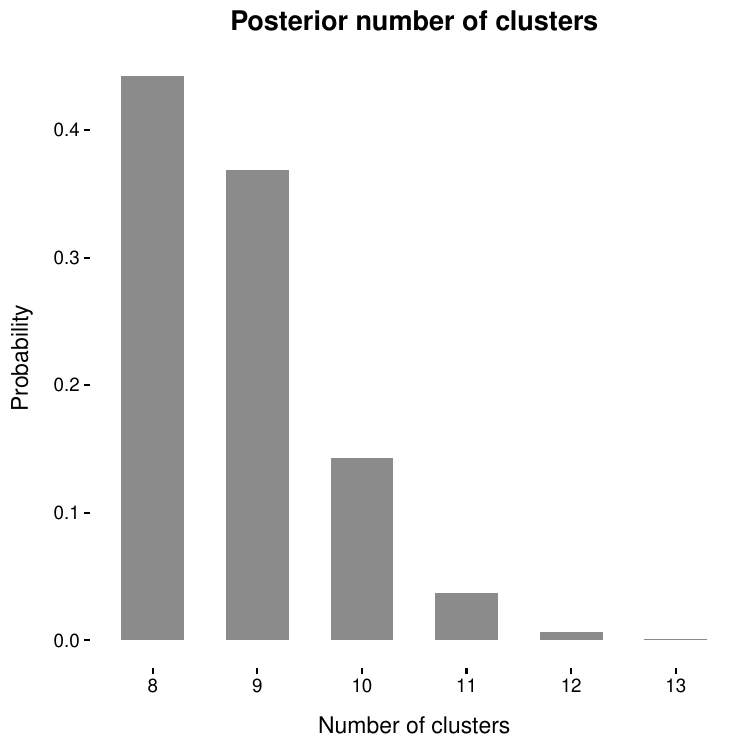}
        \caption{Posterior distribution of the number of clusters among the affiliates.}
        \label{fig:post_number_clusters_ndrangheta}
    \end{subfigure}
    \hfill
    \begin{subfigure}[t]{0.48\linewidth}
        \centering
        \includegraphics[width=\linewidth]{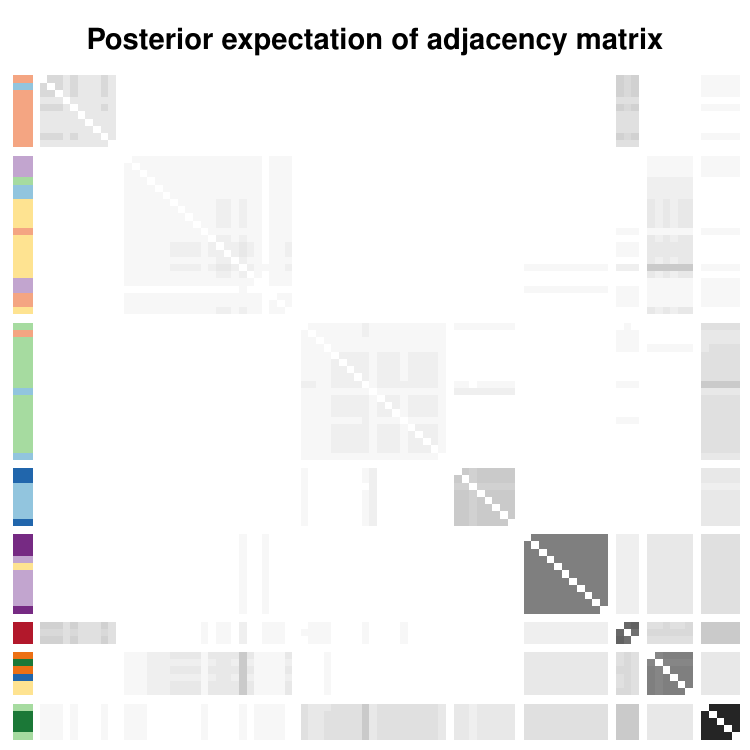}
        \caption{Posterior expectation under an unknown-groups model with groups learned from the data. Ordering matches panel~(a).}
        \label{fig:ndrangheta_mixture}
    \end{subfigure}

    \caption{Comparison of the \emph{Infinito} network with the estimated models.  
    Top row: observed adjacency matrix and the corresponding posterior expectation under the known-groups model.  
    Bottom row: number of clusters among the affiliates and posterior expectation of the adjacency matrix under the unknown-groups model. In the adjacency matrices, the color of each entry ranges from white to black as the number of co-attended meetings increases (observed in panel (a), expected in panels (b) and (d)).}
    \label{fig:ndrangheta_matrices}
\end{figure}

Figure~\ref{fig:ndrangheta_partial} describes the estimated connectivity patterns among different \emph{locali} through the posterior expected weighted adjacency matrix, which we obtained exploiting the closed-form results of Theorem~\ref{thm:posterior_fCRV}. As evident by comparing Figure~\ref{fig:ndrangheta_obs_mixture} with Figure~\ref{fig:ndrangheta_partial}, this estimated adjacency matrix captures the most significant patterns but partially disagree with the observed one, indicating that the partition induced by \emph{locali} membership is insufficient to describe the underlying connectivity patterns within the criminal organization. This points to more complex connections, extending beyond the \emph{locali} structure. Identifying and analyzing these hidden structures could provide insights into the organizational strategies of the criminal network.

All these considerations motivate the need to move to the framework described in Section \ref{sec:learning_clustering}, where the clustering structure among affiliates is learned directly from the data. It is important to note that, in general, the unknown-group case should not be regarded as a ``better'' model than the known-group case. Indeed, if the information provided by the observed groups were consistent with the data, the known-group model would capture the underlying probabilistic structure much more efficiently than the mixture model. In this specific context, however, the \emph{locali} represent an overly coarse description of the criminal organization. This is, in fact, an interesting insight---one that can be appreciated by fitting both models. More precisely, we consider the binary traits specification of model \eqref{eq:mixture_model}, referred to as the unknown-groups model. The prior structure is the same as that used for the known-groups model discussed previously, with the additional specification that the clustering is modeled using a Dirichlet process with concentration parameter equal to $\gamma = 1$. Posterior inference relies on 10,000 iterations of the \textsc{mcmc} algorithm described in Section~\ref{sec:clustering_estimation}, with the first 1,000 samples discarded as burn-in and a thinning interval of 2.

Figure~\ref{fig:post_number_clusters_ndrangheta} displays the posterior distribution of the number of clusters among affiliates. To obtain a point estimate of the clustering structure, we follow \cite{Wad(18)} and select the configuration that minimizes the posterior expectation of the Variation of Information (\textsc{vi}) metric \citep{Mei(03)}. The resulting clustering consists of eight unbalanced clusters, represented by the blocks in the colored bar of Figure \ref{fig:ndrangheta_obs_mixture}. The row colors indicate the affiliates' membership to the five \emph{locali}, with darker shades denoting bosses and lighter tones representing lower-ranked affiliates. As evident from Figure \ref{fig:ndrangheta_matrices}, the clustering structure learned from the data deviates from the one induced by \emph{locali} membership, and the posterior expected weighted adjacency matrix in Figure \ref{fig:ndrangheta_mixture} nicely aligns with the observed one. These considerations qualitatively support estimating the clustering structure from the data, rather than relying on the partition induced by the \emph{locali} membership. To strengthen this argument, we provide a complementary quantitative assessment. Following standard model selection practices, we compare the \textsc{waic} (Watanabe-Akaike Information Criteria) \citep{Watanabe13} between the two models. Specifically, the unknown-groups model yields a \textsc{waic} of $-7236$, whereas the known-groups model results in a \textsc{waic} of $-6573$, strongly indicating a preference for data-driven clustering estimation.

\begin{figure}[tbp]
    \centering
    \begin{subfigure}[t]{0.3\linewidth}
        \centering
        \includegraphics[width=\linewidth]{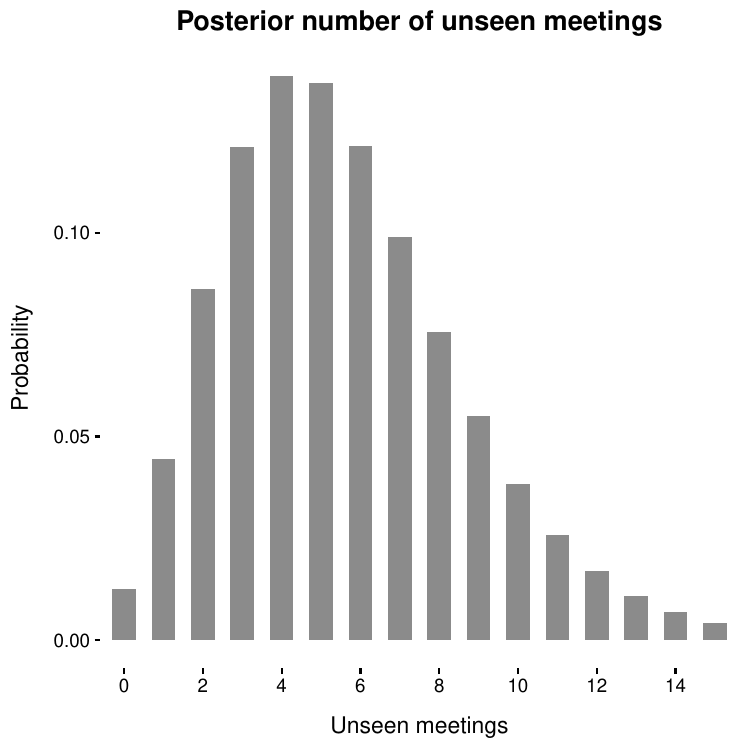}
        \caption{Unseen meetings: partition estimated from the data.}
        \label{fig:unseen_data}
    \end{subfigure}
    \hfill
    \begin{subfigure}[t]{0.3\linewidth}
        \centering
        \includegraphics[width=\linewidth]{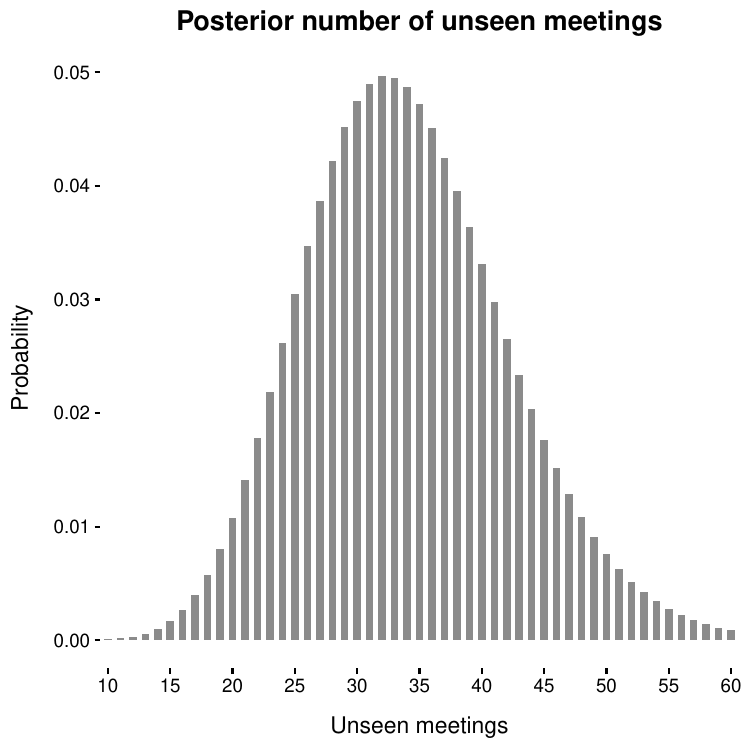}
        \caption{Unseen meetings: partition induced by \emph{locali} membership.}
        \label{fig:unseen_locali}
    \end{subfigure}%
    \hfill
    \begin{subfigure}[t]{0.3\linewidth}
    \centering
        \includegraphics[width=\linewidth]{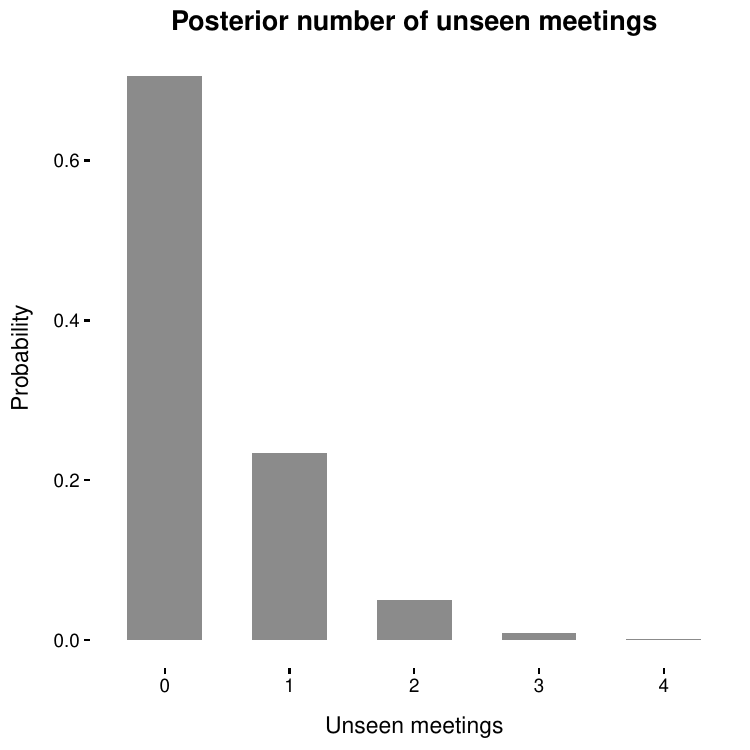}
        \caption{Unseen meetings: negative binomial mixture of \textsc{bb}s.}
        \label{fig:post_unseen_ndrangheta_negbin_bb}
    \end{subfigure}

    \caption{Posterior distributions of the number of unseen meetings in the following cases: 
    (a) unknown-groups model with clustering learned from the data;
    (b) known-groups model with partition induced by \emph{locali} membership;  
    (c) negative binomial mixture of \textsc{bb}s \citep{Ghilotti2025}.}
    \label{fig:post_summary_ndrangheta}
\end{figure}

Consistent with the findings of \cite{esbm_rigon,Lu2025}, the inferred group structures reveal intricate dynamics within this criminal organization. Notably, there is a clear tendency for clusters to form within individual \emph{locali}, though the mechanisms governing group formation vary between affiliates and bosses. This pattern suggests a strategic design aimed at ensuring resilience across different levels of the organization. Furthermore, it highlights that while lower-ranked affiliates can form peripheral groups with relative ease, the formation of core groups, comprising all the bosses, is more constrained, as reflected in their progressively smaller sizes.
A deeper examination of the cluster structure uncovers additional meaningful insights, further supporting observations made in \cite{Lu2025}. For instance, juridical documents indicate that a specific affiliate from the blue locale was attempting to establish a new locale. This is reflected in our clustering results, where this individual is grouped primarily with lower-ranked members from the red locale (first cluster from the top), suggesting a deliberate strategy of recruiting new affiliates from that locale.
Another striking pattern, even more evident in our analysis than in \cite{Lu2025}, involves a cluster composed of bosses from multiple \emph{locali}. Specifically, one green and one blue boss cluster together with two yellow bosses and two lower-ranked yellow affiliates. This aligns with reports that the green and blue bosses supported an unsuccessful attempt to increase the independence of the `Ndrangheta group in Lombardy from the ruling families in Calabria. Consequently, these individuals sought to distance themselves from their original \emph{locali} and strengthen ties elsewhere. Our inferred clustering suggests that this shift primarily moved in the direction of the yellow locale, a hypothesis supported by the cluster composition. In contrast, \cite{Lu2025} infer this movement indirectly through node positioning in a distance-based network representation, rather than as a direct outcome of their methodology.
A further distinction between our results and those in \cite{Lu2025} concerns the clustering of most yellow and purple lower-ranked affiliates. While our model merges them into a single cluster, their approach identifies them as distinct groups. However, their distance-based network representation highlights their close similarity.
Finally, it is noteworthy that core groups of bosses frequently include a few lower-ranked affiliates, suggesting that these individuals may hold more central roles than indicated in the juridical documents. As observed in \cite{Lu2025}, this is particularly evident in the case of a key affiliate grouped with the green locale bosses. In reality, this individual is a high-ranking member with crucial coordinating responsibilities across multiple \emph{locali}.

From inspecting the posterior expected weighted adjacency matrix in Figure \ref{fig:ndrangheta_mixture}, additional insights on the connectivity patterns among different groups can be obtained. The darker blocks along the diagonal corresponding to boss-dominated clusters indicate that these small groups tend to meet frequently and primarily within their own ranks. Additionally, the matrix highlights regular meetings among all boss-dominated groups, likely reflecting their need for coordinated decision-making. An exception to this pattern is the core cluster of the blue locale, which interacts only with the core cluster of the green locale.
Furthermore, as expected, core groups within each locale exhibit strong connectivity with clusters of lower-ranked affiliates from the same locale, reinforcing internal hierarchies. Notably, our earlier hypothesis regarding the cluster of bosses spanning multiple \emph{locali}, suggesting a gradual shift towards the yellow locale, gains further support. This group exhibits a higher frequency of meetings with the cluster containing all yellow lower-ranked affiliates, indicating its strategic alignment and possible integration into the yellow locale.

A key property of both the known-groups and unknown-groups models is that they explicitly allow for the estimation of the number of meetings that went undetected by investigators. Notably, such an analysis cannot be conducted in \citet{esbm_rigon, Lu2025}, as those models implicitly assume a fixed number of meetings. In the analysis of the \emph{Infinito} criminal data, we rely on the inference provided by the unknown-groups model, which was identified as the better-fitting specification based on the comparisons discussed above.
Our unknown-groups model suggests the likely presence of a few number of undetected meetings, which aligns with expectations and confirms the need of accounting for unseen binary traits. Figure \ref{fig:unseen_data} displays the posterior distribution of the number of these unseen meetings, providing a probabilistic estimate of their occurrence. This information may provide guidance to law enforcements to understand the amount of meetings that went undetected. Moreover, as formally proved in Section~\ref{sec:comparison_no_unseen}, accounting for unseen meetings directly affects the clustering structure and, indirectly, the estimated adjacency matrices in Figure~\ref{fig:ndrangheta_matrices}. The overall good fit of the proposed unknown-groups model to data, highlighted in Figure \ref{fig:ndrangheta_matrices}, also validates the estimate of the number of unseen meetings, which results to be much smaller than the estimate provided under the known-groups setting using the clustering induced by \emph{locali} membership, reported in Figure \ref{fig:unseen_locali}. We also compare these findings with those obtained from the negative binomial mixture of \textsc{bb}s \citep{Ghilotti2025}, the natural benchmark for estimating the number of unseen meetings. Hyperparameters are set as in Section \ref{sec:simulations_network}. By construction, this model imposes a homogeneous exchangeable structure, i.e. it does not take into account any group structure, leading to posterior expectations of the adjacency matrix that assign identical values across all entries. As a result, it fails to reproduce the heterogeneity observed in the \emph{Infinito} criminal data (Figure~\ref{fig:ndrangheta_obs_mixture}). Conseqeuntly, its inference on the number of unseen meetings $N^\prime$ substantially differs from that of our model. In particular, Figure~\ref{fig:post_unseen_ndrangheta_negbin_bb} shows that the competitor places high posterior probability on the absence of undetected meetings, in contrast with the inference obtained under the unknown-groups model (Figure~\ref{fig:unseen_data}).
Model comparison via the \textsc{waic} further supports the data-driven clustering approach. Indeed, the negative binomial mixture of \textsc{bb}s attains a \textsc{waic} of $-5392$, which is worse not only than the unknown-groups model but also than the known-groups specification.

\section{Discussion}

In this paper, we have introduced a new Bayesian nonparametric framework for modeling multivariate count data with an unknown number of traits. The models developed in the main manuscript are based on finite completely random vectors, while results for the more general class of completely random vectors are provided in the Supplementary Material. We focused on the cases where $A_{i\ell}\in\{0,1\}$ (binary traits) or $A_{i\ell}\in\{0,1,2,\dots\}$ (Poisson counts), though our framework naturally extends to general parameter spaces for $\theta_{jq}$, as discussed in the Supplementary Material. The general theorems established in Section~\ref{sec:partially_ex_setting} were further leveraged to derive a Gibbs sampler for a mixture model aimed at clustering finite trait allocations. We also emphasized the importance of modeling a random number of traits, which makes it possible to uncover previously unseen meetings within the `Ndrangheta dataset, while simultaneously yielding a better clustering solution, as formalized in Proposition~\ref{prop:cluster_comparison}. Nevertheless, a word of caution is warranted. Estimating the number of unseen meetings is essentially an extrapolation task, which makes it difficult to assess the accuracy of the predictions. Simulation studies confirmed the usefulness of our model under correct specification. However, in real data applications some degree of misspecification is unavoidable. Thus, a pragmatic recommendation is to regard our estimates as rough indications for the true number of unseen meetings, thereby ensuring a conservative interpretation of the results.

Although we focused on the analysis of criminal networks, the scope of our approach extends beyond this domain. For instance, we plan to explore its use to cluster microbiome profiles of different subjects, where traits correspond to operational taxonomic units observed with different frequencies in each subject. The importance of trait allocation models for microbiome data analysis has been recently emphasized by \cite{james2025microbiome}. Another promising area of application is ecology, where our method can be used to cluster sites that exhibit similar patterns in terms of species occurrences \citep{Ghilotti2025}. In this context the traits are the species, observed with varying frequencies in each site.

In addition, our paper also serves as a springboard for future methodological developments. First, all results in Section \ref{app:distribution_theory_CRV} of the Supplementary Material hold true for general \textsc{crv}s, offering a theoretical basis for a broad class of models and providing practitioners with new modeling opportunities. Second, the results in Section \ref{sec:partially_ex_setting} are the basic building blocks for addressing extrapolation problems involving dependent populations of traits. For example, one can estimate the number of unseen traits in an additional sample of subjects, thereby extending the results of \cite{Mas22,Camerlenghi2024,Ghilotti2025} when multiple-sample data are available.  While \cite{shen2024} focused on predicting new genetic variants in presence of two dependent populations, our approach offers additional flexibility in ecological applications by allowing for finite species richness, since the number of traits is assumed finite but random. Finally, we believe that it is not difficult to extend our framework to cluster a collection of trait allocations, rather than individual subjects. This would lead to models akin to nested processes \citep{Rod(08)} within the trait allocation context. The aforementioned applications and methodological developments are left for future research.

\section*{Acknowledgments}
{
The research was partially carried out while L.G. and F.C. were visiting the Department of Biostatistics, University of California Los Angeles (US), in 2024. L.G. and F.C. are supported by the European Union-Next Generation EU funds, component M4C2, investment 1.1., PRIN-PNRR 2022 (P2022H5WZ9). T.R. acknowledges support from the European Union (ERC, NEMESIS, project number: 101116718). Views and opinions expressed are however those of the author(s) only and do not necessarily reflect those of the European Union or the European Research Council Executive Agency. Neither the European Union nor the granting authority can be held responsible for them.}

\bibliography{references}

\clearpage

\setcounter{equation}{0}
\renewcommand\theequation{S\arabic{equation}}
\renewcommand\theHequation{S\arabic{equation}}

\setcounter{theorem}{0}
\renewcommand\thetheorem{S\arabic{theorem}}
\renewcommand\theHtheorem{S\arabic{theorem}}

\setcounter{corollary}{0}
\renewcommand\thecorollary{S\arabic{corollary}}
\renewcommand\theHcorollary{S\arabic{corollary}}

\setcounter{proposition}{0}
\renewcommand\theproposition{S\arabic{proposition}}
\renewcommand\theHproposition{S\arabic{proposition}}

\setcounter{lemma}{0}
\renewcommand\thelemma{S\arabic{lemma}}
\renewcommand\theHlemma{S\arabic{lemma}}

\setcounter{figure}{0}
\renewcommand\thefigure{S\arabic{figure}}
\renewcommand\theHfigure{S\arabic{figure}}

\setcounter{table}{0}
\renewcommand\thetable{S\arabic{table}}
\renewcommand\theHtable{S\arabic{table}}

\setcounter{section}{0}
\renewcommand{\thesection}{S\arabic{section}}
\renewcommand{\theHsection}{S\arabic{section}}

\begin{center}
   \LARGE Supplementary Material for:\\
    ``Bayesian nonparametric modeling of multivariate count data with an unknown number of traits''
\end{center}

\section*{Organization of the Supplementary Material}

The Supplementary Material is organized as follows. In Section \ref{app:CRV}, we review in detail the general class of completely random vectors, with special focus on the notable subclass of finite completely random vectors. Section \ref{app:distribution_theory_CRV} presents the general distribution theory under completely random vector priors, which is then specialized to derive the results in Section  \ref{sec:bayesian_analysis}. Finally, in Section \ref{app:effect_unseen}, we report the proof of Proposition \ref{prop:cluster_comparison}.

\section{Account on completely random vectors} \label{app:CRV}

In this section, we provide an account on completely random vectors (\textsc{crv}s), with the goal of illustrating the connection  between the processes presented in Section \ref{sec:model_form} and completely random vectors. The class of \textsc{crv}s can be considered as a generalization of the class of completely random measures (\textsc{crm}s), introduced by \cite{Kingman(67)}, to the multivariate setting and played a leading role in many successful stories in the \textsc{bnp} literature. Indeed many authors have focused on \textsc{crv}s to model the prior opinion in presence of partially exchangeable data, i.e., data organized in groups \citep{Gri(17),Nip(14),Cam(19)AoS}. The use of \textsc{crv}s also allows the quantification of dependence induced by the prior across the different groups of observations \citep{Cat(21)AoS,Lav(21)}.

To mathematically introduce the class of \textsc{crv}s, consider a Polish space $\X$ equipped with the corresponding Borel $\sigma$-algebra $\mathscr{X}$; consistently with the rest of the paper, $\X$ represents the space of trait labels. Denote by $(\mathsf{M}_\X, \mathcal{M}_\X)$ the space of boundedly finite (non-negative) measures on $\X$ equipped with the corresponding Borel $\sigma$-algebra generated by the topology of weak$^{\#}$-convergence \citep{Daley_volII}. We also indicate by $\mathsf{M}_\X^{ d}$ the $d$-fold product space of $\mathsf{M}_\X$, which is naturally endowed with the product $\sigma$-algebra.
Remind that a random vector of measures $\tilde{\bm{\mu}} = ( \tilde{\mu}_1, \ldots  , \tilde{\mu}_d)$  is a measurable function from a suitable probability space, say $(\Omega, \mathscr{A}, \Pp)$, taking values in $(\mathsf{M}_\X^{ d}, \calM_{\X}^{ d})$. In the sequel, we will use the notation $\tilde{\bm{\mu}} (A)= ( \tilde{\mu}_1 (A), \ldots  , \tilde{\mu}_d (A)) $ to indicate the random vector evaluated on a set $A \in \mathscr{X}$. Then, \textsc{crv}s are defined as follows.
\begin{definition}[Completely random vector]
A random vector of measures $\tilde{\bm{\mu}} = ( \tilde{\mu}_1, \ldots  , \tilde{\mu}_d)$  is a \textsc{crv} if, for any $n \in \N$ and for any collection of disjoint sets $A_1,\ldots,A_n \in \mathscr{X}$, the evaluations $\tilde{\bm{\mu}}(A_1),\ldots,\tilde{\bm{\mu}}(A_n)$ are independent random vectors on $(0,\infty)^d$.
\end{definition}
Note that this definition entails that the marginals of a \textsc{crv} are \textsc{crm}s, since they have independent increments. Moreover, \textsc{crv}s satisfy a similar decomposition as the one that holds true for \textsc{crm}s. Indeed, a \textsc{crv} consists of three main components: (i) a deterministic drift $\bm{u}$, i.e., a deterministic measure with $d$ entries; (ii) a part with random $d$-dimensional jumps $(\bm{\beta}_k)_{k \geq 1}$ at fixed locations $(y_k)_{k \geq 1}$; (iii) a part with random $d$-dimensional jumps 
$(\bm{ \theta}_{j})_{j \geq 1}$ at random locations $(\tilde{X}_j)_{j \geq 1}$. Summing up, the following representation holds true
\begin{equation}
    \label{eq:CRV_decomposition}
    \tilde{\bm{\mu}} = \bm{u}  + \sum_{k \geq 1}  \bm{\beta}_k \delta_{y_k} + \sum_{j \geq 1} \bm \theta_{j} \delta_{\tilde{X}_j}.
\end{equation}
See \citet[Theorem 3.19]{Kallenberg_2017}. As a common practice in \textsc{bnp} models, here we remove the deterministic drift and the part with fixed locations, thus working only with component (iii); in other words we consider \textsc{crv}s of the type:
\[
\tilde{\bm{\mu}} =\sum_{j \geq 1} \bm \theta_{j} \delta_{\tilde{X}_j} = \sum_{j \geq 1}(\theta_{j1}, \ldots , \theta_{jd})\delta_{\tilde{X}_j}.
\] 
Similarly to the \textsc{crm} case, any \textsc{crv} can be expressed as a functional of a Poisson point process \citep{baccelli} on the space 
$(0,\infty)^d \times \X$. Indeed, consider the Laplace functional of $\tilde{\bm{\mu}}$, which is defined as 
\[
\calL_{\tilde{\bm\mu}}(g_1,\ldots,g_d) := \E\left(e^{- \sum_{q=1}^d \int_{\X} g_q(x) \tilde{\mu}_q (\dd x)}\right)
\]
for all measurable functions $g_1, \ldots, g_d: \X \to (0,\infty)$. The Laplace functional of $\tilde{\bm{\mu}}$ admits the following  L\'evy-Khintchine representation
\begin{equation}
    \label{eq:LK_appendix}
   \calL_{\tilde{\bm\mu}}(g_1,\ldots,g_d) =
    \exp \left\{ - \int_{(0,\infty)^d \times \X} 
    \left(1-\exp\left\{- \sum_{q=1}^d g_q (x) \theta_q \right\} \right)     \nu (\dd \theta_1 \, \ldots\, \dd \theta_d \, \dd x)     \right\},
\end{equation}
for all measurable functions $g_1, \ldots, g_d: \X \to (0,\infty)$, where $\nu (\, \cdot \, )$ is a measure on $(0,\infty)^d \times \X$ referred to as the L\'evy intensity of $\tilde{\bm{\mu}}$. The measure $\nu (\, \cdot\, )$ must satisfy the following conditions
\[
\nu((0,\infty)^d \times \{x \}) =0, \; \forall x \in \X, \qquad \text{and}\qquad
\int_{(0,\infty)^d \times B} \min \{||\bm{\theta}|| , 1\} \, \nu (\dd \theta_1\, \ldots \, \dd \theta_d \, \dd x ) < \infty ,
\]
for any bounded Borel set $B \in \mathscr{X}$, where $||\bm{\theta}||$ denotes the Euclidean norm of the vector $\bm{\theta}= (\theta_1, \ldots , \theta_d)$. In the present paper we consider homogeneous \textsc{crv}s, namely \textsc{crv}s whose L\'evy intensity satisfies $\nu (\dd \theta_1\, \ldots \, \dd \theta_d \, \dd x )=
\rho_d (\dd \theta_1\, \ldots \, \dd \theta_d) \alpha (\dd x)$. In particular, assume that $\alpha (\, \cdot\,)$ is a finite non-atomic measure on $\X$ with total mass $\lambda$, thus set $P_0(\, \cdot \, ):= \alpha(\, \cdot \, )/\lambda$. In the sequel, we use the notation  $\tilde{\bm{\mu}} \sim \CRV (\rho_d, \lambda, P_0)$ to indicate the distribution of a homogeneous \textsc{crv}. We refer to \cite{Kallenberg_2017} for additional details on \textsc{crv}s.

\subsection{Finite completely random vectors} \label{app:fCRV}

In this paper, we focus on a special class of \textsc{crv}s which we refer to as finite completely random vectors (\textsc{fcrv}s). Due to the purposes of the paper, we restrict the attention to the homogeneous case; extension to non-homogeneous \textsc{crv}s is straightforward. We define (homogeneous) \textsc{fcrv}s as follows.
\begin{definition}[Finite completely random vector]
Let $\tilde{\bm{\mu}} \sim \CRV (\rho_d, \lambda, P_0)$. If $\rho_d$ is a finite measure, then $\tilde{\bm{\mu}}$ is a \textsc{fcrv}.    
\end{definition}
When $\tilde{\bm{\mu}}$ is a \textsc{fcrv}, we write $\tilde{\bm{\mu}} \sim \textsc{fcrv} (\rho_d, \lambda, P_0)$, where $\rho_d$ is necessarily finite. Without loss of generality, assume that $\rho_d = H^{(d)}$ is a probability measure and adjust $\lambda$ to account for it. 
As we will discuss in the next proposition, the specification \emph{finite} \textsc{crv} is motivated by the fact that \textsc{fcrv}s can be seen as functionals of finite Poisson point processes \citep[Section 4.3]{baccelli}. In particular, the random measures $\tilde{\mu}_1,\ldots,\tilde{\mu}_d$ are supported on a (finite) Poisson random number of common atoms. 

\begin{proposition}
    \label{prp:fCRV_appendix}
    If  $\tilde{\bm{\mu}} \sim \textsc{fcrv} (H^{(d)}, \lambda, P_0)$, then $\tilde{\bm \mu}$ can be represented as
    \begin{equation}    \label{eq:fCRV_representation_appendix}
        \tilde{\bm \mu} = \sum_{j=1}^{N} (\theta_{j1},\ldots,\theta_{jd}) \delta_{\tilde{X}_j},
    \end{equation}
    where $N$ is a Poisson random variable with parameter $\lambda>0$. Moreover, conditionally to $N$, the vectors $(\theta_{j1},\ldots,\theta_{jd})$ are iid draws from $H^{(d)}$, namely $(\theta_{j1},\ldots,\theta_{jd}) \iid H^{(d)}$, and $\tilde{X}_j \iid P_0$, as $j=1,\ldots,N$.
\end{proposition}
\begin{proof}
In order to show the result, we prove that the Laplace functional of $\tilde{\bm\mu}$, as defined in \eqref{eq:fCRV_representation_appendix}, coincides with the one of a $\textsc{fcrv}(H^{(d)}, \lambda, P_0)$.
To this end, consider the measurable functions $g_1, \ldots , g_d: \X \to (0,\infty)$ and note that, conditionally on $N$, we have 
\begin{equation*}
\begin{split}
    &\E\left[\exp \left\{- \sum_{q=1}^d \int_\X g_q (x) \tilde{\mu}_q (\dd x)  \right\} \,\mid\, N \right] =  \E\left[ \exp
\left\{- \sum_{q=1}^d \sum_{j=1}^{N} \theta_{jq} \, g_q(\tilde{X}_j) \right\} \,\mid \, N\right]\\
&\qquad = \E\left[ \E\left[ \exp
\left\{- \sum_{q=1}^d \sum_{j=1}^{N} \theta_{jq} \, g_q(\tilde{X}_j) \right\} \,\mid \, N, \tilde{X}_1,\ldots,\tilde{X}_N \right] \,\mid \, N\right]\\
& \qquad= \E\left[ \, \prod_{j=1}^N \E\left[ \exp
\left\{- \sum_{q=1}^d \theta_{jq} \, g_q(\tilde{X}_j) \right\} \,\mid \, N,  \tilde{X}_j \right] \,\mid \, N\right],
\end{split}
\end{equation*}
where we first exploit the tower property of conditional expectations and then we use the independence of the random vectors, conditionally on $N$. Observing that, conditionally on $N$, the vectors $(\theta_{j1},\ldots,\theta_{jd})$ are iid from distribution $H^{(d)}$ and $\tilde{X}_j \iid P_0$, for $j=1, \ldots , N$, the previous expression equals
\begin{align} 
        &\E\left[\exp \left\{- \sum_{q=1}^d \int_\X g_q (x) \tilde{\mu}_q (\dd x)  \right\} \,\mid\, N \right] 
= \E\left[ \, \prod_{j=1}^N \int_{(0,\infty)^d} \exp
\left\{- \sum_{q=1}^d \theta_{q} \, g_q(\tilde{X}_j) \right\} H^{(d)}(\dd \theta_1\, \ldots \, \dd \theta_d) \,\mid \, N\right] \notag \\ 
& \qquad = \prod_{j=1}^N \int_{\X} \int_{(0,\infty)^d} \exp
\left\{- \sum_{q=1}^d \theta_{q} \, g_q(x) \right\} H^{(d)}(\dd \theta_1\, \ldots \, \dd \theta_d) P_0(\dd x) \notag \\
& \qquad= \left( \int_{(0,\infty)^d \times \X} \exp
\left\{- \sum_{q=1}^d \theta_{q} \, g_q(x) \right\} H^{(d)}(\dd \theta_1\, \ldots \, \dd \theta_d) P_0(\dd x)  \right)^N. \label{eq:conditional_laplace_fCRV}
\end{align}
We can now integrate out $N$ in the final expression of \eqref{eq:conditional_laplace_fCRV} to obtain the Laplace functional of $\tilde{\bm\mu}$ defined as in \eqref{eq:fCRV_representation_appendix}. Specifically, we have
\begin{align}
        &  \calL_{\tilde{\bm\mu}}(g_1,\ldots,g_d) = \E\left[ \E\left[\exp \left\{- \sum_{q=1}^d \int_\X g_q (x) \tilde{\mu}_q (\dd x)  \right\} \,\mid\, N \right] \right] \nonumber \\
& \qquad = \E\left[ \left( \int_{(0,\infty)^d \times \X} \exp
\left\{- \sum_{q=1}^d \theta_{q} \, g_q(x) \right\} H^{(d)}(\dd \theta_1\, \ldots \, \dd \theta_d) P_0(\dd x)  \right)^N \right] \nonumber\\
& \qquad = e^{-\lambda} \sum_{j \geq 0} \frac{1}{j!}\left[ \lambda\int_{(0,\infty)^d \times \X} \exp
\left\{- \sum_{q=1}^d \theta_{q} \, g_q(x) \right\} H^{(d)}(\dd \theta_1\, \ldots \, \dd \theta_d) P_0(\dd x) \right]^j \nonumber\\
& \qquad=  e^{-\lambda} \exp\left\{ \lambda \int_{(0,\infty)^d \times \X} \exp
\left\{- \sum_{q=1}^d \theta_{q} \, g_q(x) \right\} H^{(d)}(\dd \theta_1\, \ldots \, \dd \theta_d) P_0(\dd x) \right\} \nonumber\\
& \qquad= \exp\left\{ -\lambda \left( 1 - \int_{(0,\infty)^d \times \X} \exp
\left\{- \sum_{q=1}^d \theta_{q} \, g_q(x) \right\} H^{(d)}(\dd \theta_1\, \ldots \, \dd \theta_d) P_0(\dd x)\right)\right\} \nonumber\\
& \qquad = \exp\left\{ -\lambda   \int_{(0,\infty)^d \times \X} \left( 1 - \exp
\left\{- \sum_{q=1}^d \theta_{q} \, g_q(x) \right\} \right) H^{(d)}(\dd \theta_1\, \ldots \, \dd \theta_d) P_0(\dd x) \right\},
\label{eq:laplace_fCRV}
\end{align}
where the first equality follows from the tower property,  expression \eqref{eq:conditional_laplace_fCRV} has been exploited in the second line, and finally the expectation has been evaluated relying on the fact that $N$ is a Poisson random variable with parameter $\lambda$.
From \eqref{eq:LK_appendix}, we note that the equation in  \eqref{eq:laplace_fCRV} corresponds to the Laplace functional of a $\textsc{fcrv}(H^{(d)}, \lambda, P_0)$, and the thesis follows.
\end{proof}

Based on Proposition \ref{prp:fCRV_appendix}, the connection between the prior distribution for $\tilde{\bm{\mu}}$ defined in equations~\eqref{eq:finite_CRV}--\eqref{eq:param_H}, Section \ref{sec:model_form}, and the class of \textsc{fcrv}s is evident. Indeed, the construction in  equations~\eqref{eq:finite_CRV}--\eqref{eq:param_H} precisely corresponds to the formulation in \eqref{eq:fCRV_representation_appendix}, thus $\tilde{\bm{\mu}}$ can be equivalently described as $\tilde{\bm{\mu}} \sim \textsc{fcrv} (H^{(d)}, \lambda, P_0)$.

\section{General distribution theory under CRV priors} \label{app:distribution_theory_CRV}

In this section, we provide a complete distribution theory for partially exchangeable trait allocation models \eqref{eq:partially_ex_traits}, when the prior on $\tilde{\bm\mu}$ is a generic \textsc{crv}. This encompasses both finite- and infinite-dimensional trait models. The results presented here are very general and can be specialized to a variety of different prior specifications. Moreover, they can also be applied to face prediction in presence of multiple dependent populations. In the last part of this section, we will focus on the finite trait models of Section \ref{sec:model_form}, corresponding to the case of $\tilde{\bm\mu}$ distributed as a \textsc{fcrv}. In particular, we will comment on how the results presented in Section \ref{sec:bayesian_analysis}  follow from these general theorems.

First, we recall the general model we deal with:
\begin{equation} \label{eq:mixture_model_allocation_CRV}
   \begin{split}
   Z_{iq} \mid \tilde{\mu}_q &\ind \CP (\tilde{\mu}_q), \qquad i\geq 1, \quad q=1,\ldots,d,\\
    \tilde{\bm\mu} = (\tilde{\mu}_1, \ldots, \tilde{\mu}_d) & \sim \CRV (\rho_d, \lambda, P_0) ,
\end{split} 
\end{equation}
where $\rho_d$ is a measure on $(0,\infty)^d$, $\lambda > 0$ and $P_0$ is a non-atomic probability measure on $\X$. Refer to Section \ref{app:CRV} for the definition of  $\CRV (\rho_d, \lambda, P_0)$. Clearly, model \eqref{eq:partially_ex_traits} with prior distribution given by equations~\eqref{eq:finite_CRV}--\eqref{eq:param_H} is obtained by considering $\rho_d = H^{(d)}$, with $H^{(d)}$ being a probability law on $(0,\infty)^d$, that is $\tilde{\bm\mu} \sim \textsc{fcrv}(H^{(d)}, \lambda, P_0)$.

We first focus on the marginal distribution of a sample $\bm{Z}=(Z_{iq}:  i= 1,\ldots,n_q; q=1,\ldots,d)$. As explained in the main paper, with marginal distribution of $\bm Z$, we refer to the probabilities of the events $(\bm{A} = \bm{a}, K_n = k)$, where $K_n$ is the observed number of distinct traits and $\bm A$ collects the observed counts, as introduced for Theorem \ref{thm:marginal_fCRV}, Section~\ref{sec:bayesian_analysis}. The following theorem describes the marginal probability for the event $(\bm{A} = \bm{a}, K_n = k)$.
\begin{theorem}[Marginal distribution, p\textsc{etpf}]
    \label{thm:marginal_CRV}
    Let $\bm{Z}$ be a sample from the statistical model~\eqref{eq:mixture_model_allocation_CRV}. The probability that  $\bm{Z}$ displays $K_n = k$ distinct traits with counts $\bm A = \bm{a}$ equals
    \begin{equation}
        \label{eq:pEFPF_CRV}
        \begin{split}
       \pi_n (\bm{a}) & =\frac{{\lambda}^k}{k!} \exp \left\{  -\lambda  \int_{(0,\infty)^d}  \left( 1- \prod_{q=1}^d  P(0; \theta_q)^{n_q} \right)    \rho_d (\dd \theta_1\, \ldots \,\dd \theta_d)    \right\} \\
    & \qquad\times  \prod_{\ell=1}^k \int_{(0,\infty)^d}  
    \prod_{q=1}^d   \prod_{i=1}^{n_q} P (a_{i\ell q} ; \theta_q) 
    \rho_d (\dd \theta_1\, \ldots \,\dd \theta_d)
        \end{split}
    \end{equation}
    where $n=\sum_{q=1}^d n_q$ and $\bm{n}=(n_1, \ldots , n_d)$ are the sample sizes.
\end{theorem}
\begin{proof}
    Consider the event $(\bm{A} = \bm{a}, K_n = k)$, corresponding to $\bm Z$ displaying $K_n = k$ distinct traits, with associated counts described by $\bm A = \bm{a}$. Though already specified in the definition of $\bm A$, it worth stressing that this event $(\bm{A} = \bm{a}, K_n = k)$ refers to one specific ordering of the traits among all the possible ones, when a uniform distribution on the orderings is assumed. For the sake of exposition, it is convenient to observe that $\bm A = \bm a$ describes the presence and absence of the $k$ observed traits in the subjects. In particular, for each subpopulation $q = 1,\ldots, d$ and each observed trait $\ell = 1,\ldots,k$, define  $\calB_{\ell q} = \{(i,q): a_{i\ell q} > 0, i=1,\ldots,n_q\}$, which might be empty.

    The probability of the event  under study $(\bm{A} = \bm{a}, K_n = k)$ is indicated by $\pi_n (\bm{a})$ and generalizes the well-known notion of partially Exchangeable Partition Probability Function (p\textsc{eppf}) to the setting of trait allocation models, thus we refer to it as partially Exchangeable Trait Probability Function (p\textsc{etpf}).  In order to evaluate $\pi_n (\bm{a})$, we apply a limiting argument, which is based on the evaluation of the joint distribution of $(\bm{A}, K_n)$ and the associated trait labels $X_1,\ldots,X_{K_n}$.
    To this end, we consider $k$ small balls $B_\varepsilon (X_\ell)$ centered at $X_\ell$ with radius $\varepsilon$, where $\varepsilon>0$ is sufficiently small so that all the balls are disjoint. Moreover, we set $\X^* := \X \setminus \cup_{\ell=1}^k B_\varepsilon (X_\ell)$. 
    First, focus on the probability of the event $\Ecr$
\begin{equation}
\label{eq:ECR_DEF}
\begin{split}
     \Ecr = \left\{ \right. & \exists j \geq 1 : \; \tilde{A}_{ijq} = a_{i\ell q} \text{ with } \tilde{X}_j \in B_\varepsilon (X_\ell)  , \; \forall (i,q) \in \calB_{\ell q},\\
     & \qquad\qquad\qquad\qquad\qquad\qquad\qquad q=1,\ldots,d, \;\ell=1,\ldots,k;\\
     &Z_{iq}(B_\varepsilon(X_\ell)) =0, \; \forall (i,q) \not\in \calB_{\ell q}, \;q=1,\ldots,d, \;\ell=1,\ldots,k;\\
     &Z_{iq}(\X^*)=0,\; \forall i=1,\ldots,n_q, \;q=1,\ldots,d \left. \right\},
\end{split}
\end{equation}
that is to say the \textit{infinitesimal} probability of observing $k$ distinct traits with associated counts $\bm A = \bm{a}$ and trait labels belonging to the small balls $B_\varepsilon(X_1), \ldots , B_\varepsilon (X_k)$. 

To evaluate the probability of $\Ecr$, by an application of the tower property, we have
\begin{equation}
    \Pp\left(\Ecr\right) = \E\left[\Pp\left(\Ecr\mid \tilde{\bm\mu}\right)\right] .
\label{eq_prob_xi_initial_generic}
\end{equation}
Define the following three disjoint events, whose union equals $\Ecr$, as 
\begin{align*}
   \Ecr_1 &:= \left\{  \exists j \geq 1 : \; \tilde{A}_{ijq} = a_{i\ell q} \text{ with } \tilde{X}_j \in B_\varepsilon (X_\ell)  , \; \forall (i,q) \in \calB_{\ell q},\; q=1,\ldots,d, \;\ell=1,\ldots,k  \right\},\\
   \Ecr_2 &:=\left\{ Z_{iq}(B_\varepsilon(X_\ell)) =0, \; \forall (i,q) \not\in \calB_{\ell q}, \;q=1,\ldots,d, \;\ell=1,\ldots,k \right\}, \\
     \Ecr_3 &:= \left\{ Z_{iq}(\X^*)=0,\; \forall i=1,\ldots,n_q, \;q=1,\ldots,d  \right\}.
\end{align*}
Conditionally to $\tilde{\bm\mu}$, the randomness of the observations $Z_{iq}$ from model \eqref{eq:mixture_model_allocation_CRV} is only in
the independent random variables $\tilde{A}_{ij q}$. Thus, the three events $\Ecr_1, \Ecr_2, \Ecr_3$ are independent, conditionally to $\tilde{\bm\mu}$, since they relate to independent $\tilde{A}_{ijq}$,  and their probabilities can be evaluated separately as
\begin{equation}
    \Pp\left(\Ecr\right) = \E\left[\Pp\left(\Ecr_1\mid \tilde{\bm\mu}\right) \Pp\left(\Ecr_2\mid \tilde{\bm\mu}\right)
    \Pp\left(\Ecr_3\mid \tilde{\bm\mu}\right)\right] .
\label{eq:prob_ecr123}
\end{equation}
Hereafter, by a slight abuse of notation, we write $i \in \calB_{\ell q}$ to indicate that $(i, q) \in \calB_{\ell q}$. 
For the first conditional probability in \eqref{eq:prob_ecr123}, it is easy to see that
\begin{align*}
    \Pp (\Ecr_1 \mid \tilde{\bm\mu}) 
   & = \prod_{q=1}^d \prod_{\ell=1}^k \prod_{i \in \calB_{\ell q}} \Pp\left(\exists j \geq 1 : \; \tilde{A}_{ijq} = a_{i\ell q} \text{ with } \tilde{X}_j \in B_\varepsilon (X_\ell) \, \mid\, \tilde{\bm\mu} \right) \\
   & = \prod_{q=1}^d \prod_{\ell=1}^k \prod_{i \in \calB_{\ell q}} \left( 1 - \prod_{j\geq1} (1 - P(a_{i\ell q}; \theta_{jq}))^{\delta_{\tilde{X}_j}(B_\varepsilon(X_\ell))} \right) .
\end{align*}
For the second conditional probability in \eqref{eq:prob_ecr123}, we observe that
\begin{align*}
    \Pp (\Ecr_2 \mid \tilde{\bm\mu}) &=
    \prod_{q=1}^d \prod_{\ell=1}^k \prod_{i \not\in \calB_{\ell q}}  \Pp\left(Z_{iq}(B_\varepsilon(X_\ell)) =0 \,\mid \, \tilde{\mu}_q \right) \\
    & =  \prod_{q=1}^d \prod_{\ell=1}^k \prod_{i \not\in \calB_{\ell q}} \prod_{j\geq1}P(0; \theta_{jq})^{\delta_{\tilde{X}_j}(B_\varepsilon(X_\ell))} .
\end{align*}
Analogous considerations hold true for the conditional probability of $\Ecr_3$ in \eqref{eq:prob_ecr123}. Thus, putting everything together, we have
\begin{equation}\label{eq_prob_xi_cond_mutilde_generic}
\begin{aligned}
    \Pp(\Ecr) & = \E \left[  \prod_{q=1}^d \prod_{\ell=1}^k \prod_{i \not\in \calB_{\ell q}}  \prod_{j\geq 1} P(0 ; \theta_{jq})^{\delta_{\tilde{X}_j} (B_\varepsilon (X_\ell))} \right.\\
    & \qquad\qquad \times
    \prod_{q=1}^d \prod_{\ell=1}^k \prod_{i \in \calB_{\ell q}} 
    \left( 1 - \prod_{j\geq1} (1 - P(a_{i\ell q} ; \theta_{jq}))^{\delta_{\tilde{X}_j}(B_\varepsilon(X_\ell))} \right)\\
    & \qquad\qquad \times \left.
    \prod_{q=1}^d  \prod_{j\geq 1} \prod_{i=1}^{n_q}P(0 ; \theta_{jq})^{\delta_{\tilde{X}_j} (\X^*)}
 \right] .
\end{aligned}
\end{equation}
By observing that, for any $\ell=1,\ldots,k$ and for any $q=1,\ldots,d$,
\begin{equation*}
\begin{split}
&\prod_{i \in \calB_{\ell q}} 
    \left( 1 - \prod_{j\geq1} (1 - P(a_{i\ell q} ; \theta_{jq}))^{\delta_{\tilde{X}_j}(B_\varepsilon(X_\ell))} \right)\\
 &   \qquad 
 = \sum_{\substack{v_{i\ell q}\in \{ 0,1\}\\  i \in \calB_{\ell q}}} \prod_{i \in \calB_{\ell q}}(-1)^{v_{i\ell q}} \prod_{j\geq1}(1 - P(a_{i\ell q}; \theta_{jq}))^{v_{i\ell q} \delta_{\tilde{X}_j}(B_\varepsilon(X_\ell))},
 \end{split}
\end{equation*}
the expected value in \eqref{eq_prob_xi_cond_mutilde_generic} boils down to the following sum
\begin{equation} \label{eq:PE_generic}
\begin{aligned}
    \Pp(\Ecr) & =\sum_{\substack{v_{i\ell q}\in \{0,1\} \\ i \in \calB_{\ell q}  
    \\ q=1,\ldots,d\\ \ell=1,\ldots,k} }
    \left(\prod_{q=1}^d \prod_{\ell=1}^k
    \prod_{i \in \calB_{\ell q} } (-1)^{v_{i\ell q}} \right)
    \E \left[  \prod_{q=1}^d \prod_{\ell=1}^k \prod_{i \not\in \calB_{\ell q}} \prod_{j\geq 1} P(0; \theta_{jq})^{ \delta_{\tilde{X}_j}(B_\varepsilon(X_\ell))} \right.\\
    & \qquad\qquad \times \prod_{q=1}^d \prod_{\ell=1}^k \prod_{i \in \calB_{\ell q}} \prod_{j \geq 1} 
  (1 - P(a_{i\ell q}; \theta_{jq}))^{v_{i\ell q} \delta_{\tilde{X}_j}(B_\varepsilon(X_\ell))}\\
  & \qquad\qquad \times\left. 
  \prod_{q=1}^d \prod_{j\geq 1} \prod_{i=1}^{n_q} P(0; \theta_{jq})^{ \delta_{\tilde{X}_j} (\X^*)}
 \right] .
\end{aligned}
\end{equation}
At this point, we need to evaluate the expectation in \eqref{eq:PE_generic}. 
Since the sets $B_\varepsilon (X_1), \ldots , B_{\varepsilon}(X_k)$ and $ \X^*$ are disjoint, the independence property of the \textsc{crv} $\tilde{\bm\mu}$ implies that the expected value in \eqref{eq:PE_generic} equals the following product
\begin{equation} \label{eq:expected_tot_generic}
\prod_{\ell=1}^k E_\ell \times E_* ,
\end{equation}
where we have set:
\begin{align*}
        E_\ell & := \E \left\{  \prod_{j \geq 1 } \prod_{q=1}^d \left[ \prod_{i \not\in \calB_{\ell q}} P(0; \theta_{jq})^{\delta_{\tilde{X}_j}(B_\varepsilon(X_\ell))} \cdot \prod_{i \in \calB_{\ell q}}
  (1 - P(a_{i\ell q}; \theta_{jq}))^{v_{i\ell q} \delta_{\tilde{X}_j}(B_\varepsilon(X_\ell))} \right] \right\},\\
 E_*  & :=  \E \left[   \prod_{q=1}^d \prod_{j\geq 1} \prod_{i=1}^{n_q} P(0; \theta_{jq})^{ \delta_{\tilde{X}_j} (\X^*)}\right].
\end{align*}
We now concentrate on the evaluation of $E_*$ in \eqref{eq:expected_tot_generic}:
\begin{equation} \label{eq:expect_1_generic}
    \begin{split}
        E_* & =  \E \left[   \prod_{q=1}^d \prod_{j\geq 1} \prod_{i=1}^{n_q} P(0; \theta_{jq})^{ \delta_{\tilde{X}_j} (\X^*)}\right]\\
         & = \E  \left[ \exp \left\{ \sum_{j \geq 1} \log \left(\prod_{q=1}^d \prod_{i=1}^{n_q} P(0; \theta_{jq})\right) \delta_{\tilde{X}_j} (\X^*)\right\}\right]\\
         &  = \exp \left\{  - \int_{(0,\infty)^d \times \X^*}  \left( 1-\prod_{q=1}^d \prod_{i=1}^{n_q} P(0; \theta_q) \right)    \rho_d (\dd \theta_1\, \ldots \,\dd \theta_d)   \, \lambda P_0 (\dd x)  \right\},
    \end{split}
\end{equation}
where we exploit the available expression of the Laplace functional for the \textsc{crv}s (or analogously for Poisson processes).
As for the general term $E_\ell$ in \eqref{eq:expected_tot_generic}, the L\'evy-Khintchine representation implies again
\begin{equation*} 
    \begin{aligned}
    E_\ell & = \E \left[ \exp \left\{  \sum_{j \geq 1 }\log \left( \prod_{q=1}^d \left[ \prod_{i \not\in \calB_{\ell q}} P(0; \theta_{jq})
  \prod_{i \in \calB_{\ell q}}
  (1 - P(a_{i\ell q}; \theta_{jq}))^{v_{i\ell q}  } \right]  \right)  \delta_{\tilde{X}_j}(B_\varepsilon(X_\ell))\right\} \right] \\
  &  = \exp\left\{ - \lambda P_0 (B_\varepsilon (X_\ell))   \int_{(0,\infty)^d}  \left( 1-\prod_{q=1}^d \left[ \prod_{i \not\in \calB_{\ell q}} P(0; \theta_q) \prod_{i \in \calB_{\ell q}}(1-P(a_{i\ell q}; \theta_q))^{v_{i\ell q} } \right] 
  \right) \rho_d (\dd \theta_1\, \ldots \, \dd \theta_d)  \right\}.
    \end{aligned}
\end{equation*}
By defining
\[
I_{\ell, v_{i\ell q}} :=   \int_{(0,\infty)^d}  \left( 1-\prod_{q=1}^d \left[ \prod_{i \not\in \calB_{\ell q}} P(0; \theta_q) \prod_{i \in \calB_{\ell q}}(1-P(a_{i\ell q}; \theta_q))^{v_{i\ell q} } \right] 
  \right) \rho_d (\dd \theta_1\, \ldots \, \dd \theta_d) ,
\]
we get
\begin{equation} \label{eq:expected_2_generic}
E_\ell = \exp \left\{ -\lambda P_0 (B_\varepsilon (X_\ell)) I_{\ell, v_{i\ell q}} \right\}.
\end{equation}
We now substitute the expressions \eqref{eq:expect_1_generic}--\eqref{eq:expected_2_generic} in \eqref{eq:expected_tot_generic} to evaluate the expected value in \eqref{eq:PE_generic}, which equals
\begin{equation*} \label{eq:expected_tot_evaluated_generic}
    \begin{split}
   & \prod_{\ell=1}^k \exp\left\{-\lambda P_0 (B_\varepsilon (X_\ell)) I_{\ell, v_{i\ell q}}  \right\} \times 
         \exp \left\{  -\lambda P_0(\X^*)  \int_{(0,\infty)^d}  \left( 1-\prod_{q=1}^d \prod_{i=1}^{n_q} P(0; \theta_q) \right)    \rho_d (\dd \theta_1\, \ldots \,\dd \theta_d)   \right\}\\
         & \qquad = \exp \left\{  -\lambda P_0(\X^*)  \int_{(0,\infty)^d}  \left( 1-\prod_{q=1}^d \prod_{i=1}^{n_q} P(0; \theta_q) \right)    \rho_d (\dd \theta_1\, \ldots \,\dd \theta_d)  \right\}\\
         & \qquad\qquad \times \prod_{\ell=1}^k \left( 1  -\lambda P_0 (B_\varepsilon (X_\ell)) I_{\ell, v_{i\ell q}}   +  o (P_0 (B_\varepsilon (X_\ell)) ) \right).
    \end{split}
\end{equation*}
We finally substitute the previous expression in \eqref{eq:PE_generic} to get the following final expression for the probability of the event $\Ecr$,
\begin{equation*}
    \begin{aligned}
    \Pp(\Ecr) & =\sum_{\substack{v_{i\ell q}\in \{0,1\} \\ i \in \calB_{\ell q}  
    \\ q=1,\ldots,d\\ \ell=1,\ldots,k} } \left\{
    \left(\prod_{q=1}^d \prod_{\ell=1}^k
    \prod_{i \in \calB_{\ell q} } (-1)^{v_{i\ell q}} \right) 
    \prod_{\ell=1}^k \Big( 1  -\lambda P_0 (B_\varepsilon (X_\ell)) I_{\ell, v_{i\ell q}}   +  o (P_0 (B_\varepsilon (X_\ell)) ) \Big)
    \right\}\\
  &  \qquad \qquad\times
   \exp \left\{  -\lambda P_0(\X^*)  \int_{(0,\infty)^d}  \left( 1-\prod_{q=1}^d \prod_{i=1}^{n_q} P(0; \theta_q) \right)    \rho_d (\dd \theta_1\, \ldots \,\dd \theta_d)  \right\}.
\end{aligned}
\end{equation*}
Handling the last expression, we focus on the sum over the $v_{i\ell q}$ and we get
\begin{equation*}
    \begin{aligned}
    \Pp(\Ecr) & = \exp \left\{  -\lambda P_0(\X^*)  \int_{(0,\infty)^d}  \left( 1-\prod_{q=1}^d \prod_{i=1}^{n_q} P(0; \theta_q) \right)    \rho_d (\dd \theta_1\, \ldots \,\dd \theta_d)  \right\} \\
    & \qquad \times  \prod_{\ell=1}^k \Big( \; \sum_{\substack{v_{iq}\in \{0,1\} \\ i \in \calB_{\ell q}  
    \\ q=1,\ldots,d} }
  \Big(  \prod_{q=1}^d 
    \prod_{i \in \calB_{\ell q} } (-1)^{v_{iq}}  \Big)( 1  -\lambda P_0 (B_\varepsilon (X_\ell)) I_{\ell, v_{iq}}   +  o (P_0 (B_\varepsilon (X_\ell)) ) ) \Big)\\
   & = \exp \left\{  -\lambda P_0(\X^*)  \int_{(0,\infty)^d}  \left( 1-\prod_{q=1}^d \prod_{i=1}^{n_q} P(0; \theta_q) \right)    \rho_d (\dd \theta_1\, \ldots \,\dd \theta_d)  \right\} \\
    & \qquad \times  \prod_{\ell=1}^k \Big(\;  \sum_{\substack{v_{iq}\in \{0,1\} \\ i \in \calB_{\ell q}  
    \\ q=1,\ldots,d} }
  \Big(  \prod_{q=1}^d 
    \prod_{i \in \calB_{\ell q} } (-1)^{v_{iq}}  \Big)( -\lambda P_0 (B_\varepsilon (X_\ell)) I_{\ell, v_{iq}}   +  o (P_0 (B_\varepsilon (X_\ell)) ) ) \Big)\\
\end{aligned}
\end{equation*}
where we use the fact that
\[
 \sum_{\substack{v_{iq}\in \{0,1\} \\ i \in \calB_{\ell q}  
    \\ q=1,\ldots,d} }
  \Big(  \prod_{q=1}^d 
    \prod_{i \in \calB_{\ell q} } (-1)^{v_{iq}}  \Big) = 0.
\]
By exploiting the definition of $ I_{\ell, v_{iq}}$ and applying similar arguments as above, we obtain
\begin{equation*}
  \begin{split}
    \Pp(\Ecr)  & =\exp \left\{  -\lambda P_0(\X^*)  \int_{(0,\infty)^d}  \left( 1-\prod_{q=1}^d \prod_{i=1}^{n_q} P(0; \theta_q) \right)    \rho_d (\dd \theta_1\, \ldots \,\dd \theta_d)  \right\}\\
     & \quad \times  \prod_{\ell=1}^k \left\{ \lambda P_0 (B_\varepsilon (X_\ell)) \int_{(0,\infty)^d}  \prod_{q=1}^d  \Big(  \prod_{i \not\in \calB_{\ell q}} P(0; \theta_q)  \sum_{\substack{v_{i}\in \{0,1\} \\ i \in \calB_{\ell q} } } \prod_{i \in \calB_{\ell q}}(P(a_{i\ell q}; \theta_q) -1)^{v_{i} } \Big)
    \rho_d (\dd \theta_1\, \ldots \,\dd \theta_d) \right\}\\
    &\quad +o \Big( \prod_{\ell=1}^k P_0 (B_\varepsilon (X_\ell)) \Big).
\end{split}
\end{equation*}
It is now easy to solve the summation over the $v_i$ to get the following expression:
\begin{equation*}
  \begin{split}
    \Pp(\Ecr)  & = \exp \left\{  -\lambda P_0(\X^*)  \int_{(0,\infty)^d}  \left( 1-\prod_{q=1}^d \prod_{i=1}^{n_q} P(0; \theta_q) \right)    \rho_d (\dd \theta_1\, \ldots \,\dd \theta_d)  \right\} \\
    & \quad \times  \prod_{\ell=1}^k \left\{ \lambda P_0 (B_\varepsilon (X_\ell)) \int_{(0,\infty)^d}  \prod_{q=1}^d  \Big(  \prod_{i \not\in \calB_{\ell q}} P(0; \theta_q)   \prod_{i \in \calB_{\ell q}} P(a_{i\ell q}; \theta_q) \Big)
    \rho_d (\dd \theta_1\, \ldots \, \dd \theta_d) \right\}  \\
  & \quad  +o \Big( \prod_{\ell=1}^k P_0 (B_\varepsilon (X_\ell)) \Big).
\end{split}
\end{equation*}
Finally, to obtain the targeted probability $\pi_n (\bm{a})$, namely the p\textsc{etpf} of model \eqref{eq:mixture_model_allocation_CRV}, we need to link it with the just computed probability of the event $\Ecr$. For sufficiently small $\varepsilon$, $B_\varepsilon(X_1), \ldots, B_\varepsilon(X_k)$ are disjoint, thus it holds
\begin{equation*}
    \Pp(\Ecr) = k! \prod_{\ell=1}^k P_0(B_\varepsilon(X_{\ell})) \cdot \pi_n (\bm{a}),
\end{equation*}
where $k!$ discounts for the specific ordering of the traits which is implicit in $\pi_n (\bm{a})$, as detailed in the initial comments of the proof.
Then,
\begin{equation*}
\begin{aligned}
    \pi_n (\bm{a}) &= \lim_{\varepsilon \to 0}{\frac{\Pp(\Ecr)}{k! \prod_{\ell=1}^k P_0(B_\varepsilon(X_{\ell}))} } \\
    &= \frac{{\lambda}^k}{k!}  \exp \left\{  -\lambda  \int_{(0,\infty)^d}  \left( 1-\prod_{q=1}^d \prod_{i=1}^{n_q} P(0; \theta_q) \right)    \rho_d (\dd \theta_1\, \ldots \,\dd \theta_d)  \right\} \\
     & \qquad \times  \prod_{\ell=1}^k  \int_{(0,\infty)^d}  \prod_{q=1}^d  \Big(  \prod_{i \not\in \calB_{\ell q}} P(0; \theta_q)   \prod_{i \in \calB_{\ell q}} P(a_{i\ell q}; \theta_q) \Big)
    \rho_d (\dd \theta_1\, \ldots \, \dd \theta_d), 
\end{aligned}
\end{equation*}
which can be rewritten as in the statement of the proposition.
\end{proof}

Secondly, we move to the characterization of the posterior distribution of $\tilde{\bm{\mu}}$ in~\eqref{eq:mixture_model_allocation_CRV}, conditionally to a sample $\bm{Z}$. 
\begin{theorem}[Posterior distribution]
    \label{thm:posterior_CRV}
       Let $\bm{Z}$ be a sample from the statistical model~\eqref{eq:mixture_model_allocation_CRV}. If $\bm{Z}$ displays $K_n = k$ distinct traits labeled $X_1, \ldots , X_k$, with associated counts $\bm{a}$, then the posterior distribution of 
       $\tilde{\bm\mu}$ satisfies the distributional equality        \begin{equation}
           \label{eq:posterior_CRV}
            (\tilde{\mu}_1, \ldots, \tilde{\mu}_d) \mid \bm{Z} \stackrel{d}{=}  (\mu_1^*, \ldots, \mu_d^*) +  (\mu_1', \ldots, \mu_d'),
       \end{equation}
       where $\bm{\mu}^*:= (\mu_1^*, \ldots, \mu_d^*)$ and $\bm{ \mu}':=(\mu_1', \ldots, \mu_d')$ are independent random vectors such that
       \begin{itemize}
           \item[(i)] the components of the vector $\bm{\mu}^*$ are defined as
           $ \mu_q^*(\cdot) = \sum_{\ell=1}^k \theta_{\ell q}^* \delta_{X_\ell}(\cdot)$, for $ q=1, \ldots , d
           $, and the random vectors $(\theta_{\ell 1}^* , \ldots , \theta_{\ell d}^*)$ are independent across $\ell=1, \ldots , k$, with distribution 
           \begin{equation}
           \label{eq:fell}
               H_{\ell q} (\dd \theta_1\, \ldots \,\dd \theta_d) \propto  \prod_{q=1}^d \prod_{i=1}^{n_q}  P(a_{i\ell q}; \theta_q) \cdot \rho_d (\dd \theta_1\, \ldots \, \dd \theta_d);
           \end{equation}
           \item[(ii)] the process $(\mu_1^\prime, \ldots, \mu_d^\prime)$ is a $\textsc{crv}(\rho_d^\prime, \lambda, P_0)$, with
           \[
           \rho_d^\prime (\dd \theta_1\, \ldots \,\dd \theta_d) =  \prod_{q=1}^d \prod_{i=1}^{n_q} P(0; \theta_q ) \cdot \rho_d (\dd \theta_1\, \ldots \,\dd \theta_d).
           \] 
       \end{itemize}
\end{theorem}
\begin{proof}
In order to characterize the posterior distribution of $\tilde{\bm{\mu}} = (\tilde{\mu}_1,\ldots,\tilde{\mu}_d)$, we show that its Laplace functional coincides with the one of the sum on the right-hand side of \eqref{eq:posterior_CRV}. To this end, consider $d$ measurable functions $g_q: \X  \to (0,\infty)$, as $q=1, \ldots , d$ and  focus on the Laplace functional of the vector $\tilde{\bm{\mu}}$, conditionally to $\bm Z$, 
\begin{equation*}
    \calL_{\tilde{\bm{\mu}}\mid \bm Z}(g_1,\ldots,g_d) = \E\left[\exp \left\{- \sum_{q=1}^d \int_\X g_q (x) \tilde{\mu}_q (\dd x)  \right\}  \,\mid \bm Z \right].
\end{equation*}
This conditional Laplace functional may be evaluated as 
\begin{equation}\label{eq_post_laplace_def_generic}
  \calL_{\tilde{\bm{\mu}}\mid \bm Z}(g_1,\ldots,g_d) = \lim_{\varepsilon\, \downarrow\, 0}
  \frac{1 }{\Pp(\Ecr)}\E\left[ \exp \left\{- \sum_{q=1}^d \int_\X g_q (x) \tilde{\mu}_q (\dd x)  \right\} \cdot \indicator_\Ecr \right], 
\end{equation}
where $\Ecr$ is the event defined in \eqref{eq:ECR_DEF}, for proof of Theorem \ref{thm:marginal_CRV}, which depends on  $\varepsilon$, and $\indicator_\Ecr$ denotes the corresponding indicator function. The denominator in \eqref{eq_post_laplace_def_generic} has already been computed in the proof of Theorem \ref{thm:marginal_CRV}; a similar argument may be applied to find an expression for the expected value in \eqref{eq_post_laplace_def_generic}. 
First of all, the tower property of the conditional expectation implies
\begin{equation*}
\begin{split}
    & \E\left[ \exp \left\{- \sum_{q=1}^d \int_\X g_q (x) \tilde{\mu}_q (\dd x)  \right\} \cdot \indicator_\Ecr \right] =  \E\left[ \E \left[ \exp \left\{- \sum_{q=1}^d \int_\X g_q (x) \tilde{\mu}_q (\dd x)  \right\} \cdot \indicator_\Ecr \,\mid\, \tilde{\bm\mu} \right]\right] \\
    &\qquad = \E\left[ \exp \left\{- \sum_{q=1}^d \int_\X g_q (x) \tilde{\mu}_q (\dd x)  \right\} \cdot \Pp(\Ecr\mid \tilde{\bm\mu}) \right].
 \end{split}
\end{equation*}
The probability $\Pp(\Ecr\mid \tilde{\bm\mu})$ has already been computed in the proof of Theorem \ref{thm:marginal_CRV}, while the exponential term can be explicitly written as a function of the a.s. discrete random measures $\tilde{\mu}_q$, which leads to
\begin{equation*}
    \begin{split}
       & \E\left[ \exp \left\{- \sum_{q=1}^d \int_\X g_q (x) \tilde{\mu}_q (\dd x)  \right\} \cdot \indicator_\Ecr \right] \\
       & \qquad = \E \left[ \prod_{q=1}^d  \prod_{j\geq 1}   e^{-\theta_{jq} g_q (\tilde{X}_j)}  \cdot \prod_{q=1}^d  \prod_{\ell = 1}^k
    \prod_{i \not \in  \calB_{\ell q} } \prod_{j\geq 1}   P(0; \theta_{jq})^{ \delta_{\tilde{X}_j} (B_\varepsilon (X_\ell))}  \right.\\
    & \left.\quad\quad \qquad \times \prod_{q=1}^d   \prod_{\ell = 1}^k  \prod_{i \in \calB_{\ell q}} \Big( 1- \prod_{j\geq 1} (1-P(a_{i\ell q}; \theta_{jq}))^{\delta_{\tilde{X}_j} (B_\varepsilon (X_\ell))}   \Big) \cdot \prod_{q=1}^d \prod_{j\geq 1} \prod_{i=1}^{n_q}  
    P(0; \theta_{jq})^{\delta_{\tilde{X}_j} (\X^*)}   \right].  
    \end{split}
\end{equation*}
Along similar lines as in  the proof of Theorem \ref{thm:marginal_CRV}, we end up with:
\begin{equation*}
\begin{split}
   & \E\left[ \exp \left\{- \sum_{q=1}^d \int_\X g_q (x) \tilde{\mu}_q (\dd x)  \right\} \cdot \indicator_\Ecr \right] \\
    & \qquad =
   \exp \left\{  -\int_{\X^*}  \int_{(0,\infty)^d}  \left( 1-\prod_{q=1}^d \left( e^{-\theta_q g_q(x)} \prod_{i=1}^{n_q}P(0; \theta_q)\right)  \right)    \rho_d (\dd \theta_1\, \ldots \,\dd \theta_d)  \,  \lambda P_0 (\dd x)  \right\} \\
    &\quad \qquad \times  \prod_{\ell=1}^k \left\{ \int_{B_\varepsilon (X_\ell)} \int_{(0,\infty)^d}  \prod_{q=1}^d  \Big( e^{-\theta_q g_q(x)} \prod_{i \not\in \calB_{\ell q}} P(0; \theta_q) \right.  \\
  & \qquad\qquad\qquad\times \left.  \prod_{i \in \calB_{\ell q}} P(a_{i\ell q}; \theta_q) \Big)
    \rho_d (\dd \theta_1\, \ldots \,\dd \theta_d) \, \lambda P_0(\dd x) \right\}  +o \Big( \prod_{\ell=1}^k P_0 (B_\varepsilon (X_\ell)) \Big).
\end{split}
\end{equation*}
By using the just derived  expression and $\Pp (\Ecr)$, available in the proof of Theorem \ref{thm:marginal_CRV}, we can compute the limit in \eqref{eq_post_laplace_def_generic}. Since $P_0$ is non-atomic, standard limiting arguments lead to
\begin{equation}\label{eq_post_laplace_functional_generic}
\begin{split}
      & \calL_{\tilde{\bm{\mu}}\mid \bm Z}(g_1,\ldots,g_d) \\
       &\qquad = \exp \left\{  -  \int_{(0,\infty)^d \times \X}  \left( 1-\prod_{q=1}^d  e^{-\theta_q g_q(x)} \right) 
       \rho_d^\prime (\dd \theta_1\, \ldots \, \dd \theta_d) \, \lambda P_0 (\dd x)  \right\}\\
       & \qquad \quad \times \prod_{\ell=1}^k 
            \int_{(0,\infty)^d}  \prod_{q=1}^d   e^{-\theta_q g_q(X_\ell)} \cdot H_{\ell q} (\dd \theta_1\, \ldots \,\dd \theta_d),
\end{split}
\end{equation}
where $\rho_d^\prime (\cdot)$ and $H_{\ell q}(\cdot)$ have been defined in the statement of the present theorem.
First, the exponential term in \eqref{eq_post_laplace_functional_generic} corresponds to the Laplace functional of a \textsc{crv}, say $\bm \mu^\prime$, having L\'evy intensity measure given by
\[
\nu^\prime (\dd \theta_1 \, \ldots\, \dd \theta_d \, \dd x) =
\prod_{q=1}^d \prod_{i=1}^{n_q} P(0; \theta_q)\cdot \rho_d (\dd \theta_1\, \ldots \, \dd \theta_d)   \, \lambda P_0 (\dd x).
\]
Second, the product over $\ell=1, \ldots , k$ in  \eqref{eq_post_laplace_functional_generic} corresponds 
to the Laplace functional of the vector of random measures
\[
\bm \mu^* (\cdot) = \sum_{\ell=1}^k  (\theta_{\ell 1}^*, \ldots , \theta_{\ell d}^*) \, \delta_{X_\ell} (\cdot),
\]
where the vectors $(\theta_{\ell 1}^*, \ldots , \theta_{\ell d}^*)$ are independent across $\ell =1, \ldots , k$, with distribution $H_{\ell q}(\cdot)$.
Thus, the thesis follows.
\end{proof}


For completeness, we finally provide the predictive distribution of a vector of new observations $(Z_{(n_1+1)1}, \ldots , Z_{(n_d+1)d})$, conditionally to $\bm{Z}$. 
\begin{theorem}[Predictive distribution]
    \label{thm:predictive_CRV}
    Let $\bm{Z}$ be a sample from the statistical model~\eqref{eq:mixture_model_allocation_CRV}. If $\bm{Z}$ displays $K_n = k$ distinct traits labeled $X_1, \ldots , X_k$, with associated counts $\bm{a}$, then the predictive distribution of 
    $(Z_{(n_1+1)1}, \ldots , Z_{(n_d+1)d})$ satisfies the distributional equality
    \begin{equation}
        \label{eq:predictive_CRV}
        (Z_{(n_1+1)1}, \ldots , Z_{(n_d+1)d}) \mid \bm{Z} \stackrel{d}{=} (Z_{(n_1+1)1}^*, \ldots , Z_{(n_d+1)d}^*) + (Z_{(n_1+1)1}^\prime, \ldots , Z_{(n_d+1)d}^\prime),
    \end{equation}
    where the vectors on the right-hand side are independent, and in addition:
    \begin{itemize}
        \item[(i)] the generic component $Z_{(n_q+1)q}^*$ is defined as $Z_{(n_q+1)q}^* (\cdot)= \sum_{\ell =1}^{k} A^*_{(n_q+1)\ell q} \delta_{X_\ell}(\cdot)$, for $q=1, \ldots, d$, and the random variables $A^*_{(n_q+1)\ell q}$ are independent with distribution $P(\dd a ; \theta_{\ell q}^*)$, where the $\theta_{\ell q}^*$ are given in point (i) of Theorem \ref{thm:posterior_CRV};
        \item[(ii)] the generic component $Z_{(n_q+1)q}^\prime$ is defined as $Z_{(n_q+1)q}^\prime \mid \bm{\mu}^\prime \ind \CP (\mu_q^\prime)$, for $q=1, \ldots , d$, where $\bm{\mu}'$ is the \textsc{fcrv} defined in point (ii) of Theorem \ref{thm:posterior_CRV}.
    \end{itemize}
\end{theorem}
\begin{proof}
The thesis follows by a straightforward application of Bayes formula and Theorem \ref{thm:posterior_CRV}.
\end{proof}

To conclude this section, we provide some details about how Theorems \ref{thm:marginal_fCRV}-\ref{thm:posterior_fCRV} follow as simple corollaries of Theorems \ref{thm:marginal_CRV}-\ref{thm:posterior_CRV}, respectively. Indeed, as already commented in Section \ref{app:fCRV}, the \textsc{fcrv}s analyzed in the main paper are special examples of \textsc{crv}s such that $\rho_d = H^{(d)}$, where $H^{(d)}(\cdot)$ is a probability distribution on $(0,\infty)^d$. In the particular case of Theorems \ref{thm:marginal_fCRV}-\ref{thm:posterior_fCRV}, we focus on the special choice where $H^{(d)}(\cdot\,;\psi) = H(\cdot\,;\psi) \times \cdots \times H(\cdot\,; \psi)$, with $H(\cdot\,; \psi)$ any probability distribution on $(0,\infty)$.
Thus, Theorems \ref{thm:marginal_fCRV}-\ref{thm:posterior_fCRV} follow by specializing Theorems \ref{thm:marginal_CRV}-\ref{thm:posterior_CRV} to this situation.

\begin{remark}
In addition to the example of \textsc{fcrv}s discussed in the main paper, there are many other tractable classes of \textsc{crv}s that can be analyzed based on the general theory of the present section, and that are of potential interest in applications. For example, Theorems \ref{thm:marginal_CRV}-\ref{thm:posterior_CRV}-\ref{thm:predictive_CRV} may be applied to obtain marginal, posterior, and predictive distributions of additive \textsc{crm}s introduced by \cite{Nip(14)}. Another example of interest that can be addressed with the provided theory relates to compound random measures \citep{Gri(17)}.
\end{remark}

\subsection{Extension of the modeling framework to general parameter space}\label{app:ext_genereral_space}


As highlighted in Remark \ref{rmk:single_parameter_count}, throughout the paper we assume that the parametric distribution $P(\cdot\,; \theta)$ is governed by a single positive parameter $\theta > 0$. This is related to the use of the technical tool of \textsc{crv}s for the general proofs in Section \ref{app:distribution_theory_CRV}. However, this assumption can be relaxed with basically no additional cost, allowing for a general parameter space $\mathbb{S}$, just by moving from \textsc{crv}s to Poisson processes. To be rigorous, this shift only calls for a slight change of notation, which we describe next. Starting from the exchangeable trait allocation model expressed by \eqref{eq:Zi}-\eqref{eq:cond_prob_counts}-\eqref{eq:mui_def}, it is here convenient to organize the parameters $(\theta_j)_{j \ge 1}$ in a point process $\tilde{\Psi}(\cdot) = \sum_{j \geq 1} \delta_{(\tilde{X}_j, \theta_j)} (\cdot)$ on $\X\times \mathbb{S}$, i.e., a random (locally finite) counting measure on $\X\times \mathbb{S}$, with $\mathbb{S}$ a Polish space. Remarkably, note that $\tilde{\Psi}$ and the discrete measure $\tilde{\mu}$ in \eqref{eq:mui_def} contain the exact same information. In the partially exchangeable trait allocation models in \eqref{eq:partially_ex_traits}, which are the main focus of the general theory, the group-specific parameters of interest can be similarly collected in the point processes $\tilde{\Psi}_q (\cdot) = \sum_{j \geq 1} \delta_{(\tilde{X}_j, \theta_{jq})} (\cdot)$, for $q=1,\ldots,d$, instead of using the discrete measures $\tilde{\mu}_q$. With a slight abuse of notation, we can generalize model \eqref{eq:partially_ex_traits} by writing
\begin{equation}    \label{eq:partially_ex_traits_point_process}
\begin{aligned}
    Z_{iq} \mid \tilde{\mu}_q &\ind \CP (\tilde{\Psi}_q), \qquad i\geq 1, \quad q=1,\ldots,d,\\
    (\tilde{\Psi}_1,\ldots,\tilde{\Psi}_d) &\sim \mathcal{Q}_d,
\end{aligned}
\end{equation}
where $\mathcal{Q}_d$ denotes here the de Finetti measure of the vector of point processes $(\tilde{\Psi}_1,\ldots,\tilde{\Psi}_d)$. The generalization with respect to model \eqref{eq:partially_ex_traits} is only in that $\theta_{jq} \in \mathbb{S}$ instead of $\theta_{jq} > 0$. With a slight abuse of notation, we will use without distinction the vector $(\tilde{\Psi}_1,\ldots,\tilde{\Psi}_d)$ of point processes on $\X\times \mathbb{S}$ and the point process $\tilde{\bm\Psi} = \sum_{j \geq 1} \delta_{(\tilde{X}_j, \theta_{j1},\ldots, \theta_{jd})}$ on $\X \times \mathbb{S}^d$. The natural choice for the prior distribution of $\tilde{\bm\Psi}$ in this setting is the class of Poisson point processes on $\X \times \mathbb{S}^d$. Indeed, in the special case where $\mathbb{S} = (0,\infty)$, the class of Poisson point processes induces the class of \textsc{crv} priors for the vector of random measures $\tilde{\bm\mu}$ in \eqref{eq:mixture_model_allocation_CRV}. Therefore, model \eqref{eq:partially_ex_traits_point_process} with $\tilde{\bm\Psi}$ distributed as a Poisson point process on $\X \times \mathbb{S}^d$ generalizes the main model \eqref{eq:mixture_model_allocation_CRV}, for a general parameter space $\mathbb{S}$.

A point process $\tilde{\bm\Psi}$ on $\X \times \mathbb{S}^d$, like any random measure, is uniquely characterized by its Laplace functional 
\[
\mathcal{L}_{\tilde{\bm\Psi}}(f):= \E \left[\exp \left\{ - \int_{\X \times \mathbb{S}^d} f (x, \theta_1,\ldots,\theta_d) \tilde{\bm\Psi}(\dd x\, \dd \theta_1\,\ldots\, \dd \theta_d)\right\}\right], 
\]
for any measurable function $f: \X \times \mathbb{S}^d \to (0,\infty)$. The class of Poisson point processes is characterized by the following representation of the Laplace functional,
\begin{equation*}
   \mathcal{L}_{\tilde{\bm\Psi}}(f) = 
    \exp \left\{ - \int_{\X\times \mathbb{S}^d} 
    \left(1-\exp\left\{- f(x, \theta_1,\ldots,\theta_d) \right\} \right)     \nu (\dd x\, \dd \theta_1 \, \ldots\, \dd \theta_d )     \right\},
\end{equation*}
where $\nu (\, \cdot \, )$ is a locally finite measure on $\X\times \mathbb{S}^d$ referred to as the intensity of $\tilde{\bm{\Psi}}$. To link with the general theory and notation of \textsc{crv}s in Section \ref{app:distribution_theory_CRV}, assume that the intensity measure factorizes as $\nu (\dd x\, \dd \theta_1\, \ldots \, \dd \theta_d )=
\lambda \rho_d (\dd \theta_1\, \ldots \, \dd \theta_d) P_0(\dd x)$, where $\lambda>0$ and $P_0(\, \cdot\,)$ is a non-atomic probability measure on $\X$. We write $\tilde{\bm{\Psi}} \sim \textsc{pp}(\rho_d, \lambda, P_0)$. Then, 
Theorems \ref{thm:marginal_CRV}-\ref{thm:posterior_CRV}-\ref{thm:predictive_CRV} hold identically with the only replacement of $\mathbb{S}$ instead of $(0,\infty)$ as parameter space. For completeness, we report here the most general formulation of Theorems \ref{thm:marginal_CRV}-\ref{thm:posterior_CRV} for the general parameter space $\mathbb{S}$.

\begin{theorem}[Marginal distribution, p\textsc{etpf}]
    Let $\bm{Z}$ be a sample from the statistical model~\eqref{eq:partially_ex_traits_point_process}, with $\tilde{\bm{\Psi}} \sim \textsc{pp}(\rho_d, \lambda, P_0)$. The probability that  $\bm{Z}$ displays $K_n = k$ distinct traits with counts $\bm A = \bm{a}$ equals
    \begin{equation*}
        \begin{split}
       \pi_n (\bm{a}) & =\frac{{\lambda}^k}{k!} \exp \left\{  -\lambda  \int_{\mathbb{S}^d}  \left( 1- \prod_{q=1}^d  P(0; \theta_q)^{n_q} \right)    \rho_d (\dd \theta_1\, \ldots \,\dd \theta_d)    \right\} \\
    & \qquad\times  \prod_{\ell=1}^k \int_{\mathbb{S}^d}  
    \prod_{q=1}^d   \prod_{i=1}^{n_q} P (a_{i\ell q} ; \theta_q) 
    \rho_d (\dd \theta_1\, \ldots \,\dd \theta_d)
        \end{split}
    \end{equation*}
    where $n=\sum_{q=1}^d n_q$ and $\bm{n}=(n_1, \ldots , n_d)$ are the sample sizes.
\end{theorem}

\begin{theorem}[Posterior distribution]\label{thm:posterior_pp}
       Let $\bm{Z}$ be a sample from the statistical model~\eqref{eq:partially_ex_traits_point_process}, with $\tilde{\bm{\Psi}} \sim \textsc{pp}(\rho_d, \lambda, P_0)$. If $\bm{Z}$ displays $K_n = k$ distinct traits labeled $X_1, \ldots , X_k$, with associated counts $\bm{a}$, then the posterior distribution of 
       $\tilde{\bm \Psi}$ satisfies the distributional equality        \begin{equation*}
            (\tilde{\Psi}_1, \ldots, \tilde{\Psi}_d) \mid \bm{Z} \stackrel{d}{=}  (\Psi_1^*, \ldots, \Psi_d^*) +  (\Psi_1', \ldots, \Psi_d'),
       \end{equation*}
       where $(\Psi_1^*, \ldots, \Psi_d^*)$ and $(\Psi_1', \ldots, \Psi_d')$ are independent random vectors such that
       \begin{itemize}
           \item[(i)] the generic $\Psi_q^*(\cdot)$ is  defined as
           $ \Psi_q^*(\cdot) = \sum_{\ell=1}^k  \delta_{(X_\ell, \theta_{\ell q}^*)}(\cdot)$, for $ q=1, \ldots , d
           $, and the random vectors $(\theta_{\ell 1}^* , \ldots , \theta_{\ell d}^*)$ are independent across $\ell=1, \ldots , k$, with distribution 
           \begin{equation*}
               H_{\ell q} (\dd \theta_1\, \ldots \,\dd \theta_d) \propto  \prod_{q=1}^d \prod_{i=1}^{n_q}  P(a_{i\ell q}; \theta_q) \cdot \rho_d (\dd \theta_1\, \ldots \, \dd \theta_d);
           \end{equation*}
           \item[(ii)] the vector of point processes $(\Psi_1^\prime, \ldots, \Psi_d^\prime)$ is a $\textsc{pp}(\rho_d^\prime, \lambda, P_0)$, with
           \[
           \rho_d^\prime (\dd \theta_1\, \ldots \,\dd \theta_d) =  \prod_{q=1}^d \prod_{i=1}^{n_q} P(0; \theta_q ) \cdot \rho_d (\dd \theta_1\, \ldots \,\dd \theta_d).
           \] 
       \end{itemize}
\end{theorem}

The computations we followed in Section \ref{app:distribution_theory_CRV} to prove Theorems \ref{thm:marginal_CRV}-\ref{thm:posterior_CRV} hold for the more general model \eqref{eq:partially_ex_traits_point_process}, with $\tilde{\bm{\Psi}} \sim \textsc{pp}(\rho_d, \lambda, P_0)$, with trivial modifications such as the use of Laplace functionals of Poisson point processes instead of the ones of \textsc{crv}s. The finite-dimensional trait models analyzed in Section \ref{sec:partially_ex_setting}, corresponding to $\tilde{\bm{\mu}} \sim \textsc{fcrv}(H^{(d)}, \lambda, P_0)$, are generalized to a generic parameter space via finite Poisson point processes, i.e., $\tilde{\bm{\Psi}} \sim \textsc{pp}(\rho_d, \lambda, P_0)$ where $\rho_d = H^{(d)}$ and $H^{(d)}$ is a probability distribution on $\mathbb{S}^d$. As in Section \ref{sec:partially_ex_setting}, we commonly assume $H^{(d)}(\cdot\,;\psi) = H(\cdot\,;\psi) \times \cdots \times H(\cdot\,;\psi)$, with $H(\cdot\,;\psi)$ a probability distribution on $\mathbb{S}$, parameterized by $\psi$.

To conclude this discussion on the  extension of the proposed modeling framework to general parameter space $\mathbb{S}$ for the count distribution $P(\cdot; \theta)$, $\theta \in \mathbb{S}$, we provide an example. Besides specifying $P(\cdot; \theta)$, we derive the marginal of a sample $\bm Z$ and the posterior distribution of $\tilde{\bm\Psi}$, assuming $\tilde{\bm{\Psi}} \sim \textsc{pp}(H^{(d)}, \lambda, P_0)$, and   $H^{(d)}(\cdot\,;\psi) = H(\cdot\,;\psi) \times \cdots \times H(\cdot\,;\psi)$. This results are derived by specializing the previous general theorems.

\begin{example}[Zero-inflated shifted negative binomial counts]
    
In the setting of count data with support on $\{0, 1, 2, \dots \}$, we propose to consider a zero-inflated shifted negative binomial (\textsc{zi-snb}) distribution. More precisely, let $X$ be a negative binomial random variable with parameters $c > 0$ (number of trials) and $p \in (0,1]$ (success probability). 
We say that  $Y:= X+1$ has a  shifted negative binomial random distribution, whose probability mass function $p_{\mathrm{sNB}} (\, \cdot \, ; c, p)$, depending on the parameters $(c,p)$, is supported by the set of  natural numbers. Then, we specify the distribution $P (\, \cdot \,; \theta)$ as $P (\, \cdot \,; \theta) = (1-w) \delta_0 (\, \cdot \, ) + w\, p_{\mathrm{sNB}}(\, \cdot \, ; c, p)$, with $\theta = (c,w,p)$ and $w \in (0,1) $. This choice allows to independently model the presence of a trait and its associated measurement. That is, the pattern of occurrence of traits is captured as in the binary traits setting. On top of that, when a trait is displayed, the shifted negative binomial distribution models the magnitude of the associated expression. The shifted version of the negative binomial is considered to correctly identify the actual expression of a trait.

To complete the specification of this model, we consider $H(\cdot\,;\psi)$ such that $c$ is fixed, with $w$ and $p$ following independent beta laws with parameters $(a_w,b_w)$ and $(a_p, b_p)$, respectively, so that $\psi = (c,a_w,b_w, a_p, b_p)$. In this case, the marginal distribution of $\bm Z$ depends on the whole collection of counts $\bm a$ and equals
\begin{equation}\label{eq:pEPPF_NB}
\begin{split}
    \pi_n (\bm{a}; \lambda, \psi) & = \frac{\lambda^k}{k!} \exp \left\{  - \lambda \left[  1- \prod_{q=1}^d \frac{B(a_w, b_w + n_q)}{B(a_w,b_w)}  \right]  \right\} \\
    &\qquad \times \prod_{\ell=1}^k \prod_{q=1}^d  \prod_{i: a_{i\ell q} > 1} \binom{a_{i\ell q} + c - 2}{a_{i\ell q} -1} \frac{b_{\ell q}^{(p)}}{B(a_p, b_p)} \frac{b_{\ell q}^{(w)}}{B(a_w,b_w)},
\end{split}
\end{equation}
where $b_{\ell q}^{(p)} = B(a_p + c\, m_{\ell q}, b_p + \sum_{i=1}^{n_q} a_{i\ell q} - m_{\ell q})$ and $b_{\ell q}^{(w)} = B(a_w + m_{\ell q}, b_w + n_q - m_{\ell q})$.

Moving to the posterior distribution of $\tilde{\bm\Psi}$,  each $\theta^*_{\ell q}$ in Theorem \ref{thm:posterior_pp} is a vector of three component $(c, w^*_{\ell q}, p^*_{\ell q})$, where $w^*_{\ell q}$ has a beta distribution  with parameters $(a_w + m_{\ell q}, b_w + n_q - m_{\ell q})$, and $p^*_{\ell q}$ is again a  beta with parameters $(a_p + c\, m_{\ell q}, b_p + \sum_{i=1}^{n_q} a_{i\ell q} - m_{\ell q})$, further independent between them. In addition, for each $q=1,\ldots,d$,
\begin{equation*} 
           \Psi^\prime_q(\cdot) = \sum_{j=1}^{N^\prime}  \delta_{(\tilde{X}_j^\prime, \theta_{j q}^\prime)}(\cdot), \qquad \theta^\prime_{jq} \iid H^\prime_{q}, \qquad \tilde{X}^\prime_j \iid P_0, \qquad j=1,\ldots,N^\prime, 
\end{equation*}
where $H^\prime_{q}$ is such that $c$ is fixed, $w^\prime_{jq}$ has a beta distribution with parameters $(a_w, b_w + n_q )$, independent of $p^\prime_{jq}$, which follows the prior beta law. Moreover, $N^\prime \sim \mathrm{Poisson}(\lambda^\prime)$ with
\[
\lambda^\prime = \lambda\prod_{q=1}^d B(a_w, b_w + n_q) / B (a_w,b_w) . 
\]  
\end{example}

\section{Proofs of Section \ref{sec:learning_clustering}}
\label{app:effect_unseen}

The main theoretical result of Section \ref{sec:learning_clustering} is Proposition \ref{prop:cluster_comparison}, which illustrates the effect on clustering estimation of accounting for potentially unseen traits by comparing the proposed model  in~\eqref{eq:mixture_model2} with and the na\"ive model almost identical to \eqref{eq:mixture_model2}, in which it is assumed that there are no unseen traits, i.e. we suppose $N = k$. Before showing the proof of Proposition \ref{prop:cluster_comparison}, we recall a lemma which is instrumental for it.

\begin{lemma}\label{lemma:inequality}
    Let $X$ be an almost surely non-negative random variable and $n = 1,2,\ldots$. The following holds true:
    \[
    \E(X^{n+1}) \geq \E(X^{n}) \E(X).
    \]
\end{lemma}
\begin{proof}
For any $n = 1,2,\ldots$, by Holder inequality, it holds that
\[
 \E(X^{n}) \leq \E(X^{n+1})^{\frac{n}{n+1}} = \E(X^{n+1})^{1 - \frac{1}{n+1}},
\]
and consequently
\begin{equation}\label{eq:holder_ineq}
    \E(X^{n}) \E(X^{n+1})^{\frac{1}{n+1}} \leq \E(X^{n+1}).
\end{equation}
By Jensen inequality, 
\[
\E(X^{n+1}) \geq \E(X)^{n+1}
\]
and, from \eqref{eq:holder_ineq}, we get
\[
\E(X^{n+1}) \geq \E(X^{n}) \E(X^{n+1})^{\frac{1}{n+1}} \geq  \E(X^{n}) \E(X).
\]
The thesis is proven.
\end{proof}

\subsection{Proof of Proposition \ref{prop:cluster_comparison}}\label{proof:cluster_comparison}

In Proposition \ref{prop:cluster_comparison}, we compare the predictive allocation probabilities for a generic subject $i$ under the proposed model defined in~\eqref{eq:mixture_model2}, with a Pitman-Yor prior for $\xi_h$, and a na\"ive model almost identical to \eqref{eq:mixture_model2}, in which it is assumed that there are no unseen traits, i.e. we suppose $N = k$. Recalling the notation used in the statement, let 
\begin{equation*}
p_{iq} = \mathds{P}(\mu_i = \tilde{\mu}_q \mid \bm{Z}, \bm{\mu}_{-i}), \qquad p_{i,\text{new}} = \mathds{P}(\mu_i = \text{``new''} \mid \bm{Z}, \bm{\mu}_{-i}), \qquad i=1,\dots,n,
\end{equation*}
where $\bm{\mu}_{-i} = (\mu_1,\dots\mu_{i-1},\mu_{i+1},\dots,\mu_n)$ with $d_{-i}$ distinct values $\tilde{\mu}_1,\dots,\tilde{\mu}_{d_{-i}}$, for $q = 1,\dots,d_{-i}$. Here, $p_{iq}$ denotes the predictive probability that subject $i$ belongs to the $q$th cluster formed by the remaining subjects, and $p_{i,\text{new}}$ is the predictive probability for subject $i$ to form its own cluster. Similarly, we define $p_{iq}^*$ and $p_{i,\text{new}}^*$ for the na\"ive model.

For the proposed model defined in~\eqref{eq:mixture_model2}, the predictive allocation probabilities $p_{iq}$ and $p_{i,\text{new}}$ are expressed in \eqref{eq:update_clustering}. For the associated na\"ive model, since this is identical to \eqref{eq:mixture_model2} with the specification $N = k$, the predictive allocation probabilities $p_{iq}^*$ and $p_{i,\text{new}}^*$ are obtained from \eqref{eq:update_clustering} by replacing $\pi_n(\bm a; \lambda, \psi)$ with $\pi_n(\bm a; N = k, \psi)$ in \eqref{eq:conditional_ETFP}. 

Consider the ratio between the probability that subject $i$ creates a new cluster and the probability that it is allocated to the $q$-th existing cluster. For the proposed model defined in~\eqref{eq:mixture_model2}, this is equal to
\[
\frac{p_{i,\text{new}}}{p_{iq}} = \frac{\pi_n(\bm{a}_{i,\textup{new}}; \lambda, \psi) (\gamma + d_{-i}\sigma)/(n+\gamma-1)}{\pi_n(\bm{a}_{iq}; \lambda, \psi) (n_{q,-i} - \sigma)/(n + \gamma - 1)},
\]
while the same quantity, for the na\"ive model, is equal to 
\[
\frac{p^*_{i,\text{new}}}{p^*_{iq}} = \frac{\pi_n(\bm{a}_{i,\textup{new}}; N=k, \psi) (\gamma + d_{-i}\sigma)/(n+\gamma-1)}{\pi_n(\bm{a}_{iq}; N=k, \psi) (n_{q,-i} - \sigma)/(n + \gamma - 1)}.
\]
Denote with ${\bm n}_{i,\text{new}}$ the sample sizes associated with $\bm{a}_{i,\textup{new}}$, and with ${\bm n}_{iq}$ the sample sizes associated with $\bm{a}_{iq}$. By comparing the previous two ratios, we obtain
\begin{equation*}
    \begin{aligned}
        \frac{p_{i,\text{new}}/p_{iq}}{p^*_{i,\text{new}}/p^*_{iq}} &= \frac{\pi_n(\bm{a}_{i,\textup{new}}; \lambda, \psi) }{\pi_n(\bm{a}_{iq}; \lambda, \psi) } \cdot \frac{\pi_n(\bm{a}_{iq}; N=k, \psi) }{ \pi_n(\bm{a}_{i,\textup{new}}; N=k, \psi)}\\
        &= \frac{\exp \left\{ -\lambda \left( 1-\prod_{h=1}^{d_{-i}+1}  \int_{(0,\infty)}  P (0 ; \theta)^{({\bm n}_{i,\text{new}})_h}
        H (\dd \theta; \psi) \right) \right\}}{\exp \left\{ -\lambda \left( 1-\prod_{h=1}^{d_{-i}}  \int_{(0,\infty)} P (0; \theta)^{({\bm n}_{iq})_h}
        H (\dd \theta; \psi) \right) \right\}}\\
        &= \exp\left\{ -\lambda \left[ \prod_{h=1}^{d_{-i}}  \int_{(0,\infty)}  P (0 ; \theta)^{({\bm n}_{iq})_h}
        H (\dd \theta; \psi) - \prod_{h=1}^{d_{-i}+1}  \int_{(0,\infty)}  P (0 ; \theta)^{({\bm n}_{i,\text{new}})_h}
        H(\dd \theta; \psi) \right] \right\}\\
        &= \exp\left\{ -\lambda \prod\limits_{\substack{h=1\\ h\neq q}}^{d_{-i}} \int_{(0,\infty)} P(0; \theta)^{n_{h,-i}} H(\dd \theta; \psi)  \right. \\
        &\qquad \qquad \left. \times \left[  \int_{(0,\infty)}  P (0; \theta)^{n_{q,-i} +1}
        H (\dd \theta; \psi) - \int_{(0,\infty)}  P (0; \theta)
        H (\dd \theta; \psi) \int_{(0,\infty)}  P (0; \theta)^{n_{q,-i}}
        H (\dd \theta; \psi) \right] \right\}.
    \end{aligned}
\end{equation*}
Using the inequality in Lemma \ref{lemma:inequality}, it holds that
\[
 \int_{(0,\infty)}  P (0; \theta)^{n_{q,-i} +1}
        H (\dd \theta; \psi) - \int_{(0,\infty)}  P (0; \theta)
        H (\dd \theta; \psi) \int_{(0,\infty)}  P (0; \theta)^{n_{q,-i}}
        H (\dd \theta; \psi) \geq 0,
\]
and consequently
\[
 \frac{p_{i,\text{new}}}{p_{iq}} < \frac{p^*_{i,\text{new}}}{p^*_{iq}}.
\]
The proof is complete. Note that the result is not restricted to the Pitman-Yor prior for $\xi_h$; the same argument holds for any arbitrary prior on $\xi_h$.

\section{Simulation studies}
\label{sec:simulations}

\subsection{Assessing the impact of unseen traits in clustering}\label{sec:sim_k_fixed}

Before moving to more structured simulation studies, we first provide empirical evidence on the discrepancy highlighted in Proposition~\ref{prop:cluster_comparison}, which may indeed lead to substantial inferential differences. To this end, we consider a simulated binary-outcomes example, where we compare the inference obtained from our proposed model \eqref{eq:mixture_model} with that of a latent class model in which $N = k$ is fixed, thus disregarding the possibility of unseen traits, as in \citet{Dunson2009}. We consider data organized into $d = 5$ groups, with a total of $N = 500$ traits that may potentially be observed. Group-specific probability vectors $\theta_{jq}$ are generated by drawing iid values from a $\text{Beta}(0.1, 10)$ distribution, for $q = 1,\ldots, 5$ and $j = 1,\ldots,500$. Subject-specific binary vectors are then obtained by independently sampling Bernoulli random variables $\tilde{A}_{ijq}$ with success probabilities $\theta_{jq}$, determined by the group allocation. Samples are drawn from the five groups with different sizes, namely $n_1 = 100$, $n_2 = 60$, $n_3 = 40$, $n_4 = 20$, and $n_5 = 20$, where $n_q$ denotes the sample size of group $q$. In the resulting dataset, $K_n = k = 293$ traits are observed out of the total $N = 500$. 

We compare model \eqref{eq:mixture_model} with binary traits, as in Example~\ref{ex:binary_marginal}, with the na\"ive model that disregards unseen traits. To ensure a fair comparison, both models are fitted using identical hyperparameters, set to their oracle values. Specifically, the parameters of the Beta prior $H(\cdot;\psi)$ are fixed at $\psi = (0.1, 10)$, while the Poisson prior on $N$ is specified with $\lambda = 500$. For the clustering prior, we adopt a Dirichlet process with concentration parameter $1$, corresponding to a Pitman–Yor process with parameters $\sigma = 0$ and $\gamma = 1$. The \textsc{mcmc} algorithm is run for $5{,}000$ iterations, discarding the first $500$ as burn-in, and applying thinning every $2$ iterations. Figure~\ref{fig:sim_k_fixed_num_clusters} displays the posterior distribution of the number of clusters under the two models. Consistent with the analytical result in Proposition~\ref{prop:cluster_comparison}, our proposed model \eqref{eq:mixture_model} yields fewer clusters than the naïve specification that ignores unseen traits. The discrepancy is substantial in this example, with posterior modes equal to $6$ and $10$, respectively.

\begin{figure}[tbp]
    \centering    
    \includegraphics[width = 0.4\linewidth]{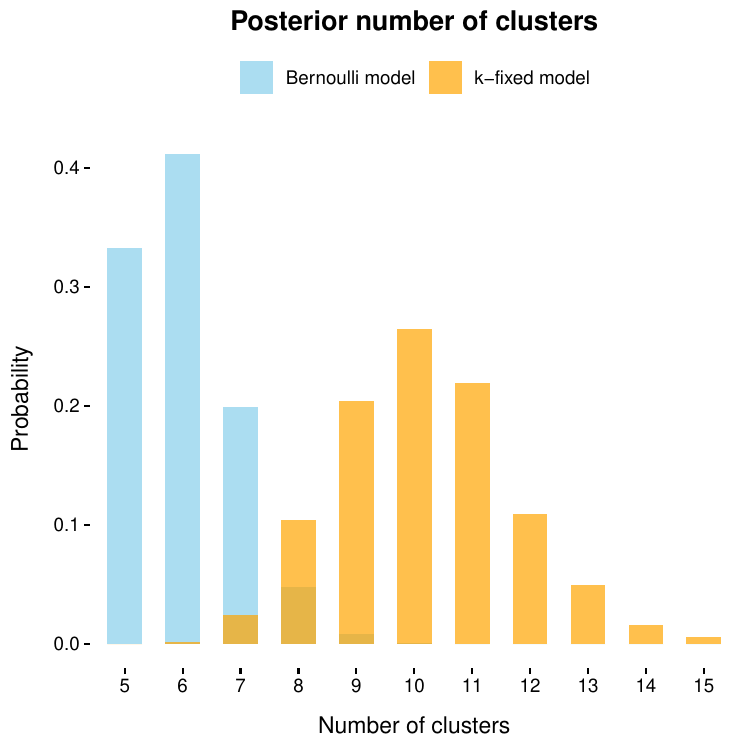}    
    
    \caption{Posterior distribution of the number of clusters under model \eqref{eq:mixture_model} with binary traits (in skyblue) and \citet{Dunson2009} model (in orange). The true number of clusters in the generating mechanism is $5$.  
    }
    \label{fig:sim_k_fixed_num_clusters}
\end{figure}

\subsection{Synthetic network data}\label{sec:simulations_network}

This section evaluates the performance of the latent class model with unseen traits described in Section~\ref{sec:learning_clustering}, referred to as \emph{unknown-groups} model. Motivated by our focus on criminal network analysis, we generate synthetic datasets designed to capture a realistically complex structure. The analysis of criminal networks \citep{esbm_rigon, Lu2025} typically aims at detecting groups of affiliates who share similar connectivity patterns. These patterns are encoded in a weighted adjacency matrix, which records the frequency of interactions between pairs of individuals. An interaction is defined as the co-attendance of two affiliates at the same meeting. The weighted adjacency matrix thus summarizes the raw data, originally collected as multivariate binary observations, where each affiliate is associated with the list of meetings they attended.  Formally, if the dataset comprises $n$ affiliates and $k$ distinct meetings, the attendance information may be arranged in an $n \times k$ binary matrix $\bm{A}$, with entries $A_{i\ell}$ indicating whether individual $i$ attended meeting $\ell$. The corresponding observed weighted adjacency matrix, denoted by $\bm{W}$, is then obtained as $\bm{W} = \bm{A}\bm{A}^T$. Its generic entry $W_{ii'}$ represents the number of meetings jointly attended by affiliates $i$ and $i'$, that is, $W_{ii'} = \sum_{\ell=1}^k A_{i\ell} A_{i'\ell}$ for $i\neq i'$. Our framework induces a probabilistic structure on $\bm{W}$: each entry $W_{ii'}$ can be interpreted as the realization of a sum of Bernoulli random variables. To assess the flexibility of the proposed methodology in the analysis of criminal networks, we consider two simulated scenarios, each characterized by distinct properties of the resulting weighted adjacency matrix. These scenarios are constructed to produce structures akin to those investigated in \cite{Lu2025}, thereby closely reflecting patterns observed in real-world criminal networks. A detailed description of the two scenarios is provided below. 

For both scenarios, the analysis is carried out using the binary traits specification of the mixture model \eqref{eq:mixture_model}, referred to as the unknown-groups model, where $H(\cdot;\psi)$ is the beta distribution with parameters $(-\alpha, \alpha + \beta)$, with $\psi = (\alpha, \beta)$. Prior distributions for the model parameters $\lambda$, $-\alpha$, and $\alpha + \beta$ are specified as follows. The parameter $\lambda$, governing the Poisson-distributed total number of meetings, is assigned a gamma prior with parameters $(\alpha_\lambda, \beta_\lambda)$, as detailed in Section \ref{sec:hyper_fitting}. The hyperparameters $(\alpha_\lambda, \beta_\lambda)$ are chosen so that the prior expected value of $N$ equals $\hat{N} = 1.5k$, where $k$ denotes the observed number of meetings in the sample, and the prior variance of $N$ is set to $10\hat{N}$. For the parameters $(a, b) = (-\alpha, \alpha + \beta)$ of the beta distribution $H$, independent gamma priors are assumed with hyperparameters $(\alpha_a, \beta_a)$ and $(\alpha_b, \beta_b)$, respectively. Specifically, $(\alpha_a, \beta_a)$ are set to induce a prior mean of $0.2$ for $a$ with a large variance, while $(\alpha_b, \beta_b)$ are chosen so that the prior mean of $b$ is $10$, also with high variance. Additionally, the clustering structure is modeled using a Dirichlet process with concentration parameter equal to $\gamma = 1$. Posterior inference relies on 10,000 iterations of the \textsc{mcmc} algorithm described in Section~\ref{sec:clustering_estimation}, with the first 1,000 samples discarded as burn-in and a thinning interval of 2.

We benchmark our approach against the \emph{negative binomial mixture of \textsc{bb}s} introduced in \cite{Ghilotti2025}, which represents the natural competitor for estimating the number of unseen traits. This model, however, is specifically tailored to homogeneous exchangeable settings. 
It falls within the class described in Example~\ref{example:ex_bin_traits}, corresponding to the binary traits specification of model~\eqref{eq:partially_ex_traits} with prior given in equations~\eqref{eq:finite_CRV}–\eqref{eq:param_H}, under the special case $d=1$.
In this framework,  $H(\cdot;\psi)$ is the beta distribution with parameters $(-\alpha, \alpha + \beta)$. The beta parameters $(a, b) = (-\alpha, \alpha + \beta)$ are equipped with independent gamma priors with hyperparameters $(\alpha_a, \beta_a)$ and $(\alpha_b, \beta_b)$, respectively. The model parameter $\lambda$, governing the Poisson-distributed total number of meetings, is assigned a gamma prior with parameters $(\alpha_\lambda, \beta_\lambda)$, as detailed in Section \ref{sec:hyper_fitting}.
By construction, the negative binomial mixture of \textsc{bb}s enforces homogeneity, resulting in posterior expectations for $\bm{W}$ that are constant across entries. Consequently, the model cannot capture potential heterogeneity in connectivity patterns that may emerge from the observed adjacency matrix.
For all comparisons presented in this paper, hyperparameters are fixed as in \cite{Ghilotti2025}, with the prior variance for $N$ set to be comparable to that used in the proposed unknown-groups model for each application. Posterior inference relies on 100,000 iterations, with the first 10,000 samples discarded as burn-in and a thinning interval of 5.

\begin{figure}[tbp]
    \centering
    \begin{subfigure}[t]{0.45\linewidth}
        \centering
        \includegraphics[width=\linewidth]{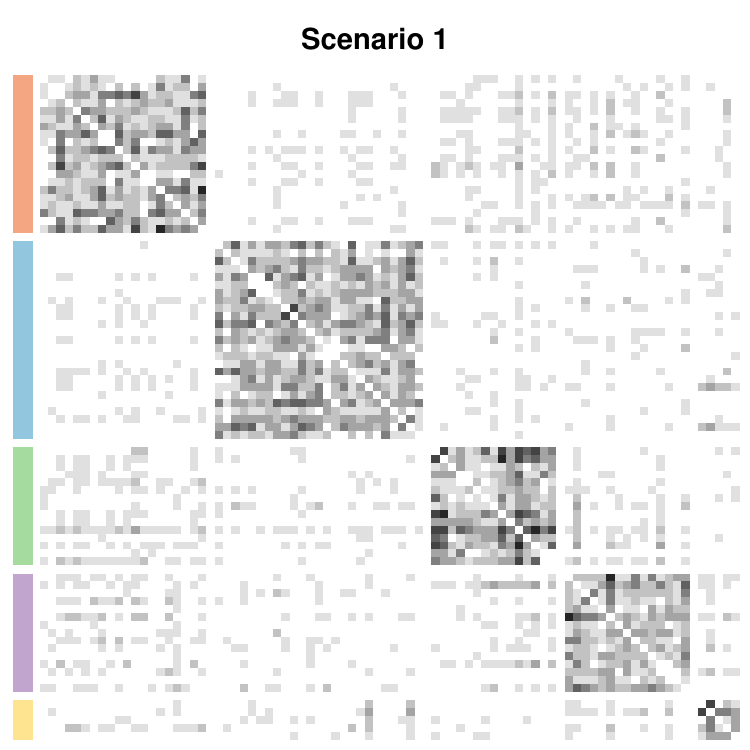}
        \caption{Simulated weighted adjacency matrix in Scenario~1. Side colors represent the true clustering.}
        \label{fig:true_adj_matrix_scenario_1}
    \end{subfigure}
    \hfill
    \begin{subfigure}[t]{0.45\linewidth}
        \centering
        \includegraphics[width=\linewidth]{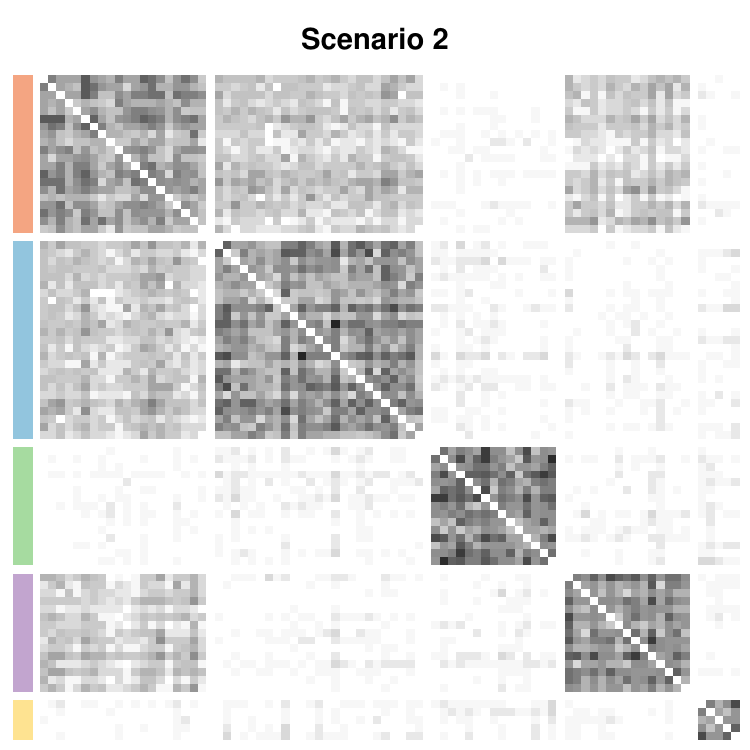}
        \caption{Simulated weighted adjacency matrix in Scenario~2. Side colors represent the true clustering.}
        \label{fig:true_adj_matrix_scenario_2}
    \end{subfigure}

    \vspace{1em}

    \begin{subfigure}[t]{0.45\linewidth}
        \centering
        \includegraphics[width=\linewidth]{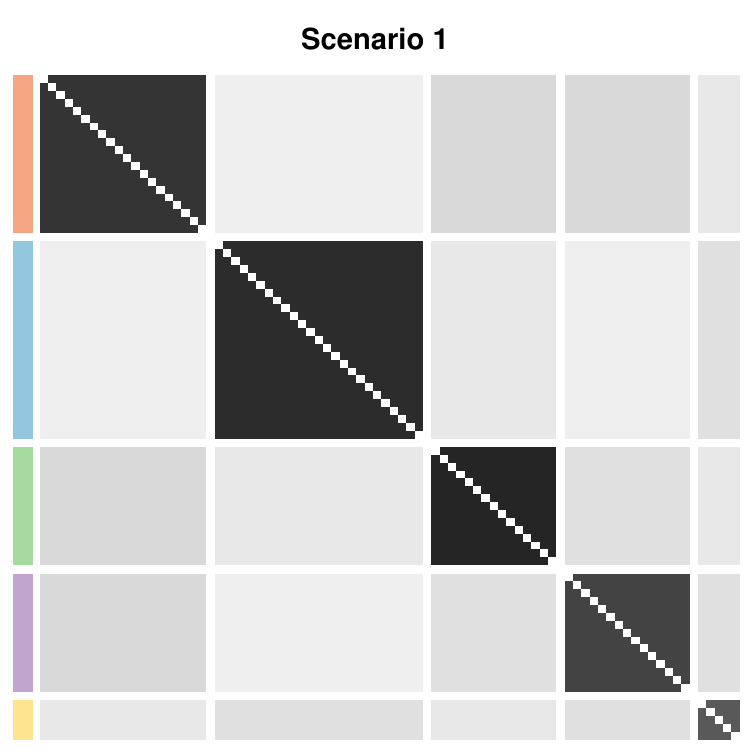}
        \caption{Posterior expectation of the weighted adjacency matrix in Scenario~1. Side colors represent the estimated clustering.}
        \label{fig:post_adj_matrix_scenario_1}
    \end{subfigure}
    \hfill
    \begin{subfigure}[t]{0.45\linewidth}
        \centering
        \includegraphics[width=\linewidth]{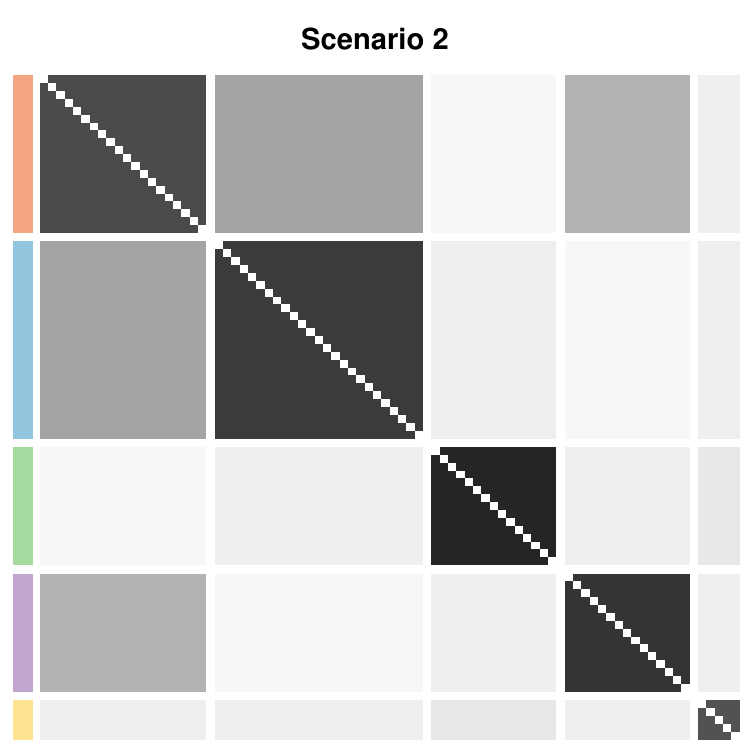}
        \caption{Posterior expectation of the weighted adjacency matrix in Scenario~2. Side colors represent the estimated clustering.}
        \label{fig:post_adj_matrix_scenario_2}
    \end{subfigure}

    \caption{Synthetic network data. Simulated and estimated weighted adjacency matrices under the two scenarios. 
    Top row: simulated adjacency matrices, with side colors denoting the true clustering. 
    Bottom row: posterior expectations of the adjacency matrices under the proposed model, with side colors denoting the estimated clustering. 
    In all panels, the color of each entry ranges from white to black as the number of co-attended meetings increases.}
    \label{fig:adj_matrix_comparison}
\end{figure}

\textbf{Scenario 1.} Data are generated with a total of $N = 500$ meetings and $d = 5$ criminal groups. Each group is associated with a core set of $15$ meetings, in which affiliates have a higher probability of being observed ($0.3$). Beyond these core meetings, each group is randomly assigned $300$, $125$, $50$, and $10$ additional meetings, with corresponding probabilities of $0.002$, $0.01$, $0.05$, and $0.3$. The total number of criminals is $n = 80$, partitioned into groups of size $n_1 = 20$, $n_2 = 25$, $n_3 = n_4 = 15$, and $n_5 = 5$. The number of observed meetings in the sample is $k = 304$. Figure~\ref{fig:true_adj_matrix_scenario_1} shows a graphical representation of the resulting adjacency matrix.

Posterior inference under the proposed unknown-groups model demonstrates strong recovery of the underlying structure. The posterior distribution of the number of clusters (not shown) correctly concentrates on the true value $d = 5$, confirming the ability of the model to detect the latent group partition. Beyond validating the clustering estimates, Figure~\ref{fig:post_adj_matrix_scenario_1} shows the posterior expectation of the weighted adjacency matrix, which can be visually compared to its simulated counterpart in Figure~\ref{fig:true_adj_matrix_scenario_1}. The close alignment between these matrices demonstrates the model’s ability to recover the connectivity structure of the data. Notably, posterior estimates of the weighted adjacency matrix are straightforward to compute, leveraging the closed-form posterior expressions in Theorem~\ref{thm:posterior_fCRV} and their specializations for the binary traits model. 

We further compare our methodology to the negative binomial mixture of \textsc{bb}s introduced in \cite{Ghilotti2025}. Both models allow inference on the number of unseen meetings $N^\prime$, which in the framework of Theorem~\ref{thm:posterior_fCRV} corresponds to the number of atoms $N^\prime$ in the measures $\mu_q^\prime$. Figure~\ref{fig:M_prime_scenario_1} reports the posterior distributions of $N^\prime$ for the two models. Our proposed model provides an accurate estimate, with uncertainty appropriately quantified. In contrast, the negative binomial mixture of \textsc{bb}s substantially overestimates $N^\prime$, failing to address the unseen features problem in this heterogeneous setting.

Model comparison based on the \textsc{waic} supports the same conclusion. The proposed unknown-groups model yields a \textsc{waic} of $68546$, while the negative binomial mixture of \textsc{bb}s yields a much larger value of $923830$.

\begin{figure}[tbp]
    \centering
    
    \begin{subfigure}[t]{0.45\linewidth}
        \centering
        \includegraphics[width=\linewidth]{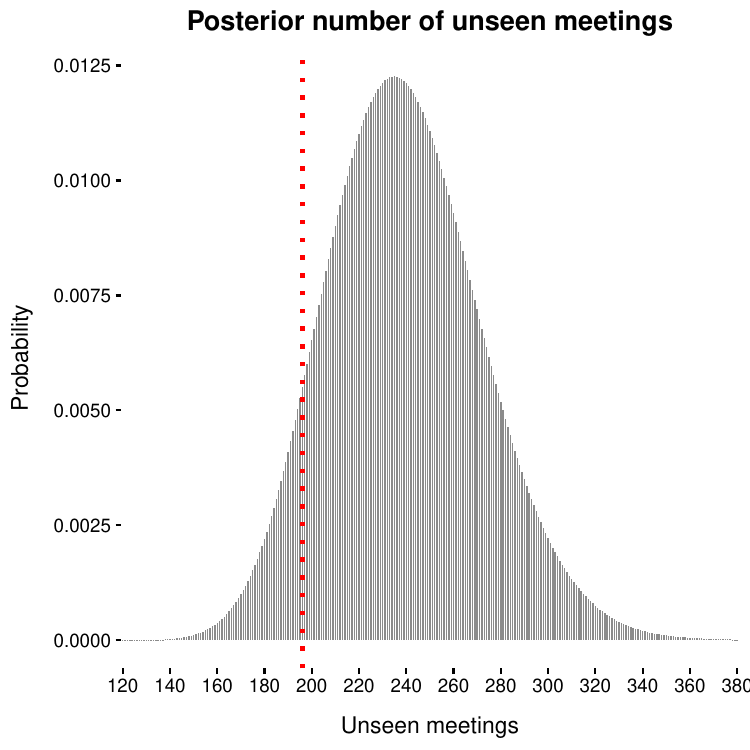}
        \caption{Unseen meetings under the unknown-groups model.}
        \label{fig:unseen_scen1_mix}
    \end{subfigure}%
    \hfill
    \begin{subfigure}[t]{0.45\linewidth}
        \centering
        \includegraphics[width=\linewidth]{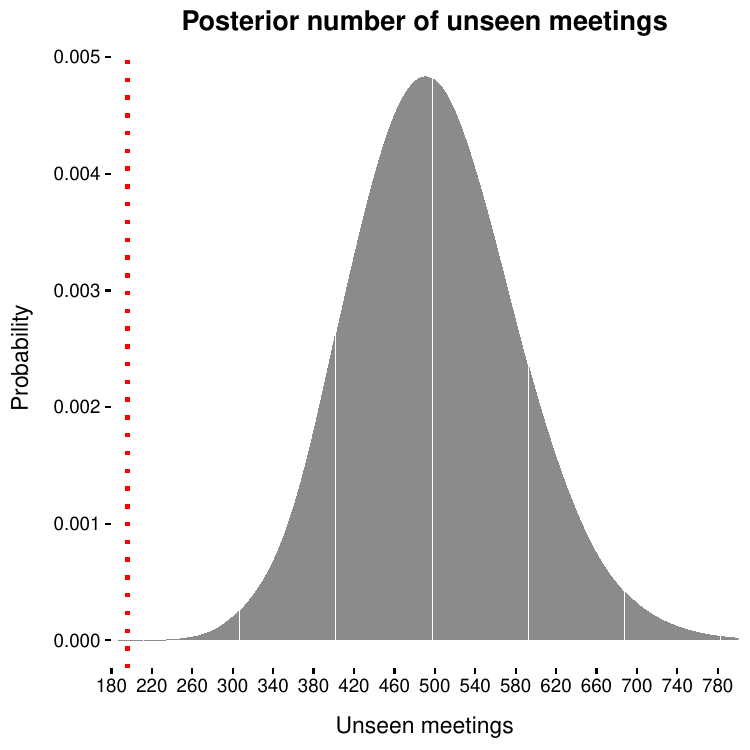}
        \caption{Unseen meetings under the negative binomial mixture of \textsc{bb}s.}
        \label{fig:unseen_scen1_negbinbb}
    \end{subfigure}

    \caption{Synthetic network data, Scenario 1. Posterior distribution of the number of unseen meetings under (a) the unknown-groups model and (b) the negative binomial mixture of \textsc{bb}s in \cite{Ghilotti2025}. The vertical red dotted lines denote the true number of unseen meetings.}
    \label{fig:M_prime_scenario_1}
\end{figure}

\textbf{Scenario 2.} This scenario also considers $N = 500$ meetings and $d = 5$ criminal groups, but introduces more heterogeneous and asymmetric connectivity patterns. Each group is associated with $250$, $150$, $50$, and $50$ meetings, with observation probabilities of $0.002$, $0.01$, $0.05$, and $0.4$, respectively. To induce additional overlap, 20 meetings have the highest probability ($0.4$) of being attended by criminals from both groups 1 and 2, while another distinct set of 20 meetings has the highest probability ($0.4$) for both groups 1 and 4. In general, all high-probability meetings do not overlap between different group pairs. As in Scenario 1, the total number of criminals is $n = 80$, partitioned into groups of size $n_1 = 20$, $n_2 = 25$, $n_3 = n_4 = 15$, and $n_5 = 5$. The number of observed meetings is $k = 340$. Figure~\ref{fig:true_adj_matrix_scenario_2} illustrates the resulting adjacency matrix.

As in Scenario 1, the proposed unknown-groups model correctly identifies the true number of clusters $d = 5$. The corresponding clustering estimates are reported in Figure~\ref{fig:post_adj_matrix_scenario_2}, alongside the posterior expectation of the weighted adjacency matrix. A visual comparison with the simulated matrix in Figure~\ref{fig:true_adj_matrix_scenario_2} confirms that the method accurately recovers the heterogeneous connectivity patterns.

We next examine inference on the number of unseen meetings $N^\prime$. Figure~\ref{fig:M_prime_scenario_2} presents the posterior distribution of $N^\prime$ under the proposed unknown-groups model and the negative binomial mixture of \textsc{bb}s in \cite{Ghilotti2025}. Figure~\ref{fig:unseen_scen2_mix} shows that our unknown-groups model recovers the number of unseen meetings accurately, while also providing a coherent quantification of uncertainty.
At first glance, the competitor appears to perform well, as its posterior distribution is centered closer to the true value (red vertical line), albeit with substantially greater dispersion. However, this apparent accuracy is largely coincidental.
In fact, since the negative binomial mixture of \textsc{bb}s enforces homogeneity, it cannot accommodate the heterogeneity clearly visible in the simulated adjacency matrix (Figure~\ref{fig:true_adj_matrix_scenario_2}). 
For ease of comparison, Figure~\ref{fig:scenario2_post_adj_mat_comparison} juxtaposes the simulated adjacency matrix, the posterior expectation of the adjacency matrix under our unknown-groups model, and the posterior expectation of the adjacency matrix under the competitor. The figure highlights that while our model reproduces the complex connectivity structure, the competitor collapses to an unrealistic homogeneous pattern.

This conclusion is further supported by a model comparison based on the \textsc{waic}. The proposed unknown-groups model achieves a \textsc{waic} of $-19107$, while the negative binomial mixture of \textsc{bb}s yields a much larger value of $217142$. Such a difference provides strong evidence that the unknown-groups model offers a substantially better fit for these data.

\begin{figure}[tbp]
    \centering
    
    \begin{subfigure}[t]{0.45\linewidth}
        \centering
        \includegraphics[width=\linewidth]{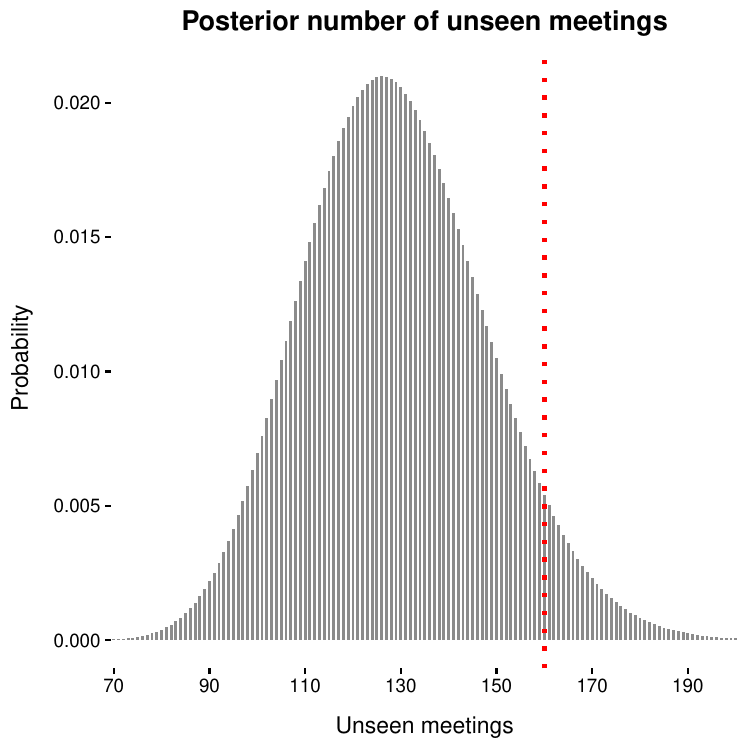}
       \caption{Unseen meetings under the unknown-groups model.}
        \label{fig:unseen_scen2_mix}
    \end{subfigure}%
    \hfill
    \begin{subfigure}[t]{0.45\linewidth}
        \centering
        \includegraphics[width=\linewidth]{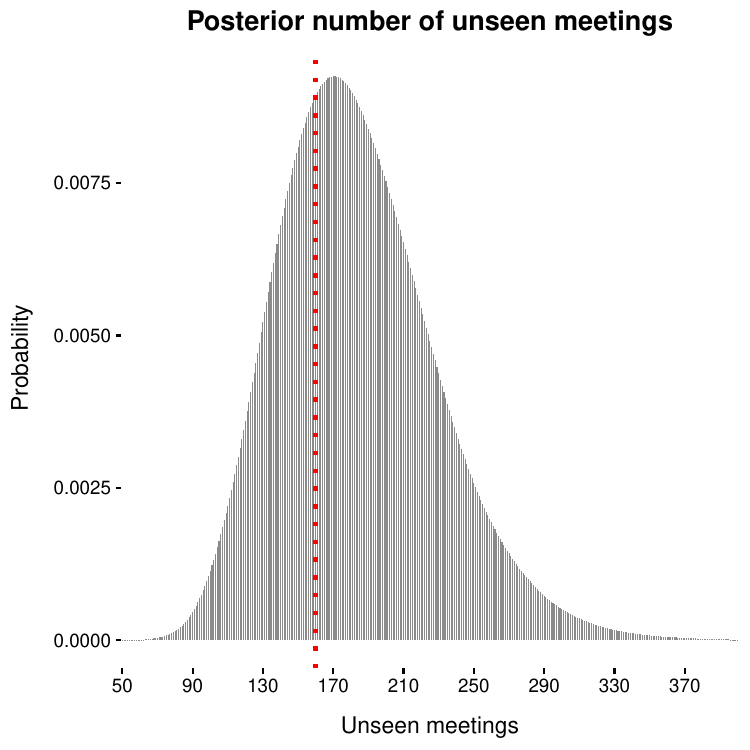}
       \caption{Unseen meetings under the negative binomial mixture of \textsc{bb}s.}
        \label{fig:unseen_scen2_negbinbb}
    \end{subfigure}

    \caption{Synthetic network data, Scenario 2. Posterior distribution of the number of unseen meetings under (a) the unknown-groups model and (b) the negative binomial mixture of \textsc{bb}s in \cite{Ghilotti2025}. The vertical red dotted lines denote the true number of unseen meetings.}
    \label{fig:M_prime_scenario_2}
\end{figure}

\begin{figure}[tbp]
    \centering
    
    \begin{subfigure}[t]{0.3\linewidth}
        \centering
        \includegraphics[width=\linewidth]{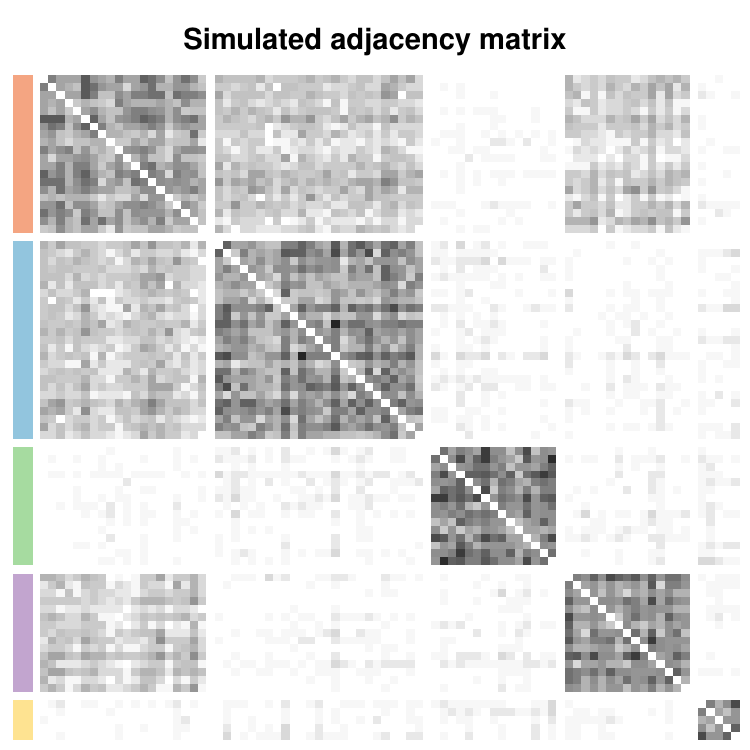}
        \caption{Simulated weighted adjacency matrix.}
        \label{fig:scenario2_simulated_adj_mat}
    \end{subfigure}%
    \hfill
    \begin{subfigure}[t]{0.3\linewidth}
        \centering
        \includegraphics[width=\linewidth]{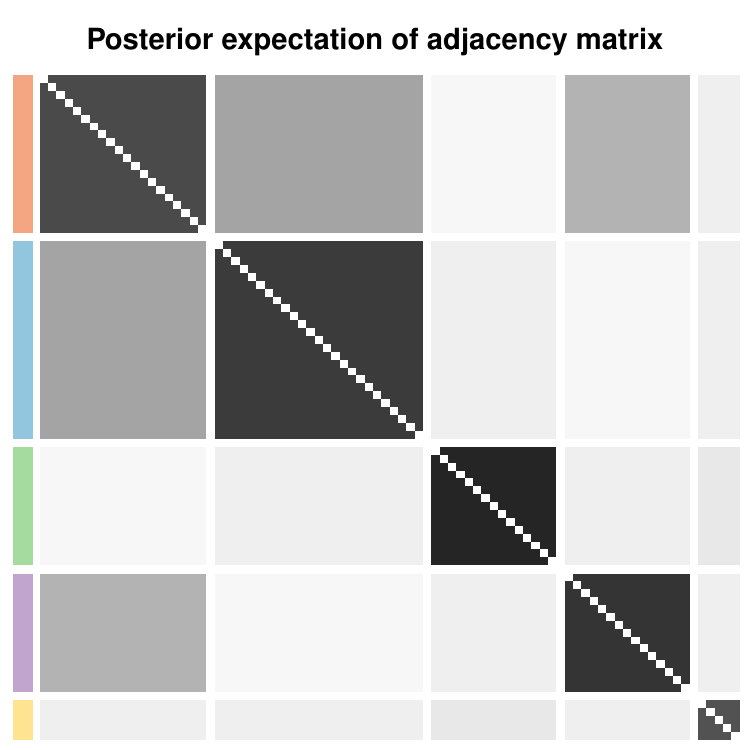}
        \caption{Posterior expectation of weighted adjacency matrix under the proposed unknown-groups model.}
        \label{fig:scenario2_post_exp_mixture}
    \end{subfigure}%
    \hfill
    \begin{subfigure}[t]{0.3\linewidth}
        \centering
        \includegraphics[width=\linewidth]{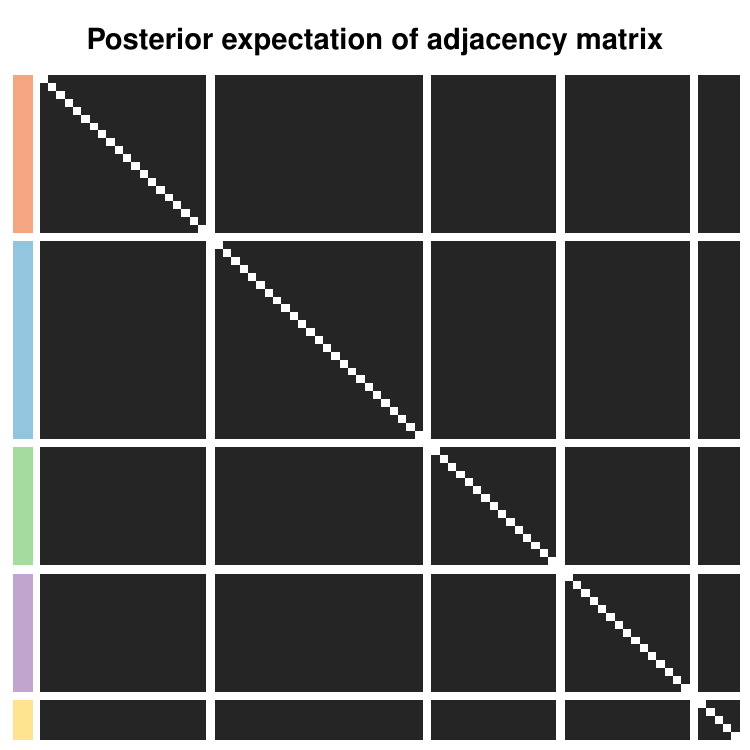}
        \caption{Posterior expectation of weighted adjacency matrix under the negative binomial mixture of \textsc{bb}s.}
        \label{fig:post_summary_ndrangheta_negbin_bb}
    \end{subfigure}

    \caption{Synthetic network data, Scenario 2. Simulated weighted adjacency matrix and its estimates under the competing models.}
    \label{fig:scenario2_post_adj_mat_comparison}
\end{figure}

\end{document}